\documentclass[
    aps,
    onecolumn,
    notitlepage,
    nofootinbib,
    superscriptaddress,
    longbibliography,
    floatfix
]{revtex4-2}

\usepackage{comment}
\usepackage[T1]{fontenc}
\usepackage[utf8]{inputenc}
\usepackage{lmodern}

\usepackage{amsmath,amssymb,amsthm,mathtools,bm}

\usepackage{graphicx}
\usepackage{booktabs}
\usepackage{array}[=2016-10-06]
\usepackage{tabularx}
\usepackage{multirow}
\usepackage{times}

\usepackage{enumitem}
\usepackage{microtype}
\usepackage{xcolor}
\usepackage{needspace}

\setlist{itemsep=3pt,topsep=5pt}
\renewcommand{\arraystretch}{1.17}

\usepackage{physics}

\usepackage{tikz}
\usetikzlibrary{arrows.meta,positioning,calc}

\makeatletter
\newcommand{\appendixcontents}{%
  \section*{CONTENTS OF THE APPENDICES}%
  \begingroup
    \c@tocdepth=2\relax
    \hypersetup{linktoc=all}%
    \@starttoc{apc}%
  \endgroup
  \let\apc@addcontentsline\addcontentsline
  \def\addcontentsline##1##2##3{%
    \apc@addcontentsline{##1}{##2}{##3}%
    \begingroup
      \def\apc@ext{##1}\def\apc@toc{toc}%
      \def\apc@kind{##2}%
      \def\apc@section{section}\def\apc@subsection{subsection}%
      \def\apc@write{%
        \addtocontents{apc}{%
          \protect\contentsline{##2}{##3}{\thepage}{\@currentHref}%
          \protected@file@percent}}%
      \ifx\apc@ext\apc@toc
        \ifx\apc@kind\apc@section
          \apc@write
        \else
          \ifx\apc@kind\apc@subsection
            \apc@write
          \fi
        \fi
      \fi
    \endgroup}%
  \clearpage
}
\makeatother

\theoremstyle{plain}
\newtheorem{theorem}{Theorem}[section]
\newtheorem{proposition}[theorem]{Proposition}
\newtheorem{lemma}[theorem]{Lemma}
\newtheorem{corollary}[theorem]{Corollary}

\theoremstyle{definition}
\newtheorem{definition}[theorem]{Definition}

\theoremstyle{remark}

\numberwithin{equation}{section}

\DeclareMathOperator{\diag}{diag}
\DeclareMathOperator{\Ran}{ran}
\DeclareMathOperator{\Span}{span}
\DeclareMathOperator{\Spec}{spec}
\DeclareMathOperator{\Sym}{Sym}

\renewcommand{\dd}{\mathrm d}
\newcommand{\e}{\mathrm e}

\renewcommand{\ip}[2]{\langle #1,#2\rangle}

\newcommand{\GF}{\mathcal G_{\mathrm F}}
\newcommand{\GFs}{\mathcal G_{\mathrm F}^{\sigma}}
\newcommand{\Sl}{\mathcal S}
\newcommand{\Gzero}{\mathcal G_0}
\newcommand{\Coh}{\mathcal C}
\newcommand{\GB}{\mathcal G}

\renewcommand{\epsilon}{\varepsilon}   

\newcommand{\aqa}{$\langle aQa ^L\rangle $ Applied Quantum Algorithms, Universiteit Leiden, Leiden, Netherlands}
\newcommand{\liacs}{LIACS, Universiteit Leiden, Leiden, Netherlands}
\newcommand{\epflip}{Institute of Physics, Ecole Polytechnique F\'{e}d\'{e}rale de Lausanne (EPFL), CH-1015 Lausanne, Switzerland}
\newcommand{\epflcqse}{Centre for Quantum Science and Engineering, Ecole Polytechnique F\'{e}d\'{e}rale de Lausanne (EPFL), CH-1015 Lausanne, Switzerland}
\newcommand{\fub}{Dahlem Center for Complex Quantum Systems, Freie Universität Berlin, 14195 Berlin, Germany}
\newcommand{\hhi}{Fraunhofer Heinrich Hertz Institute, 10587 Berlin, Germany}
\newcommand{\hzb}{Helmholtz-Zentrum Berlin f{\"u}r Materialien und Energie, 14109 Berlin, Germany}

\usepackage[pdfpagelabels]{hyperref}

\hypersetup{
    colorlinks=true,
    linkcolor = [rgb]{0.70,0.13,0.13},
    citecolor = [rgb]{0.13,0.55,0.13},
    urlcolor  = [rgb]{0.25,0.41,0.88},
    bookmarksnumbered=true,
    pdftitle={The geometry of robust Gaussianity tests},
    pdfsubject={Geometric stability, non-tolerant testing, and tolerant testing of pure Gaussian states},
    pdfkeywords={Gaussian states, fermions, bosons, property testing, geometric stability}
}

\usepackage[nameinlink,noabbrev]{cleveref}

\begin{document}

\title{Testing quantum Gaussianity with constant sample complexity} 

\author{Mahtab Yaghubi Rad}
\affiliation{\aqa}
\affiliation{\liacs}

\author{Ricard Puig}
\affiliation{\epflip}
\affiliation{\epflcqse}

\author{Saksham Hassanandani}
\affiliation{Department of Mathematics, University of Colorado Boulder, Boulder, CO 80309, USA}

\author{Luke~Coffman}
\affiliation{Department of Physics, Harvard University, Cambridge, MA 02138, USA}
\affiliation{School of Engineering and Applied Sciences, Harvard University, Cambridge, MA 02138, USA}

\author{Jens Eisert}
\affiliation{\fub}
\affiliation{\hhi}
\affiliation{\hzb}

\author{Carlos Bravo-Prieto}
\thanks{These authors contributed equally to the supervision of this work.}
\affiliation{\fub}

\author{Antonio Anna Mele}
\thanks{These authors contributed equally to the supervision of this work.}
\affiliation{\fub}

\begin{abstract}
Efficiently testing whether a quantum state possesses a given structure is both a fundamental and a practical task in quantum information. Among the most important structured families are Gaussian states, which underpin quantum optics and many-body physics while defining paradigmatic regimes of efficient classical simulation. Despite the central role of Gaussian states, optimal tests of Gaussianity have remained unknown. Here we give optimal and robust algorithms for testing fermionic and bosonic Gaussian states, with sample complexity independent of system size, particle number, and energy. Remarkably, the protocols implement experimentally feasible measurements associated with known exact symmetry characterizations of Gaussianity, using Bell sampling on two copies for fermions and passive interferometry with photon counting on three copies for bosons. We further obtain tolerant testers requiring joint measurements on a larger, but still constant, number of copies.
At the heart of our analysis is a unified approach---which we call the \emph{gradient-flow method}---that applies across all structured families considered here and underlies our optimal testers: representation theory turns approximate satisfaction of the defining symmetry into quantitative control of the gradient of the rejection probability, allowing a continuous descent from the input state to the target family whose length bounds their distance. 
\end{abstract}
 
\maketitle

\vspace*{-0.7cm}
\section{Introduction}\label{sec:introduction}

Quantum mechanics allows physical systems to exhibit properties which enable them to go beyond familiar classical methods in computation, sensing, communication, and simulation. Therefore, a central experimental task is to verify that an unknown state prepared in the laboratory has the property needed for the task at hand. Gaussianity in particular has a 
long and rich history as a guiding principle across physics and in recent years quantum information. In quantum field theory, Gaussian theories provide a tractable starting point for studying interacting systems through mean-field approximations and Wick’s theorem \cite{Stevenson1985}. 
In fermionic systems, they describe ground and thermal states of quadratic Hamiltonians and the quasiparticle states underlying much of quantum many-body physics~\cite{hackl,bravyi,terhal_divincenzo}. In bosonic systems, they encompass coherent, squeezed, thermal, and multimode Gaussian states, and are ubiquitous in quantum optics and continuous-variable quantum information~\cite{braunstein_vanloock,weedbrook_gaussian,eisert_plenio_cv}. 

At the same time, Gaussianity is closely tied to classical tractability: free-fermionic dynamics and bosonic Gaussian states and processes can be efficiently simulated classically, while non-Gaussian resources are needed to go beyond these tractable regimes and provide key ingredients for universal quantum computation in both settings~\cite{terhal_divincenzo,bravyi,bartlett_cv_simulation,bartlett_optical,weedbrook_gaussian,demelo,bell,lloyd_braunstein,albarelli_resource,mari_eisert}. With Gaussian-like properties being exhibited by stabilizer states~\cite{gross_hudsons_2006,bu_discrete_2023,bu_quantum_2023,bu_quantum_2025,bu_stabilizer_2023} and the broader group theoretic formulation of resource theories~\cite{somma_quantum_2005}, they have acted as a unifying folklore concept in quantum information theory~\cite{bauer_quadratic_2026}. It is therefore conceptually natural and practically highly relevant to ask whether Gaussianity itself can be efficiently certified from experimental data.

In technical terms, this question is most naturally formulated as a quantum property-testing problem~\cite{montanaro_dewolf,harrow_montanaro}: given copies of an unknown pure $n$-mode state vector $\ket\psi$, distinguish whether it is Gaussian or at least $\varepsilon$ far from every pure Gaussian state in trace distance with failure probability at most $\delta$. In contrast to tomography, which aims to learn a classical description of the unknown state~\cite{bittel_fermionic_tomography,chen_gaussian_tomography}, property testing only asks for this binary decision and can therefore potentially require far fewer copies.

Using the symmetries of Gaussian states, multiple few-copy procedures for testing have been identified. For fermions, pure Gaussian states satisfy a two-copy Majorana constraint~\cite{bravyi,demelo}, which can be measured 
in an experimentally friendly way using Bell sampling~\cite{haug,bell}. For bosons, passive interferometry gives analogous experimentally feasible tests using two copies for centered pure Gaussian states and three copies for general pure Gaussian states~\cite{girardi}. These tests are known to be \emph{exact}: they accept every Gaussian state with certainty, while every non-Gaussian state is detected with nonzero probability. Exactness, however, does not give a quantitative guarantee on this probability, which could in principle become arbitrarily small as the number of modes grows, even for states that remain a fixed distance from the Gaussian family. We call a test \emph{robust} if every state at distance at least $\varepsilon$ from the Gaussian family is detected with probability bounded below by a function of $\varepsilon$ alone. For fermions, the known bound for the Bell-sampling robust test scales with the number of modes~\cite{haug}, while for bosons the global guarantee requires additional energy assumptions~\cite{girardi}. This leads to our main question:
\begin{center}
    \emph{Can Gaussianity be tested with a constant number of samples that is independent of the number of modes?}
\end{center}

While believed to be possible, closing the gap between 
known constant lower bounds and mode-dependent upper bounds has remained a technical challenge, casting doubt on the true answer~\cite{bittel_fermionic_tomography,girardi,haug,bell}. In this 
work, we show that the answer to this central question is affirmative. Pure fermionic and bosonic Gaussianity can be tested using $\Theta(\varepsilon^{-2}\log(1/\delta))$ samples, independently of the number of modes and, for bosons, without any energy dependence. Our results cover pure fermionic Gaussian states and their particle-number-preserving subclass of Slater determinants, as well as general pure bosonic Gaussian states and the important special cases of centered Gaussian and coherent states. An overview of our framework is presented in Figure~\ref{F1}. 

\begin{figure}[t]
  \centering
  \includegraphics[width=0.97\textwidth]{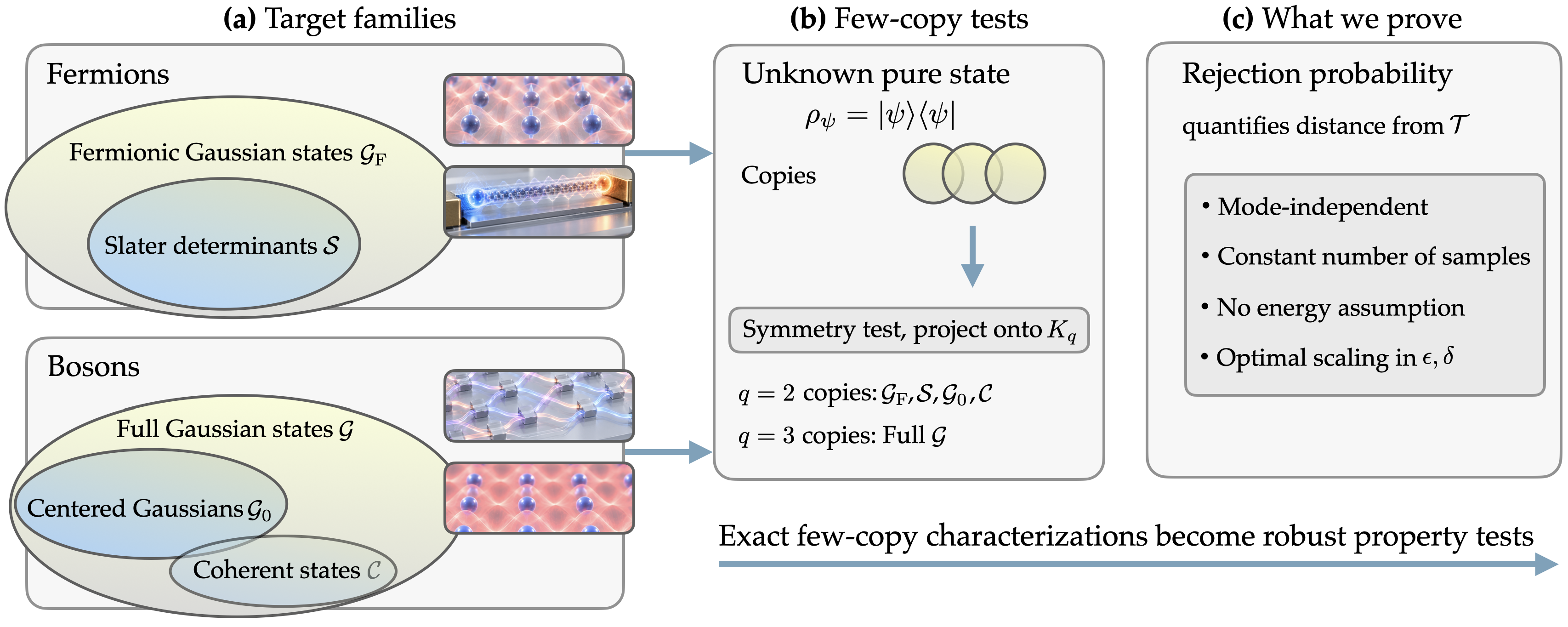}
  \caption{\textbf{Framework for robust few-copy Gaussianity testing.}
    \textbf{(a)} We consider five target families of pure states: fermionic Gaussian states $\mathcal G_F$ and their number-conserving subclass of Slater determinants $\mathcal S$, as well as full bosonic Gaussian states $\mathcal G$, centered Gaussian states $\mathcal G_0$, and coherent states $\mathcal C$. Also shown a representations of Majorana fermions and cold fermionic atoms in optical lattices as typical physical architectures for fermions, and linear optical networks and
    cold bosonic in optical lattices as typical physical platforms for bosons. \textbf{(b)} Given identical copies of an unknown state vector $\rho_\psi=\ket{\psi}\!\bra{\psi}$, the tests project onto a symmetry-defined acceptance subspace $K_q$. Two copies suffice for $\mathcal G_F$, $\mathcal S$, $\mathcal G_0$, and $\mathcal C$, whereas testing full bosonic Gaussianity requires three copies.
  \textbf{(c)} We prove that the resulting rejection probability robustly controls the distance from the respective target family, with constants independent of the number of modes. The corresponding tests have constant sample complexity, require no energy assumption in the bosonic setting, and achieve optimal dependence on the accuracy and confidence parameters $\varepsilon$ and $\delta$. Thus, exact few-copy characterizations of Gaussian states give rise to robust property tests.}
  \label{F1}
\end{figure}

The resulting protocols are the same experimentally friendly few-copy measurements suggested by the exact characterizations of the target families. Thus, our main contribution is to prove that these measurements are indeed robust. 

At a high level, the proof combines representation theory with the geometry of the rejection probability at a local and global level. Our main task is to turn exact few-copy characterizations into robust tests. Taking several identical copies gives a useful simplification: membership in the nonlinear target family becomes a condition that the tensor power lies in a fixed linear accepted subspace. Rejection is the squared distance to this subspace in the multiple-copy space, while robustness concerns the input's distance to the target family in the one-copy space. We must therefore relate distances in two different spaces: the one-copy space and the Hilbert space of $q$ copies. 
Because the input tensor power and its projection are symmetric
under copy permutations, we can work within the symmetric $q$-copy subspace.
The linear description gives us the advantage of analyzing acceptance through an orthogonal projector. However, not every vector in this linear subspace represents  $q$-identical copies of a single state. The challenge is to show that a tensor power close to this subspace comes from an input close to the target family. 

We introduce a general framework
based on \emph{gradient-flow} to establish this connection. Representation-theoretic estimates provide the gradient bounds needed. When initial rejection is
sufficiently small, this bound guarantees that the normalized
input converges to an exactly accepted state and, importantly, bounds the
total length of its path in the one-copy space. This length is
controlled by the initial rejection, showing that an input with
small rejection travels only a short distance before reaching
a target. Separately, expansions around closest targets
determine how sensitively the tests detect small deviations.
Combining these local estimates with the global bounds gives
the higher-copy tolerant-testing guarantees in the stated
small-distance regime. 

Our geometric contribution is a general mechanism for turning exact few-copy characterizations into robust tests. The framework can extend to other structured families of quantum states whenever analogous exact characterizations and gradient bounds can be established.

\begin{figure}[t]
 \centering
\includegraphics[width=0.97\textwidth]{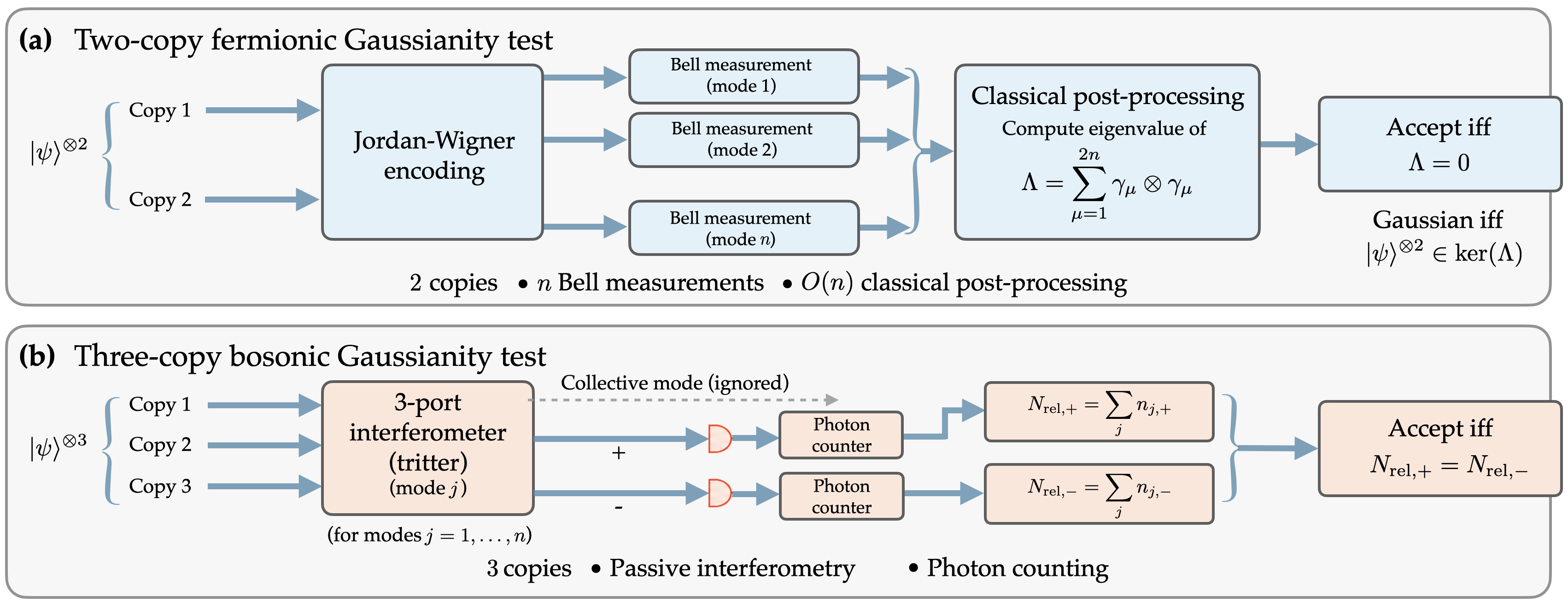}
\caption{
\textbf{Operational implementations of the few-copy Gaussianity tests.}
\textbf{(a)} Fermionic Gaussianity is tested on two copies by projecting onto the kernel of
$\Lambda=\sum_{\mu=1}^{2n}\gamma_\mu\otimes\gamma_\mu$.
Under a Jordan--Wigner encoding, this can be implemented by Bell measurements on corresponding orbital pairs, followed by linear-time classical post-processing.
\textbf{(b)} General bosonic Gaussianity is tested on three copies.
For each physical mode, a passive three-port interferometer separates the collective copy direction from two relative directions.
After photon counting on the two relative output families, the test accepts iff their total photon numbers satisfy
$N_{\mathrm{rel},+}=N_{\mathrm{rel},-}$.
The collective outputs are ignored.
}
\label{fig:protocols}
\end{figure}

\section{Main results and techniques}\label{sec:main-results-and-techniques}

We now present our main results for testing pure fermionic and bosonic Gaussian states. For an unknown pure state $\ket\psi$, let $d_{\mathcal G}(\psi)$ denote its trace distance from the corresponding pure Gaussian family $\mathcal G$. We begin with the non-tolerant problem: distinguishing an exactly Gaussian state from a state at distance at least $\varepsilon$ from every Gaussian state.

\begin{theorem}[Non-tolerant testing of Gaussian states]
\label{thm:main-nontolerant} Let $\mathcal G$ denote either the family of pure fermionic Gaussian states or the family of pure bosonic Gaussian states. There exists a test that accepts every $\ket\psi\in\mathcal G$ with certainty and rejects every pure state $\ket\psi$ satisfying $d_{\mathcal G}(\psi)\ge\varepsilon$ with probability at least $1-\delta$, using $O(\varepsilon^{-2}\log(1/\delta))$ copies, independently of the number of modes. The test consists of repeated measurements on independent blocks of two copies for fermions and three copies for bosons.
\end{theorem}

\noindent The tests admit simple implementations through Bell sampling for fermions and passive interferometry with photon counting for bosons, as shown in Figure~\ref{fig:protocols}. Our main contribution is to establish their global robustness: every pure input at distance at least $\varepsilon$ from the target family is rejected in each round with probability $\Omega(\varepsilon^2)$, uniformly in the number of modes. No parity promise is needed for fermions, and no energy or moment assumption is needed for bosons. Matching lower bounds show that the dependence on $\varepsilon$ and $\delta$ is optimal in the worst case, even when arbitrary collective measurements are allowed. The same sample complexity and perfect completeness hold for Slater determinants, centered (zero-mean) bosonic Gaussian states, and coherent states, with two-copy measurements sufficient in each case.  

To allow for small preparation errors, we next consider \emph{tolerant testing}: accepting states within distance $\varepsilon_A$ of the Gaussian family while rejecting states at distance at least $\varepsilon_B>\varepsilon_A$. Our second result retains mode-independent sample complexity when the near-Gaussian threshold is sufficiently small compared with the relative gap between the two thresholds.

\begin{theorem}[Tolerant testing of Gaussian states]
  \label{thm:main} 
  Let $\mathcal G$ denote either the family of pure fermionic Gaussian states or the family of pure bosonic Gaussian states. Let $0\le\varepsilon_A<\varepsilon_B\le1$ and set gap $0<\eta\le1$ such that $\varepsilon_B>(1+\eta)\varepsilon_A$ and $\varepsilon_A\le c\eta$ for a universal constant $c>0$. Then there exists a test that, given copies of an unknown pure state $\ket\psi$, distinguishes $d_{\mathcal G}(\psi)\le\varepsilon_A$ from $d_{\mathcal G}(\psi)\ge\varepsilon_B$ with success probability at least $1-\delta$, using $O((\varepsilon_B-\varepsilon_A)^{-2}\log(1/\delta))$ copies, independently of the number of modes.
\end{theorem}

\noindent The conditions on $\varepsilon_A, \varepsilon_B$ in terms of the gap $\eta$ and constant $c$ are to ensure the two regimes cannot become arbitrarily close to each other and hence make a tester's complexity diverge. As in the non-tolerant setting, no bosonic energy or moment assumption is required. The additional tolerance comes at the cost of a more demanding measurement: our construction acts jointly on $q=O(1/\eta)$ copies per round, rather than the fixed two- or three-copy blocks used above. Thus, resolving smaller relative gaps requires larger collective measurements in our protocol, while the total copy bound remains independent of the number of modes. The same tolerant guarantee extends to Slater determinants, centered bosonic Gaussian states, and coherent states.

\subsection{Testers and implementation}

Before introducing the testers formally, we explain the intuition
behind their construction, which will also guide the proofs of
robustness and tolerance. The central idea is to recognize the
target states through symmetries of several identical copies.

We consider transformations that mix corresponding physical modes
across the copies, using the same mixing coefficients for every
physical mode. We call this operation copy mixing.
A familiar classical example is that the product
of two identical centered Gaussian densities depends only on
$x^2+y^2$, so rotating the two variables leaves it unchanged.
Maxwell's characterization~\cite{maxwell2003illustrations} gives the converse: a rotationally
invariant product of two identical continuous probability
densities must be centered Gaussian. Related quantum
characterizations identify Gaussian states through invariance
under suitable transformations of their copies.
Intuitively, any centered pure bosonic Gaussian state is obtained
from vacuum using a Hamiltonian quadratic in the ladder operators
$a_i,a_i^\dagger$; the sum of two identical copies of this
Hamiltonian, acting on modes $(a,b)$, is invariant under real
rotations mixing the copies.

Each test accepts a linear subspace of the multiple-copy
Hilbert space. Its defining property is that a tensor power
belongs to this subspace precisely when the input belongs
to the target family. The acceptance conditions are different between target families.
We describe them together with their measurements below;
the fermionic and pure bosonic Gaussian protocols are
illustrated in Figure~\ref{fig:protocols}.

In optics, a unitary that rotates two copies is a beam splitter
interaction whose Hermitian generator is given by
\begin{align}
 \Lambda_0
 &:=i\sum_i
   \left(a_i^\dagger b_i-a_i b_i^\dagger\right),
 \qquad
 \Lambda_0\ket\psi^{\otimes2}=0
 \quad\Longleftrightarrow\quad
 \ket\psi\in\Gzero.
\end{align}
This gives a necessary and sufficient condition on a pure input
being centered Gaussian~\cite{girardi}.
An analogous two-copy characterization holds in the fermionic
setting and is called the matchgate
identity~\cite{bravyi,demelo}.
Let $\gamma_1,\ldots,\gamma_{2n}$ denote the Majorana operators
of an $n$-mode fermionic system. Then
\begin{align}
 \Lambda_{\mathrm F}
 :=\sum_{\mu=1}^{2n}\gamma_\mu\otimes\gamma_\mu,
 \label{eq:main-majorana}\quad\hspace{1.4cm}
 \Lambda_{\mathrm F}\ket\psi^{\otimes2}=0
 \quad\Longleftrightarrow\quad
 \ket\psi\in\GF.
\end{align}
This immediately defines a perfectly complete two-copy test:
measure $\Lambda_{\mathrm F}$ and accept whenever the outcome
is zero. The criterion applies to Gaussian states of either
parity (Lemma~\ref{lem:fermions-exact}).
Related copy-symmetry conditions give different acceptance
projectors for Slater determinants, centered bosonic Gaussian
states, and coherent states.

For fermions, this tester can be implemented efficiently using
Bell sampling~\cite{haug,bell}.
Indeed, the operators $\gamma_\mu\otimes\gamma_\mu$ mutually
commute and can therefore be measured simultaneously.
Under the Jordan--Wigner representation, Bell measurements
on corresponding qubit pairs of the two copies determine
their eigenvalues. The eigenvalue of $\Lambda_{\mathrm F}$
can then be reconstructed from the $n$ Bell outcomes using
$O(n)$ classical post-processing.
Hence, one round of the test requires only two copies,
$n$ parallel Bell measurements, and linear-time classical
processing. In particular, no covariance-matrix estimation
or knowledge of the parity sector is required.
Moreover, closely related Bell-sampling measurements of
fermionic non-Gaussianity have already been implemented on
quantum hardware~\cite{haug}, providing an experimental
realization of the main measurement primitive required here.
For Slater determinants, the pure fermionic Gaussian states
with definite particle number, the test projects onto the
corresponding copy-spin singlet subspace~\cite{oszmaniec_kus},
and our guarantee holds without assuming that the input itself
has a fixed particle number.

For centered bosonic Gaussian states, a fixed passive
interferometer diagonalizes this beam splitter generator
as a total photon-number difference, so the test reduces
to passive interferometry and photon counting.
Coherent states admit an even simpler two-copy condition.
A balanced beam splitter maps two identical coherent states
to a collective coherent output and vacuum in the difference
output. The test therefore accepts precisely when no photons
are detected in the difference port, as in coherent-state
comparison protocols~\cite{andersson_coherent_comparison}.

For general pure bosonic Gaussian states, we additionally
need to take care of the mean. A closely analogous
characterization holds for general pure bosonic Gaussian
states, now using three copies~\cite{girardi}.
Let $a_{j,r}$ denote the annihilation operator of physical
mode $j$ in copy $r\in\{1,2,3\}$, and define the quadratic
operator
\begin{align}
 \Lambda_{\mathrm B}:=\frac{i}{\sqrt3}\sum_{j=1}^{n}
 \Bigl[(a_{j,1}^{\dagger}a_{j,2}-a_{j,2}^{\dagger}a_{j,1})+(a_{j,2}^{\dagger}a_{j,3}-a_{j,3}^{\dagger}a_{j,2})+(a_{j,3}^{\dagger}a_{j,1}-a_{j,1}^{\dagger}a_{j,3})\Bigr].
 \label{eq:main-bosonic-generator}
\end{align}
For a pure input,
\begin{align*}
 \Lambda_{\mathrm B}\ket\psi^{\otimes3}=0
 &\quad\Longleftrightarrow\quad
 \ket\psi\in\GB.
\end{align*}
The measurement again has a simple implementation.
For each physical mode, the three corresponding copy modes
are mixed by the same passive three-port interferometer,
or \emph{tritter}. This transformation separates the
collective copy direction from two relative directions.
In particular, an equal displacement
$(\alpha_j,\alpha_j,\alpha_j)$ of the three copies is mapped
to $(\sqrt3\alpha_j,0,0)$, so the common displacement appears
only in the collective output.
In the tritter basis, $\Lambda_{\mathrm B}$ is diagonal and,
up to the convention for labeling the two relative outputs,
becomes their total photon-number difference,
$\Lambda_{\mathrm B}=N_{\mathrm{rel},+}-N_{\mathrm{rel},-}$.
The required optical ingredients are also experimentally
established: passive multiport transformations can be
decomposed into beam splitters and phase
shifters~\cite{Reck1994Unitary,Clements2016Optimal},
integrated three-port interferometers have been demonstrated
experimentally~\cite{spagnolo_tritter}, and high-efficiency
photon-number-resolving detection is
available~\cite{Lita2008Counting}.

\subsection{Robustness and tolerant bounds}
With the testers in place, we are now ready to outline the proof ideas of their robustness and tolerance. There are two related yet distinct proof techniques utilized which combine representation theory with geometry at a local and global level. The first is what we refer to as the \textit{gradient-flow method} and the second is \textit{closest-target geometry method}. 
For robustness, we show that the testers rejection probability $R(\psi)$ obeys $R(\psi)\ge \Omega(d_{\mathcal{G}}(\psi)^{2})$. Crucially, the constants are independent of the number of modes and, in the bosonic setting, without any energy or moment assumption. Repeating the corresponding few-copy measurement on independent blocks therefore gives the $O(\varepsilon^{-2}\log(1/\delta))$ sample complexity of our non-tolerant tests. For tolerance, we extend the rotational-invariance test to $q$ copies, sharpening its sensitivity to distance from the Gaussian family so that larger $q$ allows us to resolve smaller relative gaps between the acceptance and rejection thresholds.

In order to proceed, let $P_{q}$ denote the acceptance projector of a $q$-copy test such that $R_{q}(\psi)=1-\Tr(P_{q}\psi^{\otimes q})$. We first explain the argument for the two-copy tests.
On three identical copies, place the acceptance projector
on each of the pairs $12$, $13$, and $23$, leaving the remaining
copy untouched, and consider
\begin{align}
 T&=\frac13\left(P_{12}+P_{13}+P_{23}\right).
 \label{eq:main-overlap}
\end{align}
Every pair has the same acceptance probability, so the
expectation of $T$ is still $1-R_2(\psi)$.
Averaging allows us to study how the three accepted
subspaces fit together. The overlap between the tested pairs is the crucial ingredient.
When we compare their accepted components, taking overlaps
with the input on the other copies leaves a vector on their
shared copy. This is the one-copy vector that determines
how rejection changes when the input varies.
Consequently, information about the overlapping projectors
can be converted into information about the rejection gradient. 

Representation theory supplies a sufficiently strong spectral
bound for $T$, separating the common accepted space from the
rest of the spectrum. Our overlapping-projection theorem
converts this estimate into
\begin{align}
 \|\nabla R_2(\psi)\|_2^2
 &\ge16R_2(\psi)\bigl(\kappa-R_2(\psi)\bigr),
 \label{eq:main-gradient}
\end{align}
where $\kappa>0$ is independent of the number of physical modes
(Theorem~\ref{thm:overlap}).
Thus, small positive rejection cannot occur where the
rejection function is too flat. 

Starting with $R_2(\psi)<\kappa$, we follow the negative gradient
on the sphere of normalized one-copy inputs.
The gradient bound guarantees convergence to zero rejection
and bounds the total path length by $O(\sqrt{R_2(\psi)})$
(Theorem~\ref{thm:restoration}).
Exactness identifies the endpoint as a target, and the path
length bounds the input's distance from the target family:
\begin{align*}
 d_{\mathcal T}(\psi)
 &\le \text{path length}
 \lesssim\sqrt{R_2(\psi)}.
\end{align*}
This gives the desired robustness bound.
Inputs with rejection at least $\kappa$ already have a constant
rejection lower bound, so the guarantee holds at every distance.
The third copy is used only in the proof; the measured test
still acts on two copies.

\paragraph{Why the bound is independent of system size.}
For fermionic and centered bosonic Gaussian states,
copy rotations reduce the spectral analysis to representations
of a fixed rotation group. We prove a bound uniform over the
representations arising from different numbers of physical modes.
For bosons, the bound also holds in every sector of fixed total
photon number. Since copy mixing preserves these sectors,
the estimate extends to the full Fock space without
a photon-number cutoff.
The same overlapping-projector framework gives the bounds
for Slater determinants and coherent states.

\paragraph{Tolerant testing.}
Tolerance testing requires finer information than a global lower bound.
We separate the rejection probabilities of every input
within distance $\varepsilon_A$ from those of every input
beyond distance $\varepsilon_B$.
For the few-copy tests, different directions away from the
target family can have different rejection strengths,
which can obscure a small gap between these distances.

We therefore impose the appropriate copy symmetries jointly
on $q$ copies. For fermionic Gaussians, this means enforcing
the Majorana constraint on every pair; for bosons, it means
imposing the corresponding $q$-copy rotation symmetry.
We then expand around a closest target to understand their response.

\section{Discussion and conclusion}

\subsection{Related work}\label{subsec:related-work}

\paragraph{Fermionic Gaussianity testing.}
The two-copy test studied here is based on the Majorana characterization of pure fermionic Gaussian states~\cite{bravyi,demelo}. Other approaches characterize or test fermionic Gaussianity using fermionic convolution~\cite{lyu_bu,coffman}, correlation-matrix methods~\cite{bittel_fermionic_tomography}, and symmetry-adapted random purification~\cite{walterwitteveen}.  The closest results to ours analyze the same two-copy Majorana constraint through Bell sampling~\cite{haug,bell}. Bell sampling provides an experimentally simple implementation and yields a sample complexity $O(n\varepsilon^{-2}\log(1/\delta))$ for pure-state Gaussianity testing~\cite{haug}; moreover, the corresponding acceptance projector is optimal among perfectly complete two-copy tests for pure inputs of definite parity~\cite{bell}. Our result shows that the same measurement is in fact robust independently of the number of modes: its rejection probability is $\Omega(\varepsilon^2)$ for every pure state at distance at least $\varepsilon$ from the Gaussian family. This reduces the sample complexity to $O(\varepsilon^{-2}\log(1/\delta))$ without changing the measurement or imposing a parity promise.

Related works also introduce efficiently measurable notions of fermionic non-Gaussianity, including convolution-based quantities, violations of Wick's theorem and the matchgate identity, Bell-sampling monotones, and fermionic entropy~\cite{coffman,bell,leone_bittel}. These quantities provide useful witnesses and resource measures, but their efficient estimation does not by itself imply constant-sample testing in trace distance, since the corresponding witness gap at fixed distance from the Gaussian family may still decrease with the number of modes. The experimental relevance of Bell sampling has also been demonstrated on superconducting hardware, where two-copy Bell data were used to measure a purity-corrected fermionic non-Gaussianity witness~\cite{haug}.

\paragraph{Bosonic Gaussianity testing.}
For bosons, the relation between Gaussianity and multi-copy interference goes back to factorization properties under beam-splitter mixing~\cite{springer}. Quantitative stability has been studied through a quantum Darmois--Skitovich theorem, which relates approximate output factorization to proximity to Gaussian states under moment assumptions and in Hilbert--Schmidt distance~\cite{cuesta}. 
Early work studied certification of given bosonic Gaussian states with homodyne measurements~\cite{aolita_photonic_certification}. More recently, passive copy symmetries have been used to construct few-copy tests for pure Gaussian states: two copies suffice for centered Gaussian states, while three copies suffice when arbitrary displacements are allowed~\cite{girardi}. These tests are locally robust with constants independent of the number of modes, but the available global guarantees require additional assumptions, in particular energy control. Our results establish global trace-distance robustness for the associated invariant-space projectors for arbitrary pure inputs, without 
any energy or moment bound.

Several related proposals use multi-copy optical interference to witness or quantify bosonic non-Gaussianity~\cite{bu_li_bosonic,hahn_takagi_correlation,dai_zhang}. Coherent-state comparison via a beam splitter also provides an earlier instance of the difference-output vacuum test used here~\cite{andersson_coherent_comparison}. Experimentally, non-Gaussianity has been witnessed using photon counting and homodyne measurements~\cite{jezek_non_gaussian,jezek_squeezed}, and multi-photon interference in integrated tritters has been demonstrated~\cite{spagnolo_tritter}. These works establish closely related measurement primitives and non-Gaussianity witnesses, but not a uniform worst-case trace-distance guarantee to the pure Gaussian family.

\paragraph{Hardness of testing against mixed Gaussian states.}
The fact that our testing analysis focus on the set of pure Gaussian states, and not on the more general mixed-Gaussian class, is essential to get constant sample-complexity. In fact it has been shown that testing arbitrary mixed fermionic states against the full mixed Gaussian family requires exponentially many copies~\cite{bittel_fermionic_tomography}, with analogous bosonic lower bounds known under suitable energy and separation assumptions~\cite{girardi}.  

\paragraph{Testing other structured quantum states.}
Dimension-independent property testing has important precedents in quantum information. The two-copy product test relates rejection probability directly to distance from the set of product states, independently of the number and dimensions of the subsystems~\cite{harrow_montanaro}, and its quantitative behavior has subsequently been sharpened~\cite{soleimanifar_wright,beckey_jeronimo_wu}. For qubit stabilizer states, symmetries of stabilizer tensor powers give a six-copy test whose copy complexity is independent of the number of qubits~\cite{gross_nezami_walter}, providing a particularly close precedent for dimension-independent testing of a symmetry-defined state family. Convolution-based stabilizer tests provide a related approach~\cite{bu_stabilizer_2023}. Representation-theoretic projector characterizations apply more broadly to structured pure-state families, including Slater determinants~\cite{oszmaniec_kus}. These examples share the same general theme as the present work: a few-copy symmetry can characterize a structured family, but obtaining a property tester additionally requires quantitative control of how the acceptance probability changes with distance from that family.

\paragraph{Learning fermionic Gaussian and near-Gaussian states.}
Gaussian-state learning asks for substantially more information than property testing. For fermions, efficient learning has been developed for free-fermionic states using tomography, fermionic shadows, and matchgate-based measurements~\cite{aaronson_grewal,ogorman,zhao_fermionic_shadows,wan_matchgate_shadows}. Trace-distance bounds in terms of correlation matrices further improve tomography of pure and mixed fermionic Gaussian states~\cite{bittel_fermionic_tomography}. Efficient learning also extends beyond the Gaussian family: states prepared using a limited number of bounded-locality non-Gaussian fermionic gates can be learned using single-copy measurements~\cite{mele_herasymenko}. Closely related work studies the agnostic problem of finding a Slater determinant with nearly maximal overlap with an arbitrary fixed-particle-number input~\cite{paul_zhao_dai}.

\paragraph{Learning bosonic Gaussian and near-Gaussian states.}
For bosons, tomography guarantees based on first and second moments have been developed for Gaussian and structured non-Gaussian states~\cite{mele_cv_tomography,bosonic_trace_distance}. Adaptive protocols reduce the dependence on the energy scale to doubly logarithmic factors~\cite{energy_independent_tomography}, while recent work characterizes the limitations imposed by restricting to Gaussian measurements~\cite{chen_gaussian_measurements,gaussian_loglog}. The unrestricted sample complexity of fermionic and bosonic Gaussian tomography is now known to be $\Theta((n^2+\log(1/\delta))/\varepsilon^2)$ for both pure and mixed states, without an energy bound in the bosonic case~\cite{chen_gaussian_tomography}. Our result therefore gives a separation between testing and tomography in their dependence on the number of modes: reconstructing an unknown Gaussian state requires quadratically many copies in $n$, whereas testing whether an unknown pure state is Gaussian requires a number of copies independent of $n$ at fixed accuracy and confidence.

\subsection{Conclusion and open questions}
\label{sec:discussion}

In this work, we close the gap between exact few-copy characterizations of pure Gaussian states and robust Gaussianity testing. For both fermions and bosons, the known few-copy measurements already provide experimentally simple non-tolerant tests: Bell measurements on two copies for fermions, and passive interferometry with photon counting on two or three copies for bosons. We show that these tests are globally robust, with sample complexity independent of the number of modes and, for bosons, without any energy assumption. We also obtain tolerant testers, although these require collective measurements on a number of copies that grows as the two testing thresholds approach each other.

This leaves several natural questions. First, what is possible with single-copy measurements? The tests studied here exploit genuinely multi-copy observables. Can Gaussianity still be tested with mode-independent sample complexity using adaptive single-copy measurements, or is coherent access to several copies fundamentally necessary? More generally, it would be interesting to characterize the tradeoff between sample complexity and quantum memory, interpolating between single-copy protocols and unrestricted collective measurements. Closely related questions have recently been studied for stabilizer-state testing: single-copy protocols require a number of copies growing with the number of qubits~\cite{hinsche_helsen_singlecopy}, while bounded coherent quantum memory leads to explicit tradeoffs between memory and sample complexity~\cite{arunachalam_schatzki_memory}.

Another interesting question is whether the present guarantees can be extended beyond pure inputs, and can the tolerant regime be sharpened? Our tolerant result applies only in the small-threshold regime $\varepsilon_A\le c\eta$. Determining the optimal relation between the two thresholds, and understanding which analogous guarantees remain possible for mixed inputs, would give a more complete picture of Gaussianity testing.

Finally, the geometric proof techniques developed here rely on only a few properties of the target family and may therefore be of independent interest beyond Gaussianity testing, including other property-testing problems and in broader settings. We hope that the present work stimulates further research in this direction.

\subsection*{Acknowledgments}

\textbf{Concurrent work.}
Toward the completion of this manuscript, we became aware of concurrent and independent work by Iyer, O'Gorman, Parekh, Thompson, and Zhao~\cite{iyer2026optimal}, entitled ``Optimal testing of fermionic and bosonic Gaussian states.'' Their work also proves dimension-independent sample complexity for testing the classes of Gaussian states considered here. We have planned for both manuscripts to appear on the arXiv simultaneously. We thank the authors for the friendly communication.

\textbf{Support and funding.}
M. Y. R., R. P., and C. B.-P. thank the Centre for Quantum Technologies in Singapore for hosting Hackamonth 2026, where some of the ideas contributing to this work were first developed.
A. A. M. is grateful to Lennart Bittel, Lorenzo Leone, and Francesco Anna Mele for discussions in earlier projects related to this problem.

M. Y. R. is supported by the project Divide \& Quantum (Project No.~1389.20.241) of the research programme NWA-ORC, which is partly financed by the Dutch Research Council (NWO).
R. P. acknowledges support from the SNF Quantum Flagship Replacement Scheme (Grant No.~215933).
L. C. acknowledges support from the National Science Foundation Graduate Research Fellowship under Grant No.~2140743.
J. E., C. B.-P., and A. A. M. are supported by the BMFTR (Hybrid++, PasQuops, PraktiQOM, QuSol, QSolid, MuniQC-Atoms), the Munich Quantum Valley, Berlin Quantum, the Quantum Flagship (Millenion, PasQuans2), QuantERA, the European Research Council (DebuQC), the Clusters of Excellence (MATH+, ML4Q), and the DFG (CRC 183, SPP 2514, and BoLaCo).
A. A. M. further acknowledges support from a 2025 Google PhD Fellowship.

Any opinions, findings, conclusions, or recommendations expressed in this material are those of the authors and do not necessarily reflect the views of the National Science Foundation.

\textbf{AI statement.}
The authors conjectured the main statement underlying the fermionic Gaussianity test and used ChatGPT 5.5 Pro, together with substantial human intuition and mathematical guidance, to develop an initial proof. ChatGPT 5.6 Pro, ChatGPT 6.0 Pro Astra, and Claude were subsequently used as interactive research tools to strengthen these results and develop extensions to pure states far from the fermionic Gaussian family and, later, to the bosonic setting. These models played a valuable role in the development of the work. All final arguments were critically reviewed, verified, and rewritten by the authors, with particular attention to making the presentation and technical machinery accessible to the quantum information community.

\bibliography{gaussian_geometry_references,BigReferences70}

\onecolumngrid
\appendix

\newpage

\setcounter{figure}{0}
\setcounter{theorem}{0}

\begin{center}
\large{Supplementary Material for ``Testing quantum Gaussianity with constant sample complexity''
}
\end{center}
\counterwithin{figure}{section}

\appendixcontents

\newpage

\section{Reader's guide}\label{sec:guide}

The proof starts from a symmetry of several identical copies that holds exactly for the target states. The test checks this symmetry, so perfect acceptance characterizes the target family. Our task is to show that a state which almost always passes must be close to that family. Small rejection places the copies near the accepted subspace, but projection onto this subspace need not preserve their identical-copy form. We therefore need a way to reach a target by changing the original input. Figure~\ref{fig:guide-argument} shows how the proof steps fit together. Appendix~\ref{sec:math_tools} collects the preliminaries. The general argument is developed in Appendix~\ref{sec:preliminaries}, with its fermionic and bosonic ingredients established in Appendices~\ref{sec:fermions} and~\ref{sec:bosons}.

The left branch of Figure~\ref{fig:guide-argument} establishes global robustness. We first compare tests on overlapping groups of copies to find a direction that reduces rejection. For a two-copy test, we consider the three pairs within three copies. Rotations of the copy labels describe how their accepted spaces fit together. Representation theory gives uniform bounds on these overlaps. Since the tests share a copy, the overlap bounds control how rejection changes when the original input moves. Theorem~\ref{thm:overlap} makes this connection precise: sufficiently small positive rejection guarantees a direction of decrease whose strength is quantitatively controlled.

We then follow the direction of steepest decrease. Theorem~\ref{thm:restoration} shows that this gradient flow converges to a perfectly accepted state and bounds the entire path length in terms of the initial rejection. Exactness identifies the endpoint as a target. The input cannot be farther from this target than the length of the path. Thus small rejection forces a nearby target, giving a global robustness bound. Applying this argument gives the fermionic Gaussian and Slater bounds in Theorem~\ref{thm:fermions-calibration}, and the centered-Gaussian and coherent-state bounds in Theorem~\ref{thm:boson-base}. The constants are independent of the number of modes. The bosonic estimates also hold uniformly across photon numbers, so no energy cutoff is needed.

Unknown displacement requires one additional step. Mixing two identical inputs cancels their common displacement in the difference output. This output can be mixed, so we retain the joint pure state of both outputs. Gradient flow then shows that small rejection makes the difference output close to a pure centered Gaussian and nearly unentangled with the sum output. Combining this conclusion with the coherent-state bound gives a four-copy guarantee for the original input. A comparison of copy symmetries transfers that guarantee to three copies. Theorem~\ref{thm:boson-full} gives the global distance bound for the full bosonic Gaussian family.

The right branch of Figure~\ref{fig:guide-argument} gives a more precise estimate near the target family. We choose a closest target and study the deviation from it. The deviation has no component along an allowed first-order change within the target family. Otherwise, moving the target would make it closer. This restriction, formalized in Lemma~\ref{lem:normal-stationarity}, leaves only directions that depart from the family. We calculate how strongly the test detects these directions and control the remaining terms in the expansion. Lemma~\ref{lem:normal-bound} turns this local sensitivity into a rejection bound. This is a parallel argument: it sharpens the estimate near the family, while gradient flow supplies the guarantee at arbitrary distance.

The local and global bounds now give the testing guarantees. Perfect acceptance of target states bounds rejection from above for nearby inputs. Global robustness bounds it from below for farther inputs, and repetition gives ordinary testing as in Theorem~\ref{thm:main-nontolerant}. For tolerant testing with closely spaced thresholds, we impose the copy symmetries jointly on more copies. Their local response becomes sufficiently uniform to separate the near and far promises. Theorem~\ref{thm:normal-hierarchy} transfers the global bounds to these larger tests, ensuring that more distant inputs also remain detectable. Combining both estimates and repeating the measurement gives Theorem~\ref{thm:testing-optimal-gap}, within its stated small-distance regime.

Finally, explicit deviations show that the local sensitivity bounds are sharp. Examples with known target distances also provide a near input and a far input that have large overlap. Any successful tester must distinguish this pair. Theorem~\ref{thm:testing-lower-bound} turns their overlap into a lower bound on the number of copies required by any measurement strategy. Comparing it with our testing guarantees establishes optimal copy dependence in the families and distance regimes covered by these results.

\clearpage
\begin{figure}[!ht]
\centering
\begingroup
\hyphenpenalty=10000
\exhyphenpenalty=10000
\begin{tikzpicture}[
  box/.style={draw=#1!65!black, fill=#1!6, rounded corners=2pt,
    line width=0.55pt, align=center, font=\small,
    inner xsep=7pt, inner ysep=7pt, text width=6.6cm},
  box/.default=gray,
  wide/.style={box=gray, text width=14.65cm},
  arr/.style={-{Stealth[length=2mm]}, line width=0.7pt,
    draw=gray!80!black},
  localarr/.style={arr, draw=blue!65!black}
]
\node[wide, anchor=north, minimum height=1.25cm] (tests) at (0,0)
  {\textbf{Exact symmetries define the tests}\\[2pt]
  Perfectly accepted inputs are exactly the target states.\\
  Fermionic and bosonic characterizations in Appendices~\ref{sec:fermions} and~\ref{sec:bosons}};

\node[box=teal, anchor=north, minimum height=2.1cm] (overlap) at (-4,-1.95)
  {\textbf{Compare overlapping tests}\\[3pt]
  Rotations of copy labels describe the accepted spaces. Their overlaps admit uniform spectral bounds.\\[3pt]
  Lemmas~\ref{lem:so3}, \ref{lem:fermions-su3}, and~\ref{lem:boson-overlaps}};
\node[box=blue, anchor=north, minimum height=2.1cm] (closest) at (4,-1.95)
  {\textbf{Choose a closest target}\\[3pt]
  The deviation has no component along the allowed motions within the target family.\\[3pt]
  Lemma~\ref{lem:normal-stationarity}};

\node[box=teal, anchor=north, minimum height=2.65cm] (flow) at (-4,-4.75)
  {\textbf{Turn spectral bounds into a short path}\\[3pt]
  Overlaps bound the rejection gradient.\\
  Follow the direction of steepest decrease.\\
  Small rejection bounds the path length to a perfectly accepted target.\\[3pt]
  Theorems~\ref{thm:overlap} and~\ref{thm:restoration}};
\node[box=blue, anchor=north, minimum height=2.65cm] (local) at (4,-4.75)
  {\textbf{Measure local sensitivity}\\[3pt]
  Copy rotations determine the response to allowed deviations. Controlling the remaining terms gives a lower bound near the family.\\[3pt]
  Lemmas~\ref{lem:normal-transition} and~\ref{lem:normal-bound}};

\node[box=teal, anchor=north, minimum height=3.3cm] (global) at (-4,-8.1)
  {\textbf{Obtain global distance bounds}\\[3pt]
  Fermions, centered Gaussians, and coherent states: Theorems~\ref{thm:fermions-calibration} and~\ref{thm:boson-base}.\\[5pt]
  Displaced Gaussians: start with proving the four-copy bound then transfer this guarantee to three copies.\\[3pt]
  Theorem~\ref{thm:boson-full}};
\node[box=blue, anchor=north, minimum height=3.3cm] (examples) at (4,-8.1)
  {\textbf{Construct examples with known distance}\\[3pt]
  Explicit deviations give the local sensitivity bounds.\\[5pt]
  Appendices~\ref{sec:fermions} and~\ref{sec:bosons}, and Lemma~\ref{lem:testing-exact-curves}};

\node[box=teal, anchor=north, minimum height=2.45cm] (testing) at (-4,-12.1)
  {\textbf{Combine bounds and repeat the test}\\[3pt]
  Larger joint tests combine global control with local sensitivity to give tolerant testing.\\[3pt]
  Theorems~\ref{thm:normal-hierarchy} and~\ref{thm:testing-optimal-gap}};
\node[box=blue, anchor=north, minimum height=2.45cm] (lower) at (4,-12.1)
  {\textbf{Bound the copies needed by any tester}\\[3pt]
  Distinguishing the hard pair already requires many copies, even with arbitrary joint measurements.\\[3pt]
  Theorem~\ref{thm:testing-lower-bound}};

\node[wide, anchor=north, minimum height=1.2cm] (optimal) at (0,-15.25)
  {\textbf{Compare the upper and lower copy bounds}\\[3pt]
  Optimal dependence on the distance gap and confidence in the stated regimes};

\draw[arr] ($(tests.south)+(-4,0)$) -- (overlap.north);
\draw[arr] ($(tests.south)+(4,0)$) -- (closest.north);
\draw[arr] (overlap.south) -- (flow.north);
\draw[arr] (flow.south) -- (global.north);
\draw[arr] (global.south) -- (testing.north);
\draw[arr] (closest.south) -- (local.north);
\draw[arr] (local.south) -- (examples.north);
\draw[arr] (examples.south) -- (lower.north);
\draw[localarr] (local.west) -- (0,0 |- local.west) |- (testing.east);
\draw[arr] (testing.south) -- ($(optimal.north)+(-4,0)$);
\draw[arr] (lower.south) -- ($(optimal.north)+(4,0)$);
\end{tikzpicture}
\endgroup
\caption{The proof has two complementary parts. The global argument turns approximate copy symmetry into a nearby target. The local argument quantifies sensitivity to small departures. Their bounds meet in the testing step, while explicit examples give independent lower bounds. Appendix~\ref{sec:preliminaries} supplies the general implications, and Appendices~\ref{sec:fermions} and~\ref{sec:bosons} establish the ingredients for each family. The full bosonic reduction keeps the joint pure output throughout the gradient-flow argument.}
\label{fig:guide-argument}
\end{figure}
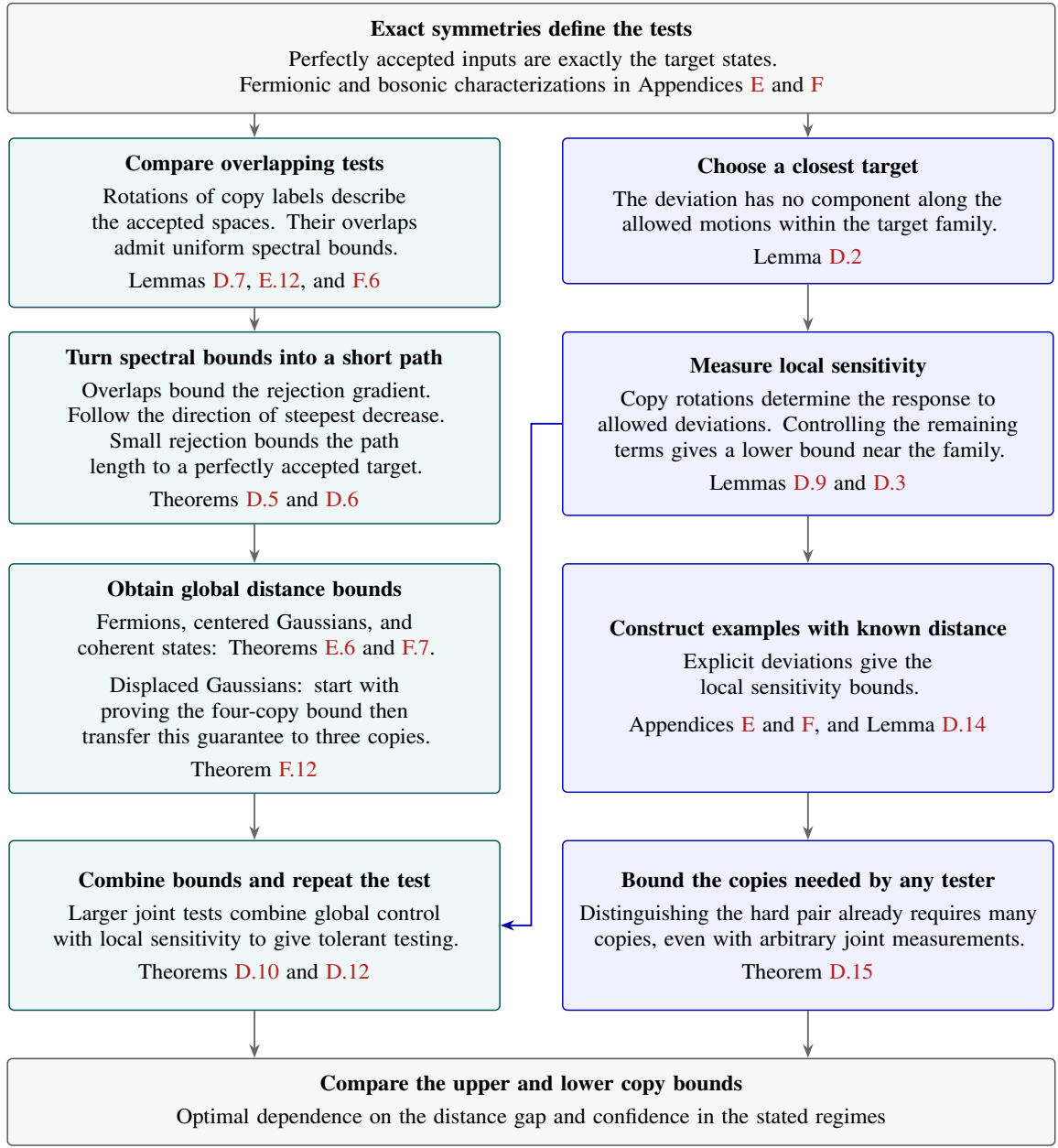

\newpage

\section{Notation table}

\begin{table}[htbp]

\centering
\small
\renewcommand{\arraystretch}{1.12}
\caption{Notation used across multiple subsections. Subscripts specifying the target family or copy number are omitted when clear from context.}
\label{tab:notation}
\begin{tabularx}{\linewidth}{
    @{} >{\raggedright\arraybackslash}p{0.27\linewidth} X @{}
}
\toprule
Notation & Meaning \\
\midrule

$n$ & Number of physical modes. \\

$\mathcal H$
& One-copy Hilbert space. \\

$\ket\psi,\ket\phi$
& Pure-state vectors. $\ket\psi$ usually denotes the unknown input. \\

$\ket g$
& Target or reference state. In closest-target arguments, a state maximizing the overlap with $\ket\psi$. \\

$\ket\Omega$
& Vacuum state, or the reference Slater determinant in Section~\ref{sec:fermions}. \\

\midrule

$\mathcal T$
& Target family of unit vectors, including their global phases. \\

$\mathcal G_{\mathrm F}$
& Pure fermionic Gaussian states. \\

$\mathcal S$
& Slater determinant states, allowing arbitrary particle number. \\

$\mathcal G_0$
& Centered pure bosonic Gaussian states. \\

$\mathcal C$
& Coherent states. \\

$\mathcal G$
& All pure bosonic Gaussian states, including displacements. \\

\midrule

$q$
& Number of identical input copies measured in one block. \\

$P$, $P_q$, $P_{\mathcal T,q}$
& Orthogonal acceptance projector, with the copy number and target family indicated when needed. \\

$Q$
& Rejection projector: $Q=I-P$. \\

$R(\psi)$, $p(\psi)$
& Rejection and acceptance probabilities. For a projector test, $R(\psi)=\|Q\ket\psi^{\otimes q}\|_2^2$ and $p(\psi)=1-R(\psi)$. \\

$\mathcal Z$
& Perfectly accepted one-copy states: $\mathcal Z=\{\ket\psi:\|\ket\psi\|_2=1,\ R(\psi)=0\}$. For an exact test, $\mathcal Z=\mathcal T$. \\

\midrule

$F_{\mathcal T}(\psi)$
& Best squared overlap with the target: $\sup_{\ket g\in\mathcal T}|\ip g\psi|^2$. \\

$d_{\mathcal T}(\psi)$
& Half trace distance to the target family: $d_{\mathcal T}=\sqrt{1-F_{\mathcal T}}$. \\

$\theta_{\mathcal T}(\psi)$
& Angular distance to the target: $\theta_{\mathcal T}=\arccos\sqrt{F_{\mathcal T}}$, so $d_{\mathcal T}=\sin\theta_{\mathcal T}$. \\

\midrule

$\ket\eta$ (closest-target expansion)
& Normalized deviation in $\ket\psi=\sqrt{1-d_{\mathcal T}(\psi)^2}\ket g +d_{\mathcal T}(\psi)\ket\eta$, with $\ip g\eta=0$. \\

$D(g^{\otimes q})[\eta]$
& First-order change of the tensor power at $\ket g$ in the direction $\ket\eta$. \\

$\lambda$, $\lambda_{\mathcal T,q}$
& A positive lower bound on $\|Q(D(g^{\otimes q})[\eta])\|_2^2$, and the sharp uniform quadratic rejection coefficient for a specified target family and test, respectively. \\

$\ket{\chi_\psi}$
& Unnormalized correction in Eq.~\eqref{eq:restoration-contraction}. For permutation-invariant $P$, $\nabla R(\psi)=-2q\ket{\chi_\psi}$. \\

$\nabla R$
& Gradient on the one-copy unit sphere, using the real metric $\Re\ip uv$. \\

$\kappa$
& Constant in the gradient inequality $\|\nabla R\|_2^2\ge4q^2R(\kappa-R)$ on $0<R<\kappa$. \\

$\mathcal U_q$
& Universal rejection upper bound: $\mathcal U_q(d_{\mathcal T}^2)=1-(1-d_{\mathcal T}^2)^q$. \\

$\ell_\kappa$, $\ell_{\mathrm B}$
& Global rejection lower-bound functions defined in Eqs.~\eqref{eq:ell-definition} and \eqref{eq:ellB-definition}. \\

$B_q$, $h_{q,\lambda}$
& Squared norm of the higher-order tensor remainder, and the resulting rejection lower bound from the closest-target expansion, see Eqs.~\eqref{eq:normal-B} and \eqref{eq:normal-h}. \\

\midrule

$\varepsilon$, $\delta$
& Distance threshold for ordinary testing, and error parameter. \\

$a$, $b$ (tolerant testing)
& Near and far distance thresholds, with $0\le a<b\le1$. \\

$\eta$ (tolerant testing)
& Scalar relative gap, defined by $b=(1+\eta)a$, distinct from the deviation vector $\ket\eta$. \\

$M$ (testing statistics)
& Number of independent measured blocks. $qM$ copies are used. \\

$\mathcal U$, $L$ (testing statistics)
& Upper rejection bound for near inputs and lower rejection bound for far inputs, with $\mathcal U<L$. \\

\midrule

Units and conventions
& $\hbar=1$. 
\\

\bottomrule
\end{tabularx}
\end{table}

\newpage
\section{Preliminaries}\label{sec:math_tools}

This appendix introduces the mathematical tools used in the rest of the supplemental material: the rejection probability and the statistics of repeated tests (Sections~\ref{sec:tensor-powers} and~\ref{sec:tolerant_testing_intro}), groups, Lie groups, and representations (Section~\ref{sec:group-theory}), calculus on the unit sphere (Section~\ref{sec:math-geometry}), and the fermionic and bosonic settings themselves (Section~\ref{sec:math-fb}).

\subsection{Tensor powers}
\label{sec:tensor-powers}
For a nonempty phase-invariant family $\mathcal T$ of unit vectors, put
\begin{subequations}\label{eq:distances}
\begin{align}
 F_{\mathcal T}(\psi)&=\sup_{\ket g\in\mathcal T}|\ip g\psi|^2, \nonumber\\
 d_{\mathcal T}(\psi)&=\sqrt{1-F_{\mathcal T}(\psi)}, \nonumber\\
 \theta_{\mathcal T}(\psi)&=\arccos\sqrt{F_{\mathcal T}(\psi)}. \label{eq:d-sin}
\end{align}
\end{subequations}
We call $F_{\mathcal T}$ the \emph{target fidelity}. Thus $d_{\mathcal T}(\psi)=\sin\theta_{\mathcal T}(\psi)$. We suppress the target subscript and the input argument when unambiguous. For normalized states, the trace distance and the phase-adjusted vector distance satisfy
\begin{subequations}\label{eq:prelim-trace-distance}
\begin{align}
 \tfrac12\|\ket\psi\bra\psi-\ket g\bra g\|_1
 &=\sqrt{1-|\ip g\psi|^2}, \nonumber\\
 \inf_\alpha\|\ket\psi-\e^{i\alpha}\ket g\|_2^2
 &=2(1-|\ip g\psi|).
\end{align}
\end{subequations}
Moreover $\bigl||\ip g\psi|^2-|\ip g\phi|^2\bigr| \le2\|\ket\psi-\ket\phi\|_2$. Taking suprema shows that $F,d,\theta$ are continuous, without requiring an attained supremum. For a positive integer $q$, let $P$ be an orthogonal projector on $\mathcal H^{\otimes q}$ and set $Q=I-P$. Define the rejection and acceptance probabilities as
\begin{subequations}\label{eq:main-residual}
\begin{align}
 R(\psi)&=\|Q\ket\psi^{\otimes q}\|_2^2, \nonumber\\
 p(\psi)&=1-R(\psi)\label{eq:acceptance}.
\end{align}
\end{subequations}
Here $\Ran P$ denotes the range of $P$, its accepted subspace. Exactness for $\mathcal T$ means that, for unit $\ket\psi$, $\ket\psi^{\otimes q}\in\Ran P$ if and only if $\ket\psi\in\mathcal T$. Perfect completeness means $R(g)=0$ for all $\ket g\in\mathcal T$ and implies the universal bound
\begin{align}
 R(\psi)&\le1-F_{\mathcal T}(\psi)^q\le qd_{\mathcal T}(\psi)^2,
 \label{eq:universal-upper}
\end{align}
since $p(\psi)=\|P\ket\psi^{\otimes q}\|_2^2 \ge\sup_{\ket g\in\mathcal T}|\ip g\psi|^{2q} =F_{\mathcal T}(\psi)^q=(1-d_{\mathcal T}(\psi)^2)^q$. The map $\ket\psi\mapsto\ket\psi^{\otimes q}$ is a bounded polynomial on bounded sets, so $R$ is smooth.

\subsection{Tolerant testing}\label{sec:tolerant_testing_intro}
Here we present a tolerant tester. Indeed, we introduce a procedure that allows us to obtain a tolerant tester. This is a standard result, presented for example in~\cite{boucheron_concentration}. For $0\le a<b\le1$, an $(a,b)$ tolerant test accepts inputs with $d_{\mathcal T}(\psi)\le a$ and rejects inputs with $d_{\mathcal T}(\psi)\ge b$, with bounded error in each case. There is no requirement for inputs between these two distances. The case $a=0$ is non-tolerant testing. 

We apply the same $q$-copy measurement to $M$ disjoint blocks of the input $\ket\psi^{\otimes qM}$. Write $Y_i=1$ when block $i$ rejects and $Y_i=0$ when it accepts. The outcomes are independent and satisfy
\begin{align}
 \Pr\{Y_i=1\}&=r=R(\psi), \nonumber \\
 \overline Y&=\frac1M\sum_{i=1}^M Y_i.
 \label{eq:testing-empirical-mean}
\end{align}
To construct a tolerant test, we need an upper bound $\mathcal{U}$ on rejection for every near input and a lower bound $L>\mathcal{U}$ for every far input. The following proposition converts this separation into a copy count.

\begin{proposition}[From rejection gaps to testing guarantees]\label{prop:testing-statistics}
Suppose a perfectly complete $q$-copy test rejects every far input with probability at least $L>0$. In the non-tolerant case, accepting only when every block accepts has zero error on targets and error at most $\delta$ on far inputs, using
\begin{align}
 q\left\lceil L^{-1}\log(1/\delta)\right\rceil
 \label{eq:testing-exact-count}
\end{align}
copies. For tolerant testing, suppose instead that near inputs have rejection probability in $[0,\mathcal{U}]$ and far inputs in $[L,1]$, where $0\le \mathcal{U}<L\le1$. Accept when $\overline Y<(L+\mathcal{U})/2$ and reject otherwise. Each error is at most $\delta/2$ using
\begin{align}
 q\left\lceil\frac{8L}{(L-\mathcal{U})^2}\log\frac2\delta\right\rceil
 \label{eq:testing-tolerant-count}
\end{align}
copies, for $0<\delta<1$.
\end{proposition}

\subsection{Group theory}\label{sec:group-theory}

To develop the argument, we first introduce the group-theoretic notions on which it relies. Recall that a group $G$ is a set with an associative binary operation (we write $gh$ for $g\circ h$) that has an identity $1_G$ and an inverse $g^{-1}$ for each $g\in G$; the identity and inverses are then unique and two-sided. A subgroup $H\le G$ is a subset closed under products and inverses. If $K\le G$ and $gKg^{-1}=K$ for all $g\in G$, then $K$ is \emph{normal} ($K\unlhd G$) and the cosets $G/K=\{gK:g\in G\}$ form a group. 

A homomorphism $f:G\to G'$ satisfies $f(g_1g_2)=f(g_1)f(g_2)$, which forces $f(1_G)=1_{G'}$ and $f(g^{-1})=f(g)^{-1}$; its kernel is normal, and the first isomorphism theorem gives $G/\ker f\cong\operatorname{im}f$.\\

\paragraph{Global phase.}
The first quotient we meet is global phase: the scalars $\{\e^{i\alpha}I\}\cong U(1)$ are normal in $U(\mathcal H)$ and do not change $\rho_\psi$, so a unitary only matters up to a phase, and $F_{\mathcal T}$, $d_{\mathcal T}$, and $R$ are all blind to it; the double cover $\mathrm{Spin}(2n)\to SO(2n)$ of Table~\ref{tab:math-orbits}, whose kernel is $\{\pm I\}$, is a quotient of the same kind.\\

\paragraph{Actions and orbits.}
A group $G$ \emph{acts} on a set $X$ if there is a map $G\times X\to X$, $(g,x)\mapsto g\cdot x$, with (i) $1_G\cdot x=x$ and (ii) $(gh)\cdot x=g\cdot(h\cdot x)$. Define the \emph{orbit} $\mathrm{orb}_G(x)=Gx=\{g\cdot x:g\in G\}$ and the \emph{stabilizer} $\mathrm{stab}_G(x)=G_x=\{g\in G:g\cdot x=x\}\le G$. The orbit-stabilizer theorem says $gG_x\mapsto g\cdot x$ is a well-defined bijection $G/G_x\to Gx$. Every target family in this paper (for $\GF$ each parity sector, for $\Sl$ each particle number) is a single orbit of a reference state $\ket\Omega$ under a group $G$ of free unitaries, the Gaussian unitaries or a subgroup of them (Table~\ref{tab:math-orbits}),
\begin{align*}
 \mathcal T=\{U_h\ket\Omega:\ h\in G\}\cong G/G_\Omega,
\end{align*}
so $d_{\mathcal T}$ is a distance to an orbit. Two things follow. First, orbits are not subspaces: $\ket{0,0,0,0}$ and $\ket{1,1,1,1}$ both lie in $\mathcal G_{\mathrm F}^{+}$, but $\sqrt{1-t^2}\ket{0,0,0,0}+t\ket{1,1,1,1}$ does not for $0<t\le1/\sqrt2$.

We can see that $d_{\mathcal T}$ is a distance to a curved set (so we can't linearly project), and this is why we need the geometry in Section~\ref{sec:math-geometry}. Second, for $V\in G$, $V\mathcal T=\mathcal T$ gives $F_{\mathcal T}(V\psi)=F_{\mathcal T}(\psi)$, and if $V^{\otimes q}$ commutes with $P$ then $R(V\psi)=R(\psi)$ as well.

\subsubsection{Lie groups}\label{sec:math-lie}
Continuous symmetries of a physical system are generated by their velocity at the identity; by Noether's theorem, that generator is a conserved quantity. Lie theory formalizes this: the tangent space at the identity, with its commutator, linearizes the group, and exponentiating a tangent vector recovers a one-parameter family of symmetries. For us, this is the dictionary between Gaussian unitaries and the quadratic Hamiltonians that generate them, and it is how the tangent directions of a target family, which control the closest-target expansion, are computed.

A {Lie group} is a smooth manifold with a group structure whose multiplication and inversion are smooth. The groups in Table~\ref{tab:math-orbits} are closed subgroups of $GL(m,\mathbb R)$ or $GL(m,\mathbb C)$, so by Cartan's closed subgroup theorem~\cite{hall_lie} they are Lie groups, and so are their stabilizers; in particular $G/G_x$ is a smooth manifold. The \emph{Lie algebra} is the tangent space at the identity,
\begin{align}
 \mathfrak g=T_{1}G=\{X:\e^{tX}\in G\text{ for all }t\in\mathbb R\},
\end{align}
closed under $[X,Y]=XY-YX$. The exponential map $\mathfrak g\to G$ is a diffeomorphism near $0$, and $t\mapsto\e^{tX}$ are exactly the one-parameter subgroups. For unitary groups, take $X=-iH$ with $H$ self-adjoint, so $\e^{-itH}$ is the familiar time evolution; on bosonic Fock space, $H$ is unbounded and Stone's theorem still matches self-adjoint generators with strongly continuous one-parameter unitary groups. A smooth action of $G$ on a manifold has orbits $Gx\cong G/G_x$ that are immersed submanifolds (embedded and compact when $G$ is compact), with
\begin{align}
 T_x(Gx)=\mathfrak g\cdot x=\Bigl\{\tfrac{\dd}{\dd t}\e^{tX}\!\cdot x\big|_{t=0}:X\in\mathfrak g\Bigr\},
 \qquad
 \dim Gx=\dim G-\dim G_x.
 \label{eq:math-orbit-tangent}
\end{align}
The fermionic Gaussian unitaries act on the finite-dimensional Fock space, so this applies directly. The bosonic ones act on an infinite-dimensional space, only strongly continuously, but the orbit map $h\mapsto U_h\ket\Omega$ is still smooth because the vacuum is an analytic vector for quadratic generators, and Eq.~\eqref{eq:math-orbit-tangent} holds there as well. Finally, a compact group carries a unique left- and right-invariant probability measure $\dd h$ (\emph{Haar measure}), and for a unitary representation $h\mapsto U_h$,
\begin{align}
 P_G=\int_GU_h\,\dd h
 \label{eq:math-haar}
\end{align}
is the orthogonal projector onto the vectors fixed by every $U_h$: invariance gives $U_kP_G=P_G=P_GU_k$, hence $P_G^2=P_G$, and invariance under $h\mapsto h^{-1}$ gives $P_G^\dagger=P_G$. 
For $U(1)$ acting by $\e^{i\theta A}$ with integer spectrum, $P_G=\mathbf 1_{\{0\}}(A)$ projects onto $\ker A$, which is how the bosonic tests are measured by photon counting (Section~\ref{subsec:boson-measurements}, Eq.~\eqref{eq:boson-optical-generator}).

\begin{table}[ht]
 \centering
 \small
 \begin{tabular}{@{}llllc@{}}
 \toprule
 Family & Free group $G$ & Generators ($\mathfrak g$) & Orbit, real dimension & Compact\\
 \midrule
 $\GFs$ & $\mathrm{Spin}(2n)\to SO(2n)$ & $\gamma_\mu\gamma_\nu$ & $SO(2n)/U(n)$, \ $n(n-1)$ & yes\\
 $\Sl$, $k$ particles & $U(n)$ (passive) & $a_i^\dagger a_j$ & $\mathrm{Gr}(k,\mathbb C^n)$, \ $2k(n-k)$ & yes\\
 $\Coh$ & displacements $\mathbb C^n$ & $a_j^\dagger,\ a_j$ & $\mathbb C^n$, \ $2n$ & no\\
 $\Gzero$ & $Sp(2n,\mathbb R)$ (metaplectic) & $a_i^\dagger a_j,\ a_i^\dagger a_j^\dagger,\ a_ia_j$ & $Sp(2n,\mathbb R)/U(n)$, \ $n(n+1)$ & no\\
 $\GB$ & both of the above & all of the above & $\bigl(Sp(2n,\mathbb R)\ltimes\mathbb C^n\bigr)/U(n)$, \ $n(n+3)$ & no\\
 \bottomrule
 \end{tabular}
 \caption{Each target family is the orbit of a reference state ($\ket\Omega$, or $e_1\wedge\cdots\wedge e_k$ for $\Sl$) under a group of free unitaries, generated by operators quadratic (for displacements, linear) in the mode operators. The stabilizer in $SO(2n)$, $U(n)$, or $Sp(2n,\mathbb R)$ is the number-conserving part, $U(n)$ or $U(k)\times U(n-k)$; for $\Coh$ it is trivial up to phase. Compactness decides whether closest targets exist automatically.}
 \label{tab:math-orbits}
\end{table}

It is worth noting that $\mathrm{Spin}(2n)$ and $U(n)$ are compact, so the fermionic families, as continuous images of compact groups (one orbit per parity or particle number), are compact, and the continuous function $\ket g\mapsto|\ip g\psi|$ attains its maximum on them; this is why a closest fermionic target always exists. The bosonic groups are not compact, and there a closest target comes from Lemma~\ref{lem:normal-attainment}.

\subsubsection{Representation theory}\label{sec:math-rep}
Linear algebra splits a space by the eigenvalues of one operator; representation theory splits it under a whole group at once, into pieces the group cannot mix, and anything built from the group respects the split. This is how we structure our tests; each accepted subspace is the set of vectors fixed by a second group, acting on copies rather than modes. For bosons this group rotates the copy labels orthogonally (Eq.~\eqref{eq:boson-group-average}); for fermionic Gaussian states it rotates the Majorana copy labels, and for Slater determinants it transfers particles between copies (below). Decomposing $\mathcal H^{\otimes q}$ under this action measures how far $\ket\psi^{\otimes q}$ is from fixed, one irreducible piece at a time, each finite-dimensional; Lemmas~\ref{lem:so3} and~\ref{lem:fermions-su3} bound a spectral gap uniformly over all pieces, which is why the bounds do not depend on the number of modes or, for bosons, on photon number.

A (unitary) \emph{representation} of $G$ on $V$~\cite{hall_lie} is a continuous homomorphism $h\mapsto U_h\in U(V)$. A subspace is invariant if every $U_h$ preserves it, and $V$ is \emph{irreducible} (an irrep) if its only closed invariant subspaces are $0$ and $V$. For unitary representations the orthogonal complement of an invariant subspace is invariant, and every irrep of a compact group is finite-dimensional, so $V$ splits into irreps (a Hilbert direct sum in infinite dimensions); grouping equivalent ones gives the \emph{isotypic decomposition}
\begin{align}
 V\cong\bigoplus_iV_i\otimes M_i,
 \qquad
 \dim V=\sum_im_i\dim V_i,
 \label{eq:math-isotypic}
\end{align}
where $G$ acts on the irrep $V_i$ and trivially on the multiplicity space $M_i\cong \mathbb{C}^{m_i}$ (note this is just the vector space indexing $m_i$ copies of $V_i$ in $V$ and isn't itself a subspace of $V$). Schur's lemma says a map between irreps that commutes with $G$ is zero or an isomorphism, and an endomorphism of an irrep commuting with $G$ is a scalar. Hence inequivalent isotypic components are orthogonal, operators commuting with $G$ act as $\bigoplus_iI_{V_i}\otimes A_i$, and operators built from $G$ (group elements, generators, averages) act as $\bigoplus_iB_i\otimes I_{M_i}$. The fixed vectors form the trivial isotypic component, onto which Eq.~\eqref{eq:math-haar} projects. 

\paragraph{Symmetric subspace.}
The simplest copy group is the symmetric group $S_q$, acting on $\mathcal H^{\otimes q}$ by permuting the tensor factors, $W_\pi(v_1\otimes\cdots\otimes v_q)=v_{\pi^{-1}(1)}\otimes\cdots\otimes v_{\pi^{-1}(q)}$. Averaging it as in Eq.~\eqref{eq:math-haar}, a finite sum here, gives
\begin{align}
 P_{\Sym}=\frac1{q!}\sum_{\pi\in S_q}W_\pi,
 \label{eq:math-sym}
\end{align}
the orthogonal projector onto the vectors fixed by every permutation, which form the \emph{symmetric subspace} $\Sym^q(\mathcal H)=\Ran P_{\Sym}$; weighting each $W_\pi$ by the sign of $\pi$ gives the antisymmetric subspace $\Lambda^q(\mathcal H)$ instead. We use two properties. First, $\ket\psi^{\otimes q}\in\Sym^q(\mathcal H)$ for every $\ket\psi$, since permuting identical copies does nothing. Second, an operator that commutes with every $W_\pi$, such as $U^{\otimes q}$ or a permutation-invariant acceptance projector (the setting of Lemma~\ref{lem:first-variation}), maps $\Sym^q(\mathcal H)$ into itself, so a test can be studied on $\Sym^q(\mathcal H)$ alone. The same construction builds the bosonic and fermionic particle spaces (Section~\ref{sec:math-fock}), where it symmetrizes particles rather than copies.

\paragraph{Spin.}
The irreps we need are those of $SU(2)$ and $SO(3)$, the groups of angular momentum~\cite{hall_lie}. Intuitively, an irrep is labeled by one number, its spin $\ell$, and inside it the angular momentum along any axis takes each value $-\ell,\ldots,\ell$ exactly once. Concretely, take a basis $J_x,J_y,J_z$ of the Lie algebra and the ladder operators $J_\pm=J_x\pm iJ_y$, so that
\begin{align}
 [J_z,J_\pm]=\pm J_\pm,\qquad[J_+,J_-]=2J_z
 \label{eq:math-su2}
\end{align}
(the relations of the Slater transfer operators, Eq.~\eqref{eq:fermions-slater-generators}). The first relation says $J_\pm$ raise and lower the $J_z$-eigenvalue, called the \emph{weight}, by one; a finite-dimensional irrep $V_\ell$ is a single ladder from a top vector of weight $\ell$ (a \emph{highest weight vector}, killed by $J_+$) down to $-\ell$, with one state per weight,
\begin{align}
 J_z\ket{\ell,m}=m\ket{\ell,m},\qquad m=-\ell,-\ell+1,\ldots,\ell,
 \label{eq:math-weight-basis}
\end{align}
so $\dim V_\ell=2\ell+1$ and $\ell\in\frac12\mathbb N$. The Casimir $J^2=J_x^2+J_y^2+J_z^2$ commutes with every generator, so by Schur's lemma it equals $\ell(\ell+1)$ on $V_\ell$; its kernel is the spin-$0$ part, the \emph{singlets}, which is what the Slater test accepts (Eq.~\eqref{eq:fermions-slater-test}). The irreps of $SO(3)$ are those with integer $\ell$ (a $2\pi$ rotation acts as $(-1)^{2\ell}$), and $V_\ell$ can be realized as the degree-$\ell$ spherical harmonics on the unit sphere. Finally, the component $J_u=u\cdot J$ along a unit axis $u$ is conjugate to $J_z$ by a rotation, so it has spectrum $\{-\ell,\ldots,\ell\}$ on $V_\ell$ with each eigenvalue once; for integer $\ell$ the vectors in $V_\ell$ fixed by all rotations about $u$ therefore form a line, spanned by a \emph{zonal} state. This is the fact behind Lemma~\ref{lem:so3}, where $J_u$ is written $L_u$.

\paragraph{Copy symmetries of the tests.}
Each test accepts the vectors fixed by a group acting on the copy labels:
\begin{itemize}
 \item for $\GF$, the rotations $SO(q)$ of the Majorana copy labels; on two copies the generator is $J_{1,2}=-\tfrac12\Pi_1\Lambda$ (Eq.~\eqref{eq:normal-fermion-generators}), so Bravyi's criterion $\Lambda\ket\psi^{\otimes2}=0$~\cite{bravyi} says exactly that $\ket\psi^{\otimes2}$ is fixed (Lemma~\ref{lem:normal-fermion-fixed-spaces});
 \item for $\Sl$, the transfers $SU(q)$ of particles between copies; on two copies this is the $SU(2)$ above, and the test accepts its singlets;
 \item for bosons, subgroups $G\le O(q)$ acting through $\Gamma(I_n\otimes O)$ (Eq.~\eqref{eq:boson-group-average}).
\end{itemize}
The global bounds use an average of overlapping projectors, $T=\frac1m\sum_iP_i$. Since each $P_i$ is a projector, the eigenvalue-$1$ eigenspace of $T$ is exactly $\bigcap_i\Ran P_i$, the states every $P_i$ accepts, and a bound $\Spec(T)\subset\{1\}\cup[0,\beta]$ with $\beta<1$ is a spectral gap below it; Theorem~\ref{thm:overlap} turns this gap into the gradient inequality. Here $m=3$: on three copies the three pairwise tests become rotations about three orthogonal axes ($SO(3)$, for $\GF$ and for bosons) or the three pairwise $SU(2)$ subgroups of $SU(3)$ (for $\Sl$). Lemmas~\ref{lem:so3} and~\ref{lem:fermions-su3} bound the spectrum of the average of the three pair projectors by $\{1\}\cup[0,7/12]$ and $\{1\}\cup[0,5/9]$ in every representation, so the bounds depend on neither the spin nor, through it, the number of modes or photons. For $SO(3)$ this comes from the zonal states above, whose overlaps are Legendre polynomials, $\ip{z_u}{z_v}=P_\ell(u\cdot v)$.

\subsection{Geometry}\label{sec:math-geometry}

An input is a unit vector, so all of the geometry takes place on the unit sphere $S=\{\psi\in\mathcal H:\|\psi\|_2=1\}$; the target families are curved subsets of it (orbits, Section~\ref{sec:group-theory}). The rejection probability is a smooth function $R:S\to[0,1]$ that vanishes exactly on the target family, by exactness (Section~\ref{sec:tensor-powers}), and robustness asks how fast $R$ grows as one leaves $\mathcal T$. Near $\mathcal T$ this is a second-derivative question; far from it, it is a question about critical points and gradients. Both only need $S$ as a smooth manifold with the round metric $\Re\ip uv$, whose geodesics are the great circles $\cos s\,\psi+\sin s\,v$ ($\Re\ip\psi v=0$, $\|v\|_2=1$). Since $\mathcal T$ contains every phase of its elements, the geodesic distance from $\psi$ to $\mathcal T$ is the angle of Eq.~\eqref{eq:distances},
\begin{align*}
 \theta_{\mathcal T}=\arccos\sqrt{F_{\mathcal T}},\qquad d_{\mathcal T}=\sin\theta_{\mathcal T}.
\end{align*} For bosons $S$ is infinite-dimensional, and the same calculus applies. The one other geometric object we use is the set of tensor powers $\mathcal M_q=\{\psi^{\otimes q}\}\subset\Sym^q(\mathcal H)$, where the tests act. The tensor-power map turns homogeneous degree-$q$ polynomial conditions on $\psi$ into linear conditions on $\psi^{\otimes q}$ (it is the Veronese map into $\Sym^q(\mathcal H)$), which is why a projector can test a curved family. It is also why one copy is never enough: for $q=1$ the map is the identity, so a perfectly complete one-copy test accepts every unit vector in the span of $\mathcal T$, and the span contains non-targets (such as the superposition in Section~\ref{sec:group-theory}). On two copies the quadratic equations that cut out $\GF$, $\Sl$, $\Gzero$, and $\Coh$ become linear, while $\GB$ needs three (Proposition~\ref{prop:boson-two-copy-impossible}). The map also scales lengths by exactly $\sqrt q$,
\begin{align*}
 \|D(\psi^{\otimes q})[v]\|_2=\sqrt q\,\|v\|_2\qquad(\ip\psi v=0),
\end{align*}
since the $q$ summands of Eq.~\eqref{eq:tensor-derivative} are then orthogonal. So geometry on $\mathcal M_q$ is geometry on $S$ but rescaled.

\subsubsection{Calculus on the sphere}\label{sec:math-sphere}
Viewing $\mathcal H$ as a real inner product space with $\Re\ip uv$ (the real metric), a curve on the unit sphere through $x$ has its velocity in the tangent plane $x^\perp$, and the gradient of $f$ along the sphere is the tangent vector whose inner product with each velocity is the rate of change of $f$ (similar to $\mathbb R^3$). A \emph{smooth curve of unit vectors} through $\psi$ is a family $\ket{\psi(t)}\in S$, differentiable in a real parameter $t$ near $0$ (not necessarily physical time), with $\ket{\psi(0)}=\ket\psi$ and velocity $\ket{\dot\psi(0)}=\ket v$, so $\ket{\psi(t)}=\ket\psi+t\ket v+O(t^2)$. Differentiating $\ip{\psi(t)}{\psi(t)}=1$ at $t=0$ gives $2\Re\ip\psi v=0$, so the velocities form the \emph{tangent space}
\begin{align}
 T_\psi S=\{v:\Re\ip\psi v=0\},
 \label{eq:math-tangent-space}
\end{align}
and each such $v$ occurs: along a great circle below, or, for $\ip\psi v=0$, along $\ket{\psi(t)}=(\ket\psi+t\ket v)/\sqrt{1+t^2\|v\|_2^2}$. Normalization only forces $\Re\ip\psi v=0$; the remaining $\Im\ip\psi v$ is a first-order change of global phase, which the stronger condition $\ip\psi v=0$ also removes. For a smooth $f:S\to\mathbb R$, such as $p$ or $R$, the directional derivative $Df(\psi)[v]=\frac{\dd}{\dd t}f(\psi(t))\big|_{t=0}$ depends only on $v$ and is real-linear in it, so it is represented by a unique tangent vector, the \emph{gradient}:
\begin{align}
 Df(\psi)[v]=\Re\ip{\nabla f(\psi)}v\quad\text{for all }v\in T_\psi S,
 \qquad
 \nabla f(\psi)=\tilde\nabla f(\psi)-\Re\ip\psi{\tilde\nabla f(\psi)}\,\psi.
 \label{eq:math-gradient}
\end{align}
Here $\tilde\nabla f$ is the Euclidean gradient of any smooth extension of $f$ to $\mathcal H$ (for $p$, the function $\ip{\psi^{\otimes q}}{P\psi^{\otimes q}}$ on all of $\mathcal H$), and the second formula subtracts the normal component. Lemma~\ref{lem:first-variation} computes this gradient for the rejection probability, $\nabla R=-2q\ket{\chi_\psi}$, which is tangent since $\ip\psi{\chi_\psi}=0$ (Eq.~\eqref{eq:restoration-contraction}).

Because $R(\e^{i\alpha}\psi)=R(\psi)$, the phase direction $i\psi\in T_\psi S$ is always flat, and it suffices to vary along directions with $\ip\psi v=0$, which is the convention of Lemma~\ref{lem:first-variation}. A \emph{critical point} is a $\psi$ with $\nabla R(\psi)=0$. At a critical point, the second derivative of $R$ along a curve in $S$ with velocity $v$ depends only on $v$ (two such curves differ in acceleration by a tangent vector, which pairs with the vanishing gradient), which defines $\operatorname{Hess}_SR(\psi)[v,v]$; as in calculus, a negative value rules out a local minimum.

At a target $g$, $R$ takes its minimum value $0$, so $g$ is critical. Along a great circle $\psi_s=\cos s\,g+\sin s\,\eta$, with unit $\eta$ orthogonal to $g$ and to the tangent directions $\mathfrak g\cdot g$ of the family (Section~\ref{sec:math-closest}), $Q\psi_s^{\otimes q}=s\,Q\,D(g^{\otimes q})[\eta]+O(s^2)$, so
\begin{align}
 \operatorname{Hess}_SR(g)[\eta,\eta]=2\bigl\|Q\,D(g^{\otimes q})[\eta]\bigr\|_2^2\ge2\lambda_{\mathcal T,q}
 \label{eq:math-normal-hessian}
\end{align}
by Eq.~\eqref{eq:higher-copy-coefficient-bound}, while the Hessian vanishes along $\mathcal T$, where $R\equiv0$. So $R$ is flat along the family and curves up at a uniform rate across it, with the mode-independent constants of Table~\ref{table:normal-coefficients}. This is the local picture. The global question is how steep $R$ is away from $\mathcal T$, which gradient flow answers (Section~\ref{sec:math-flow}); for fermions we also show that $R$ has no other local minima (Section~\ref{sec:math-closest}).

\subsubsection{Closest-point projections}\label{sec:math-closest}
In a vector space, the closest point of a subspace is the orthogonal projection, and the error is orthogonal to the subspace. For a curved set this survives to first order. Let $g$ maximize $|\ip g\psi|$ over $\mathcal T$ (it exists when $\mathcal T$ is compact, and by Lemma~\ref{lem:normal-attainment} for bosons) and write $\psi=\sqrt{1-d_{\mathcal T}^2}\,g+d_{\mathcal T}\,\eta$ as in Eq.~\eqref{eq:normal-frame}. Then $\eta$ is orthogonal to the tangent space of the orbit at $g$, which by Eq.~\eqref{eq:math-orbit-tangent} is $\mathfrak g\cdot g$; its part orthogonal to $g$ is a complex subspace for all our groups, so
\begin{align*}
 \ip v\eta=0\qquad\text{for every }v\in\mathfrak g\cdot g\text{ with }\ip gv=0
\end{align*}
(Lemma~\ref{lem:normal-stationarity}). After moving $g$ to the reference state, $\mathfrak g\cdot\ket\Omega$ is spanned, beyond the phase, by what the generators of Table~\ref{tab:math-orbits} create from the vacuum: pairs $a_i^\dagger a_j^\dagger\ket\Omega$ for $\GF$ and $\Gzero$, single photons $a_j^\dagger\ket\Omega$ for $\Coh$, both for $\GB$, and single orbital replacements for $\Sl$ (annihilators kill $\ket\Omega$, number-conserving quadratics act on it by a scalar, and on the Slater reference $a_i^\dagger a_j$ with $j$ occupied and $i$ empty is one replacement). The deviation therefore has no components in those degrees, which gives the restrictions of Table~\ref{table:normal-allowed-degrees}; and its second-order consequence is Eq.~\eqref{eq:math-normal-hessian}. Expanding the tests to leading order in the allowed deviations gives the local bound of Lemma~\ref{lem:normal-bound} and the sharp coefficients $\lambda_{\mathcal T,q}$ of Table~\ref{table:normal-coefficients}, which are what tolerant testing uses.

For fermions the same calculus also describes the rest of the landscape. Within a fixed parity (for Slaters, particle-number) sector, the Hessian identity of Theorem~\ref{thm:compact-maximum} gives, at every critical point with $R>0$, directions $v_j$ with $\sum_j\operatorname{Hess}_SR(\psi)[v_j,v_j]<0$ (Theorem~\ref{thm:fermion-hessian}), and Corollary~\ref{cor:fermion-unpromised-landscape} extends this to superpositions of sectors; so the targets are the only local minima of $R$ (Section~\ref{subsec:fermion-first-proof}). The distance bounds do not need this; they come from gradient flow (Section~\ref{sec:math-flow}).

\subsubsection{Gradient flow on the sphere}\label{sec:math-flow}
In $\mathbb R^d$, following $\dot x=-\nabla f$ decreases $f$ at rate $\frac{\dd}{\dd t}f=-\|\nabla f\|^2$. If moreover $\|\nabla f\|\ge c\sqrt f$ wherever $f>0$ (a {\L}ojasiewicz-type inequality), then per unit of arc length $s$,
\begin{align}
 \frac{\dd\sqrt f}{\dd s}=-\frac{\|\nabla f\|}{2\sqrt f}\le-\frac c2,
\end{align}
so $\sqrt f$ reaches $0$ within length $2\sqrt{f(x_0)}/c$, and the zero set lies within that distance of $x_0$. On the sphere one follows the tangential gradient instead, so the flow stays on $S$, and the argument is otherwise the same. Our distance bounds do this with $f=R$ (Section~\ref{sec:restoration}, and Section~\ref{subsec:fermion-second-proof} for fermions). The flow $\dot\psi=2q\ket{\chi_\psi}=-\nabla R$ (Lemma~\ref{lem:first-variation}) is tangent and orthogonal to $\psi$, the gradient inequality (Eq.~\eqref{eq:restoration-slope}) comes from the spectral gap of Theorem~\ref{thm:overlap}, and the endpoint lies in the zero set $\mathcal Z=\mathcal T$ by exactness. Integrating the inequality, $\arcsin\sqrt{R/\kappa}$ decreases at rate at least $q$ per unit length, so for $r_0<\kappa$
\begin{align*}
 \|\nabla R\|_2^2\ge4q^2R(\kappa-R)
 \quad\Longrightarrow\quad
 \theta_{\mathcal T}(\psi)\le\text{path length}\le\tfrac1q\arcsin\sqrt{r_0/\kappa},
\end{align*}
which is the envelope of Theorem~\ref{thm:restoration}. This uses neither compactness nor a closest target, which is why the same argument covers the infinite-dimensional, noncompact bosonic setting.

\subsection{Fermionic and bosonic quantum computation}\label{sec:math-fb}
Identical particles in three dimensions are either fermions or bosons. Fermions model matter; electrons in molecules and materials, and cold atoms in optical lattices, which qubit devices simulate through the Jordan--Wigner encoding. Bosons model light and vibrations, and underlie photonic and continuous-variable quantum technologies. In both, the Gaussian states are the free, noninteracting states: ground and thermal states of quadratic Hamiltonians, including the Slater determinants of Hartree--Fock theory~\cite{Arute2020HartreeFock} and the coherent and squeezed states of quantum optics~\cite{weedbrook_gaussian}. They are also classically easy in the right model: matchgate circuits with occupation-basis inputs and measurements~\cite{terhal_divincenzo}, and bosonic Gaussian circuits with Gaussian measurements, can be simulated efficiently. Table~\ref{tab:math-fb} summarizes how these two settings compare.
\begin{table}[ht]
 \centering
 \small
 \begin{tabular}{@{}lll@{}}
 \toprule
 & Fermions & Bosons\\
 \midrule
 One mode & $\mathbb C^2$ (occupation $0$ or $1$) & $\ell^2(\mathbb N)\cong L^2(\mathbb R)$\\
 Relations & $\{a_i,a_j^\dagger\}=\delta_{ij}$, $\{a_i,a_j\}=0$ & $[a_i,a_j^\dagger]=\delta_{ij}$, $[a_i,a_j]=0$\\
 $n$ modes & $\bigoplus_{k=0}^n\Lambda^k\mathbb C^n\cong\mathbb C^{2^n}$ & completion of $\bigoplus_{k\ge0}\Sym^k\mathbb C^n$\\
 $k$ particles & $\Lambda^k\mathbb C^n$, dimension $\binom nk$ & $\Sym^k\mathbb C^n$, dimension $\binom{n+k-1}k$\\
 Coordinates & Majoranas $\gamma_\mu$; Jordan--Wigner qubits & Bargmann functions $f(z)$, $z\in\mathbb C^n$\\
 Gaussian data & covariance $\Gamma(\psi)=-\Gamma(\psi)^T$; pure iff orthogonal & $Z=Z^T$, $\|Z\|_\infty<1$; and $b\in\mathbb C^n$\\
 Superselection & parity $(-1)^N$ & none; displacements allowed\\
 Test & $\ker\Lambda$ on 2 copies; Bell sampling & copy rotations (2 or 3 copies); photon counting\\
 \bottomrule
 \end{tabular}
 \caption{The fermionic and bosonic settings. Free groups and compactness are in Table~\ref{tab:math-orbits}.}
 \label{tab:math-fb}
\end{table}
\subsubsection{Fermions}
The creation and annihilation operators $a_j^\dagger,a_j$ satisfy
\begin{align}
 \{a_i,a_j\}&=0, \qquad
 \{a_i,a_j^\dagger\}=\delta_{i,j}I.
 \label{eq:fermions-car}
\end{align}
Here $\{A,B\}=AB+BA$. The particle-number and parity operators are
\begin{align}
 N&=\sum_{j=1}^n a_j^\dagger a_j, \qquad
 \Pi=(-1)^N.
\end{align}
The eigenvalues $+1$ and $-1$ of $\Pi$ correspond to even and odd particle number, respectively. The self-adjoint Majorana operators are

\begin{align}
 \gamma_{2j-1}=a_j+a_j^\dagger, \qquad
 \gamma_{2j}=-i(a_j-a_j^\dagger), \qquad
 \{\gamma_\mu,\gamma_\nu\}&=2\delta_{\mu,\nu}I.
 \label{eq:fermions-majorana-car}
\end{align}

It follows from \eqref{eq:fermions-majorana-car} that Fock space is the unique irrep of the Clifford algebra on $\mathbb R^{2n}$, and Jordan--Wigner writes the $\gamma_\mu$ as Pauli strings (Section~\ref{subsec:fermion-measurements}). Parity superselection is why no free Hamiltonian is linear in the $\gamma_\mu$, and Fock space is finite-dimensional, so the free group is compact. The Gaussian data is the covariance matrix of Eq.~\eqref{eq:fermions-covariance-definition}: by Wick's theorem the higher Majorana moments of a Gaussian state are Pfaffians of it, so it determines the state, and $\psi$ is pure Gaussian iff it is orthogonal, $\|\Gamma(\psi)\|_{\mathrm F}^2=2n$ (Eq.~\eqref{eq:fermions-covariance}). Matchgates are the nearest-neighbor Gaussian unitaries on the Jordan--Wigner line~\cite{terhal_divincenzo}.

\subsubsection{Bosons}
For $n$ bosonic modes, the creation and annihilation operators satisfy
\begin{align}
 [a_i,a_j]&=0, \qquad
 [a_i,a_j^\dagger]=\delta_{i,j}I.
 \label{eq:bosons-ccr}
\end{align}
These relations have no finite-dimensional representation (the trace of $[a,a^\dagger]$ would vanish, while that of $I$ does not), so $a_j,a_j^\dagger$ are unbounded. Fock space has the orthonormal number basis
\begin{align}
 \ket\alpha&=\prod_{j=1}^n\frac{(a_j^\dagger)^{\alpha_j}}{\sqrt{\alpha_j!}}\ket\Omega,
 \qquad \alpha\in\mathbb N^n,
\end{align}
where the vacuum $\ket\Omega$ is annihilated by every $a_j$; in each mode $a^\dagger\ket k=\sqrt{k+1}\ket{k+1}$ and $a\ket k=\sqrt k\ket{k-1}$, and $N=\sum_ja_j^\dagger a_j$ counts photons. The Gaussian unitaries are generated by Hamiltonians of degree at most two in the $a_j,a_j^\dagger$~\cite{weedbrook_gaussian}: the quadratic part squeezes and mixes modes, and the linear part displaces. The pure Gaussian states $\GB$ are the images of $\ket\Omega$ under Gaussian unitaries; the centered ones $\Gzero$ are those reached without the linear part, and the coherent states $\Coh$ are those reached with the linear part alone. Squeezing and displacement are both unbounded, and so is the photon number along the families, which is why no photon-number cutoff can be assumed (Section~\ref{sec:testing-implementation}). In phase-space language, with quadratures $a_j=(x_j+ip_j)/\sqrt2$, the Gaussian unitaries act on $(x,p)\in\mathbb R^{2n}$ affinely, by symplectic maps and translations, just as the fermionic ones rotate the Majorana span; this is where the groups of Table~\ref{tab:math-orbits} come from. Appendix~\ref{sec:bosons} works instead in the Bargmann representation, where a state is an entire function on $\mathbb C^n$, $a_j^\dagger$ acts as multiplication by $z_j$, and $a_j$ as $\partial_{z_j}$ (Eq.~\eqref{eq:boson-bargmann}); there the Gaussian states become Gaussian functions (Lemma~\ref{lem:boson-gaussian-coordinates}), and the tests become identities between functions.

\subsubsection{Fock space}\label{sec:math-fock}
A single particle lives in a Hilbert space $\mathfrak h$; the space in which any number of identical particles live is called \emph{Fock space}: 
\begin{equation}
    \mathcal F(\mathfrak h)=\bigoplus_{m=0}^\infty\mathfrak h_m
    \label{eq:fock-def},\qquad \mathfrak h_m\text{ is the $m$-particle Hilbert space.}
\end{equation}
Both fermionic and bosonic Fock space are built from the one-particle space $\mathbb C^n$ by second quantization; for $n$ modes (all fermionic or all bosonic), the Fock space is
\begin{align}
 \text{Fermions: }\bigoplus_{k=0}^{n}\Lambda^k\mathbb C^n,\qquad\text{Bosons: }\overline{\bigoplus_{k=0}^\infty\Sym^k\mathbb C^n}
\end{align}

where $\Lambda^k\mathbb C^n$ and $\Sym^k\mathbb C^n$ are the antisymmetric and symmetric subspaces of $(\mathbb C^n)^{\otimes k}$ (Section~\ref{sec:math-rep}), now (anti)symmetrizing particles instead of copies, and the bar is the Hilbert-space completion (the fermionic sum stops at $k=n$ by the Pauli principle, the bosonic one does not). A one-particle unitary $u$ acts on Fock space as $\Gamma(u)=\bigoplus_ku^{\otimes k}$, the map in Eq.~\eqref{eq:boson-group-average}; these are the passive unitaries, generated by the number-conserving quadratics $a_i^\dagger a_j$. Every test here commutes with total particle number on $\mathcal H^{\otimes q}$, so it acts sector by sector, and each sector is finite-dimensional (finitely many particles in finitely many modes); this is how the spectral statements of Section~\ref{sec:math-rep} apply even though $a$ and $a^\dagger$ are unbounded. For bosons, $q$ copies of $n$ modes are $nq$ modes, $\mathcal F_n^{\otimes q}=\bigoplus_k\Sym^k(\mathbb C^n\otimes\mathbb C^q)$, with the copy group acting on $\mathbb C^q$. The main difference is that fermions have finitely many sectors ($N\le n$) and bosons infinitely many, so every bosonic spectral bound must hold uniformly over all of them, as Lemma~\ref{lem:so3} does.

\paragraph{Grassmannians and pure spinors.}
The families are cut out by quadratic equations, which is why two copies suffice: each particle-number sector of $\Sl$ is a Grassmannian, cut out by the quadratic Pl\"ucker relations~\cite{oszmaniec_kus}, and the fermionic Gaussian states are Cartan's pure spinors, cut out by $\Lambda\ket\psi^{\otimes2}=0$~\cite{bravyi}. Quadratic equations in $\psi$ are linear in $\psi^{\otimes2}$ (Section~\ref{sec:math-geometry}).

\paragraph{Copy rotations.} 
The copy rotations behind the tests are the quantum version of Maxwell's characterization in the main text. If $X$ and $Y$ are independent with the same density $f$, their joint density $f(x)f(y)$ is unchanged by every rotation of the plane exactly when $f$ is a centered Gaussian. Convolution is one slice of this: the rotation by $\pi/4$ sends $X$ to $(X+Y)/\sqrt2$, so a rotation-invariant pair is in particular a fixed point of the self-convolution step that the central limit theorem iterates. Quantumly, the two copies are two registers, and rotating them is a beam splitter acting on every mode at once, $(a_j,b_j)\mapsto(a_j\cos\theta+b_j\sin\theta,\,-a_j\sin\theta+b_j\cos\theta)$. The pure states whose two copies are unchanged by every such rotation are exactly the centered Gaussian states (Lemma~\ref{lem:boson-centered-zero}); for fermions the rotation acts on the copy label of each Majorana operator, and the fixed states are the Gaussian ones (Lemma~\ref{lem:normal-fermion-fixed-spaces}, with Bravyi's criterion~\cite{bravyi} in the form of Lemma~\ref{lem:fermions-exact}). The convolution version, which keeps one output and traces out the other, fixes exactly the centered Gaussian states of finite energy by the quantum central limit theorem~\cite{becker_convergence_2021} (related characterizations through beam-splitter factorization are in~\cite{springer,cuesta}), and~\cite{coffman} use its fermionic analogue to quantify non-Gaussian magic. Our tests use the closely related pure-state condition that $\ket\psi^{\otimes2}$ itself be fixed, which needs no partial trace and, by Lemma~\ref{lem:boson-centered-zero}, no assumption on the moments. Displacements are where the third copy comes in: a displacement shifts every copy by the same amount, and a rotation of two copies does not keep the shifts $(1,1)$ equal, so displaced states fail the two-copy rotation test; on three copies, the rotations about the axis $(1,1,1)$ do keep an equal shift, and averaging over them is the test of Section~\ref{subsec:boson-tetrahedron}.

\newpage

\section{General proof strategy}\label{sec:preliminaries}
This appendix develops the common framework used to turn exact
characterizations into robust tests. The gradient-flow argument
gives global bounds relating rejection probability to distance
from the target family. Expansions around closest targets
determine how sensitively the tests detect small deviations.
Together, these results give tolerant-testing guarantees, which
we compare with lower bounds on the number of copies required.

\begin{enumerate}
\item \textbf{Section~\ref{sec:first-order}: Determine how rejection
changes with the input.}
The tensor-power map connects changes of the one-copy input to
changes of the state measured by the test. We compute the gradient
of rejection as a one-copy vector
(Lemma~\ref{lem:first-variation}). This identifies the direction
in which the input will move in the gradient-flow argument.

\item \textbf{Section~\ref{sec:normals}: Analyze small deviations
from a closest target.}
Choosing a closest target restricts the possible deviations from
it (Lemma~\ref{lem:normal-stationarity}). We expand the tensor
power and control the remaining terms to obtain a local rejection
bound (Lemma~\ref{lem:normal-bound}). Examples with known distance
from the target family show that the leading coefficients are
sharp (Corollary~\ref{cor:normal-distance-sharpness}). These local
estimates will be used for tolerant testing.

\item \textbf{Section~\ref{sec:restoration}: Turn
overlapping-projector estimates into global distance bounds.}
A spectral condition on overlapping acceptance projectors
(Eq.~\eqref{eq:overlap-spectrum}) gives a lower bound on the norm
of the rejection gradient (Theorem~\ref{thm:overlap}). For an input
with sufficiently small rejection, the gradient flow then
converges to a perfectly accepted state along a path of bounded
length (Theorem~\ref{thm:restoration}). When the test exactly
characterizes the target family, this one-copy path bounds the
input's distance from that family
(Eq.~\eqref{eq:restoration-angle-bound}). We also prove the
rotation-projector estimate used in the Gaussian applications
(Lemma~\ref{lem:so3}).

\item \textbf{Section~\ref{sec:higher-copy-tests}: Construct and
analyze higher-copy tests.}
We define tests acting jointly on more copies
(Definition~\ref{def:normal-projectors}) and show how copy
rotations determine their response to small deviations
(Lemma~\ref{lem:normal-transition}). Comparing these tests with
tests on smaller groups of copies transfers the basic global
distance bounds to the higher-copy setting
(Theorem~\ref{thm:normal-hierarchy}). The fermionic and bosonic
appendices provide the corresponding coefficient calculations
(Sections~\ref{subsec:fermion-higher-response}
and~\ref{subsec:boson-higher-response}), summarized in
Table~\ref{table:normal-coefficients}.

\item \textbf{Section~\ref{sec:testing}: Convert rejection bounds
into tolerant-testing guarantees.}
We bound rejection from above for inputs close to the target
family and from below for inputs farther away. When these bounds
are separated, repeated measurements distinguish the two cases
(Proposition~\ref{prop:testing-statistics}). Combining the local
and global estimates gives mode-independent copy bounds for
closely spaced distance thresholds in a specified small-distance
regime (Theorem~\ref{thm:testing-optimal-gap}). We also give the
separation conditions available from the basic few-copy tests
(Proposition~\ref{prop:testing-baseline-radii}).

\item \textbf{Section~\ref{sec:testing-lower-bounds}: Prove lower
bounds on the required number of copies.}
We choose a near input and a far input whose target distances
are known exactly but whose overlap is large
(Lemma~\ref{lem:testing-exact-curves}). Distinguishing these two
states already requires many copies. This gives lower bounds
for any testing procedure
(Theorem~\ref{thm:testing-lower-bound}). Together with our upper
bounds, these establish the optimal dependence on the distance
gap and error probability in the regimes covered by those bounds.

\item \textbf{Section~\ref{sec:testing-implementation}: Account for
measurement inaccuracies.}
We bound how errors in the implemented acceptance operator change
rejection probabilities. This determines how accurately the
measurement must be implemented to preserve the separation
between near and far inputs, and how many repetitions are then
sufficient (Proposition~\ref{prop:testing-measurement-error}).
\end{enumerate}

\subsection{First-order changes}
\label{sec:first-order}

We first study how the rejection probability changes when the
input state varies. This will identify the direction in which
rejection decreases most rapidly, providing the starting point
for the gradient-flow argument in Section~\ref{sec:restoration}.
The test acts on several copies, while the state we vary belongs
to the one-copy Hilbert space $\mathcal H$. The tensor-power map
connects these two descriptions:
\begin{align*}
 \ket\psi
 &\longmapsto \ket\psi^{\otimes q}
 \in \Sym^q(\mathcal H)
 \subseteq \mathcal H^{\otimes q}.
\end{align*}
Here $\Sym^q(\mathcal H)$ consists of vectors invariant under
permutations of the copies. Identical copies always lie in
this subspace.

Let $P$ be the orthogonal projector describing acceptance,
and set $Q=I-P$. For a normalized input, the acceptance and
rejection probabilities are
\begin{align*}
 p(\psi)
 &=\|P\ket\psi^{\otimes q}\|_2^2,
 &
 R(\psi)
 &=\|Q\ket\psi^{\otimes q}\|_2^2
 =1-p(\psi).
\end{align*}
Perfect acceptance means that the tensor power lies entirely
in the linear accepted space $\Ran P$.
For a test that exactly characterizes a target family
$\mathcal T$,
\begin{align*}
 R(\psi)=0
 \quad\Longleftrightarrow\quad
 \ket\psi^{\otimes q}\in\Ran P
 \quad\Longleftrightarrow\quad
 \ket\psi\in\mathcal T.
\end{align*}
The first equivalence follows from the projector definition;
the second must be proved for the particular test.
Since the tensor-power map is polynomial, the condition
$Q\ket\psi^{\otimes q}=0$ gives homogeneous polynomial
constraints of degree $q$ on the input. Exactness says that
their normalized solutions are precisely the target states.

Robustness requires a quantitative version of this
characterization. Two distances matter:
\begin{align*}
 \sqrt{R(\psi)}
 &=\inf_{\ket\xi\in\Ran P}
   \|\ket\psi^{\otimes q}-\ket\xi\|_2,\\
 d_{\mathcal T}(\psi)
 &=\inf_{\ket g\in\mathcal T}
   \sqrt{1-|\ip g\psi|^2}.
\end{align*}
The first is the distance from the tensor power to the accepted
linear subspace, where the comparison vector $\ket\xi$ need
not be normalized. The second is the trace distance from the
original input to the target family. Our goal is to show that
small rejection forces the second distance to be small.
Exactness alone does not provide this quantitative bound.

The nearest vector in the accepted subspace is
$P\ket\psi^{\otimes q}$, but it need not itself be a tensor
power. We therefore vary the one-copy input and study how its rejection changes. The gradient estimate established later will allow us to follow these changes along a short path to an exactly accepted input, which exactness identifies as a target state. The estimate bounds the
length of the path from $\ket\psi$ to $\ket g$ in the one-copy
space. Although $\ket g$ need not be a closest target, this
path length bounds the input's distance to the target family.

\begin{figure}[t]
\centering
\includegraphics[width=0.97\textwidth]
  {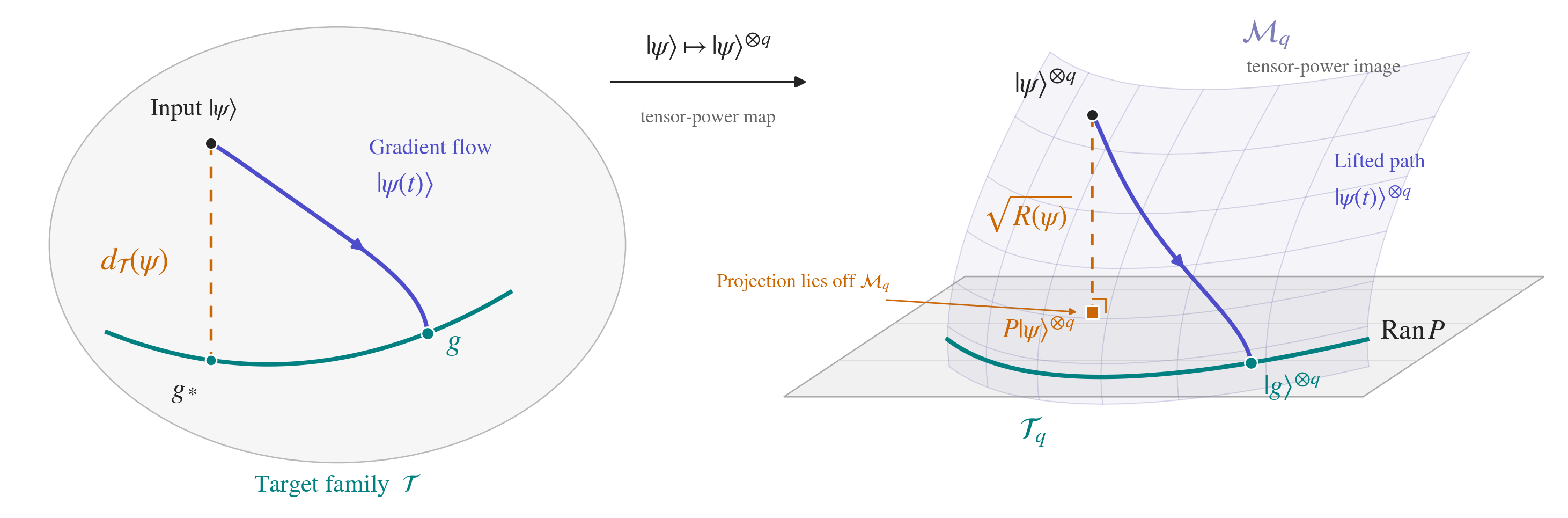}

\caption{\textbf{From small rejection to a nearby target.}
Left: a normalized one-copy input $\ket\psi$ and the target
family $\mathcal T\subseteq\mathcal H$.
Right: normalized tensor powers form the nonlinear set
$\mathcal M_q=\{\ket\phi^{\otimes q}:
\|\ket\phi\|_2=1\}$ inside $\Sym^q(\mathcal H)$.
The gray plane represents the accepted linear subspace
$\Ran P$.
Exactness identifies its intersection with $\mathcal M_q$
as the target tensor powers:
$\mathcal T_q=\{\ket g^{\otimes q}:\ket g\in\mathcal T\}
=\mathcal M_q\cap\Ran P$.
The orange dashed lines indicate the two distances being
compared: the one-copy trace distance
$d_{\mathcal T}(\psi)=\inf_{\ket g\in\mathcal T}
\sqrt{1-|\ip g\psi|^2}$,
and the Hilbert-space distance
$\sqrt{R(\psi)}=\|\ket\psi^{\otimes q}
-P\ket\psi^{\otimes q}\|_2$
to the accepted subspace.
The projection $P\ket\psi^{\otimes q}$ need not itself
be a tensor power.
The blue curves show the one-copy gradient flow and its
tensor-power image, which remains in $\mathcal M_q$.
When the gradient estimate holds,
the flow converges to a target $\ket g$ along a path
whose one-copy length is controlled by the initial rejection.
This length bounds $d_{\mathcal T}(\psi)$, although
$\ket g$ need not be the closest target $\ket{g_*}$.
The geometry and distances are schematic, not to scale.}
\label{fig:tensor-power-geometry}
\end{figure}
For the first-variation calculation below, assume that $P$
commutes with copy permutations. It then preserves
$\Sym^q(\mathcal H)$. For fixed $q$, the tensor-power map is
a continuous polynomial, and $R$ is smooth. The product rule
gives
\begin{align}
 D(\psi^{\otimes q})[v]
 &:=\left.\frac{\dd}{\dd t}
   (\ket\psi+t\ket v)^{\otimes q}\right|_{t=0}
 \label{eq:tensor-directional-derivative}\\
 &=\sum_{a=1}^q
   \ket\psi^{\otimes(a-1)}
   \otimes\ket v
   \otimes\ket\psi^{\otimes(q-a)}.
 \label{eq:tensor-derivative}
\end{align}
Each summand changes one copy in the direction $\ket v$,
leaving the other copies unchanged. First-order changes of
a normalized input satisfy
$\Re\ip\psi v=0$;
the stronger condition $\ip\psi v=0$
also removes the global-phase direction.

To express the change in acceptance using a one-copy vector,
take the overlap of $P\ket\psi^{\otimes q}$ with $\ket\psi$
in the first $q-1$ copies. The resulting vector has overlap
$p(\psi)$ with $\ket\psi$. Subtracting this parallel component
defines
\begin{subequations}
\label{eq:restoration-contraction}
\begin{align}
 \ket{\chi_\psi}
 &=(\bra\psi^{\otimes(q-1)}\otimes I)
   P\ket\psi^{\otimes q}
   -p(\psi)\ket\psi \nonumber\\
 &=\Tr_{1,\ldots,q-1}
   \!\left[P\rho_\psi^{\otimes q}\right]\ket\psi
   -p(\psi)\ket\psi,
 \label{eq:ketchi_psi}
\end{align}
where $\rho_\psi=\ket\psi\bra\psi$, and the partial trace
leaves an operator on the remaining copy. Consequently,
\begin{align}
 \ip\psi{\chi_\psi}
 &=0, \nonumber\\
 \|p(\psi)\ket\psi+\ket{\chi_\psi}\|_2^2
 &=p(\psi)^2+\|\ket{\chi_\psi}\|_2^2.
\end{align}
\end{subequations}
The correction $\ket{\chi_\psi}$ is not necessarily normalized
and may vanish. The following lemma shows that it determines
the gradient of rejection.

\begin{lemma}[First variation of the rejection probability]
\label{lem:first-variation}
Suppose that the acceptance projector $P$ commutes with
copy permutations.
On the one-copy unit sphere, equipped with the real metric
$\Re\ip uv$, the gradient of the rejection probability is
\begin{align}
 \nabla R(\psi)&=-2q\ket{\chi_\psi}.
\end{align}
\end{lemma}

\begin{proof}
The directional derivative $Dp(\psi)[v]$ measures how the acceptance
probability changes when the normalized input state moves away from
$\ket\psi$ with initial velocity $\ket v$. Choose a smooth curve of
unit vectors satisfying $\ket{\psi(0)}=\ket\psi$ and
$\ket{\dot\psi(0)}=\ket v$, and define
\begin{align}
 Dp(\psi)[v]
 &:=\left.\frac{\dd}{\dd t}p(\psi(t))\right|_{t=0}.
 \label{eq:acceptance-directional-derivative}
\end{align}
Equivalently, for small $t$
\begin{align}
     p(\psi(t))=p(\psi)+t\,Dp(\psi)[v]+O(t^2).
\end{align}
Indeed, this first-order change depends only on the initial state $\ket\psi$
and velocity $\ket v$, not on the particular curve chosen.

We can compute the directional derivative above by just writing the acceptance probability $p(\psi)$ in Eq.~\eqref{eq:acceptance} explicitly as an inner product. Indeed, using the derivative of the product rule we immediately see that 
\begin{subequations}\label{eq:acceptance-first-variation}
\begin{align}
 Dp(\psi)[v]
 =&2\Re\ip{P\psi^{\otimes q}}{D(\psi^{\otimes q})[v]} \nonumber\\
 =&2q\Re\ip{P\psi^{\otimes q}}{\psi^{\otimes(q-1)}\otimes v
   } \nonumber\\
 =&2q\Re\ip{p(\psi)\psi+\chi_\psi}v \nonumber\\
 =&2q\Re\ip{\chi_\psi}v.
\end{align}
\end{subequations}
In the second equality we substituted Eq.~\eqref{eq:tensor-derivative} and use the permutation invariance property of $P$.   In the third equality we used Eq.~\eqref{eq:restoration-contraction}, and in the last equality we use $\Re\ip\psi v=0$, which follows from normalization of $\ket{\psi(t)}$. The gradient $\nabla R$ describes how the rejection probability
responds to a small change of the normalized input state.
For any variation $\ket v$ satisfying $\Re\ip\psi v=0$,
which preserves normalization to first order, it satisfies 
\begin{align}
 DR(\psi)[v]&=\Re\ip{\nabla R}v.
\end{align}
On the other hand, since $R=1-p$,
Eq.~\eqref{eq:acceptance-first-variation} gives
\begin{align}
 DR(\psi)[v]
 &=-Dp(\psi)[v] \nonumber\\
 &=-2q\Re\ip{\chi_\psi}v \nonumber\\
 &=\Re\ip{-2q\chi_\psi}v.
\end{align}
Thus the real overlap of $-2q\ket{\chi_\psi}$ with any allowed
variation $\ket v$ gives the first-order change in rejection
probability. Moreover, $\ip\psi{\chi_\psi}=0$ by
Eq.~\eqref{eq:restoration-contraction}, so changing the input
in this direction preserves normalization to first order. Consequently $\nabla R(\psi) = -2q\ket \chi_\psi$.

\end{proof}

Lemma~\ref{lem:first-variation} describes how rejection changes
under small variations of the input. We next seek a lower bound
on rejection in terms of the distance to the target family.
To do so, we expand the input around a closest target state
and use the restrictions imposed by this choice to determine
which components of the deviation are detected by the test.

\subsection{Expansion around a closest target}\label{sec:normals}

By expanding our state around the closest target in $\mathcal{T}$, we can obtain several important properties. First, it allows us to impose restrictions on some properties of the state. Indeed, we will see in Lemma~\ref{lem:normal-stationarity} that the deviation from that target cannot be arbitrary. Not only that, but we use this analysis in Lemma~\ref{lem:normal-bound}, to find a local lower-bound on the rejection probability of the tester. In both cases we need to use specific properties of the groups we are testing on, however, the generality of our proof method allows us to construct a general proof strategy that we can later just ``plug'' the necessary properties at the last step. Finally, Corollary~\ref{cor:normal-distance-sharpness} shows that the leading coefficients of these local lower bounds are optimal uniformly over the number of modes, by identifying input states that attain them.

Before we proceed, there could be concerns on wether the closest target even exists. When the target state is a fermionic state, this is ensured by the compactness of the group. For bosons, it follows from Lemma~\ref{lem:normal-attainment}. Now that we have clarified this concern, we are ready to present our state $\psi$ and its closest target $g$, that without loss of generality we can choose to have a positive (and thus real) with the input state, i.e. $0<d_{\mathcal T}(\psi)<1$, 
\begin{align}
 \ket\psi
 &=\sqrt{1-d_{\mathcal T}^2}\,\ket g+d_{\mathcal T}\ket\eta,
 \label{eq:normal-frame}\\
 \ip g\eta&=0,\qquad \|\ket\eta\|_2=1.
 \label{eq:normal-frame-normalization}
\end{align}

The same maximization argument is also used in the local bosonic analysis of \cite[App.~B.3]{girardi}. Table~\ref{table:normal-allowed-degrees} below summarizes the restrictions on $\ket\eta$ that we obtain by applying Lemma~\ref{lem:normal-stationarity} to each target family. In each case, we first transform the closest target to the reference state shown, applying the same transformation to the input and the deviation. The final column gives the justification for this transformation; the excitation restrictions follow from the orthogonality argument in the lemma below.

\begin{table}[ht]
\centering
\small
\caption{Allowed excitation components after transforming the closest target to the listed reference state and applying the same transformation to the input. The vacuum $\ket\Omega$ has no particles in any mode. The Slater determinant $\ket{e_1\wedge\cdots\wedge e_k}$ is the normalized antisymmetric product of $k$ orthonormal occupied orbitals, relative to which replacements are counted. The target families are pure fermionic Gaussian states with fixed parity $\sigma\in\{+,-\}$ ($\GFs$), pure fermionic Gaussian states of either parity ($\GF$), Slater determinants allowing all particle numbers ($\Sl$), coherent states ($\Coh$), centered pure bosonic Gaussian states with zero displacement ($\Gzero$), and all pure bosonic Gaussian states including arbitrary displacement ($\GB$). For odd-parity fermionic Gaussian targets, the reduction to vacuum includes a parity-changing unitary.}
\begin{tabularx}{\linewidth}{@{}l @{\hspace{2em}} >{\raggedright\arraybackslash}p{0.14\linewidth} >{\raggedright\arraybackslash}X >{\raggedright\arraybackslash}p{0.16\linewidth}@{}}
\toprule
Target family & Reference state & Allowed components of $\ket\eta$ & Reduction proved in \\
\midrule
$\GFs$ & $\ket\Omega$ & Even particle degrees $k\ge4$ & Lemma~\ref{lem:fermion-closest} \\
$\GF$ & $\ket\Omega$ & Particle degrees $k\ne0,2$ & Proof of Lemma~\ref{lem:fermion-closest} \\
$\Sl$ & $\ket{e_1\wedge\cdots\wedge e_k}$ & At least two orbital replacements, or particle number different from $k$ & Lemma~\ref{lem:fermion-closest} \\
$\Coh$ & $\ket\Omega$ & Photon degrees $k\ge2$ & Lemma~\ref{lem:boson-frames} \\
$\Gzero$ & $\ket\Omega$ & Photon degrees $k\ne0,2$ & Lemma~\ref{lem:boson-frames} \\
$\GB$ & $\ket\Omega$ & Photon degrees $k\ge3$ & Lemma~\ref{lem:boson-frames} \\
\bottomrule
\end{tabularx}
\label{table:normal-allowed-degrees}
\end{table}

We now prove the common statement behind these restrictions: choosing a closest target forces the deviation to be orthogonal to the allowed first-order changes within the target family.

\begin{lemma}[Orthogonality at a closest target]\label{lem:normal-stationarity}
Let $\ket g$ be a closest target to $\ket\psi$, and let $\ket\eta$ be the deviation in Eq.~\eqref{eq:normal-frame}, with $0<d_{\mathcal T}<1$. Consider normalized target states $\ket{g(t)}\in\mathcal T$ close to $\ket g$. For a small real parameter $t$, write
\begin{align}
 \ket{g(t)}&=\ket g+t\ket v+O(t^2),
 \qquad \ip gv=0,
\end{align}
where the last condition is obtained by choosing the $t$-dependent global phase of $\ket{g(t)}$. Suppose that every allowed change $\ket v$ can also be made with $\ket v$ replaced by $i\ket v$. Then the deviation $\ket{\eta}$ in $\ket{\psi}$ cannot have components of $\ket{v}$, i.e. 
\begin{align}
    \ip v\eta=0
\end{align}
for every such $\ket v$.
\end{lemma}

\begin{proof}
Let $\ket{g(t)}$ be a smoothly varying normalized target state, where $t$ is a real parameter, and let $\ket v$ be its derivative at zero. Normalization makes the real part of $\ip gv$ vanish. We can choose the $t$-dependent global phase of $\ket{g(t)}$ so that the imaginary part also vanishes. Thus
\begin{align}
 \ket{g(0)}&=\ket g,\qquad
 \left.\frac{\dd}{\dd t}\ket{g(t)}\right|_{t=0}=\ket v,
 \label{eq:normal-target-curve}\\
 \ip gv&=0.
 \label{eq:normal-horizontal-velocity}
\end{align}
Since $\ket g$ is a closest target, the squared overlap has a maximum at zero, so its derivative vanishes. The product rule and Eq.~\eqref{eq:normal-frame} give
\begin{align}
 0&=\left.\frac{\dd}{\dd t}|\ip{g(t)}\psi|^2\right|_{t=0}
 \label{eq:normal-stationarity-derivative}\\
 &=2\Re\!\left(\overline{\ip g\psi}\,\ip v\psi\right)
 \nonumber\\
 &=2d_{\mathcal T}\sqrt{1-d_{\mathcal T}^2}\,\Re\ip v\eta.
 \label{eq:normal-stationarity-real-part}
\end{align}
The prefactor is positive, so the real part of the overlap vanishes. By assumption, we can also vary the target so that its derivative at zero is $i\ket v$. Repeating the overlap calculation for this variation gives
\begin{align}
 \Re\ip v\eta&=0,\\
 \Im\ip v\eta&=\Re\ip{iv}\eta=0,
\end{align}
which means that
\begin{align}
 \ip v\eta&=0.
 \label{eq:normal-stationarity-orthogonality}
\end{align}
\end{proof}

Now that Lemma~\ref{lem:normal-stationarity} has identified the possible deviations from a closest target, we analyze how the tester detects them. Suppose the test accepts $\ket g$ perfectly. For a small real parameter $t$, consider the normalized state $(\ket g+t\ket\eta)/\sqrt{1+t^2}$. Expanding its tensor power as in Eq.~\eqref{eq:tensor-derivative} gives a rejection probability $R$ that looks like
\begin{align}
 R\!\left(\frac{\ket g+t\ket\eta}{\sqrt{1+t^2}}\right)
 &=t^2\left\|Q\left(\sum_{a=1}^q
 \ket g^{\otimes(a-1)}\otimes\ket\eta\otimes\ket g^{\otimes(q-a)}
 \right)\right\|_2^2+O(t^3).
 \label{eq:normal-W}
\end{align}
The coefficient of $t^2$ measures how quickly rejection increases as we add a small component along $\ket\eta$. It is determined by the rejected part of the sum containing exactly one 
copy of $\ket\eta$.

For a state at distance $d_{\mathcal T}$, the tensor expansion contains terms with two or more copies of $\ket\eta$. Let $\ket w$ denote the sum of these terms, including their expansion coefficients. Since $\ip g\eta=0$, terms with $\ket\eta$ in different sets of copy positions are orthogonal, so
\begin{align}
 \|\ket w\|_2^2
 &=\sum_{r=2}^q\binom qr
 d_{\mathcal T}^{2r}(1-d_{\mathcal T}^2)^{q-r}
 \label{eq:normal-remainder-norm}\\
 &=1-(1-d_{\mathcal T}^2)^q
 -qd_{\mathcal T}^2(1-d_{\mathcal T}^2)^{q-1}.
 \label{eq:normal-B}
\end{align}
In particular, this norm 
depends only on $q$ and $d_{\mathcal T}$. After applying $Q$, the remainder can partly cancel the contribution from the terms containing exactly one copy of $\ket\eta$. To account for this, let $0<\lambda\le q$ be a lower bound on the coefficient of $t^2$ in Eq.~\eqref{eq:normal-W}, and define
\begin{align}
 h_{q,\lambda}(d_{\mathcal T}^2)
 &=\left[
 \sqrt{\lambda}\,d_{\mathcal T}(1-d_{\mathcal T}^2)^{(q-1)/2}
 -\sqrt{1-\lambda/q}\,\|\ket w\|_2
 \right]_+^2,
 \label{eq:normal-h}\\
 [r]_+&=\max\{r,0\}.
 \label{eq:normal-positive-part}
\end{align}

The next lemma shows that $h_{q,\lambda}(d_{\mathcal T}^2)$ is a lower bound on the rejection probability, including the possible cancellation from the remainder. The argument applies to any valid coefficient bound $\lambda$. Table~\ref{table:normal-coefficients} lists the optimal lower-bounds $\lambda_{\mathcal T,q}$. Their derivations are deferred to the family-specific sections cited in the final column.

\begin{table}[ht]
\centering
\small
\renewcommand{\arraystretch}{1.6}
\caption{Quadratic rejection coefficients for pure fermionic Gaussian states with fixed parity $\sigma\in\{+,-\}$ ($\GFs$), pure fermionic Gaussian states of either parity ($\GF$), Slater determinants allowing all particle numbers ($\Sl$), coherent states ($\Coh$), centered pure bosonic Gaussian states ($\Gzero$), and all pure bosonic Gaussian states including arbitrary displacement ($\GB$). For $\GFs$, the input has the specified parity. These coefficients are optimal uniformly over the number of modes; smaller fermionic sectors can admit larger coefficients.}
\label{table:normal-coefficients}
\begin{tabularx}{\linewidth}{@{}l *{6}{>{\centering\arraybackslash}X}@{}}
\toprule
Target family & $\GFs$ & $\GF$ & $\Sl$ & $\Coh$ & $\Gzero$ & $\GB$ \\
\midrule
$\lambda_{\mathcal T,q}$
& $\dfrac{q(q-1)}{q+2}$
& $\dfrac{q(q-1)}{q+2}$
& $\dfrac{q(q-1)}{q+1}$
& $q-1$
& $\dfrac{q(q-1)}{q+2}$
& $\dfrac{(q-1)(q-2)}q$ \\
Copies
& $q\ge2$
& $q\ge2$
& $q\ge2$
& $q\ge2$
& $q\ge2$
& $q\ge3$ \\
Proved in
& Section~\ref{subsec:fermion-higher-response}
& Section~\ref{subsec:fermion-higher-response}
& Section~\ref{subsec:fermion-higher-response}
& Section~\ref{subsec:boson-higher-response}
& Section~\ref{subsec:boson-higher-response}
& Section~\ref{subsec:boson-higher-response} \\
\bottomrule
\end{tabularx}
\end{table}

Next we introduce the lemma. 

\begin{lemma}[Rejection bound from a closest-target expansion]\label{lem:normal-bound}
Let $P$ be an orthogonal projector satisfying
\begin{align}
 P\ket g^{\otimes q}&=\ket g^{\otimes q},
 \label{eq:normal-bound-completeness}\\
 \bigl\|Q\bigl(D(g^{\otimes q})[\eta]\bigr)\bigr\|_2^2
 &\ge\lambda>0.
 \label{eq:normal-bound-coefficient}
\end{align}
Then the state in Eq.~\eqref{eq:normal-frame} satisfies
\begin{align}
 R(\psi)&\ge h_{q,\lambda}(d_{\mathcal T}^2).
 \label{eq:normal-bound-conclusion}
\end{align}
\end{lemma}
\begin{proof}
Define the normalized first-order change
\begin{align}
 \ket{\widetilde\eta}
 &:=\frac{1}{\sqrt q}D(g^{\otimes q})[\eta]
 \nonumber\\
 &=\frac{1}{\sqrt q}\sum_{a=1}^q
 \ket g^{\otimes(a-1)}\otimes\ket\eta\otimes\ket g^{\otimes(q-a)}.
 \label{eq:normal-single-replacement}
\end{align}
Separating the terms containing zero, one, and at least two copies of $\ket\eta$ gives
\begin{align}
 \ket\psi^{\otimes q}
 &=(1-d_{\mathcal T}^2)^{q/2}\ket g^{\otimes q}
 +\sqrt q\,d_{\mathcal T}(1-d_{\mathcal T}^2)^{(q-1)/2}\ket{\widetilde\eta}
 +\ket w,
 \label{eq:normal-tensor-expansion}\\
 \|\ket{\widetilde\eta}\|_2&=1,
 \qquad \ip{\widetilde\eta}w=0.
 \label{eq:normal-remainder-orthogonality}
\end{align}
The remainder $\ket w$ contains all terms with at least two tensor factors equal to $\ket\eta$. Terms with $\ket\eta$ in different sets of copy positions are orthogonal because $\ip g\eta=0$. Its squared norm therefore equals one minus the contributions from terms containing zero or one copy of $\ket\eta$. Normalize the rejected component of $\ket{\widetilde\eta}$ by setting
\begin{align}
 s&=\|Q\ket{\widetilde\eta}\|_2^2,
 \qquad
 \ket z=\frac{Q\ket{\widetilde\eta}}{\sqrt{s}}.
 \label{eq:normal-rejected-direction}
\end{align}
By Eqs.~\eqref{eq:normal-single-replacement} and \eqref{eq:normal-bound-coefficient}, we have $s\ge\lambda/q>0$. Using the tensor expansion in Eq.~\eqref{eq:normal-tensor-expansion}, the rejection amplitude satisfies
\begin{align}
 \|Q\ket\psi^{\otimes q}\|_2
 &\ge\Re\ip z{Q\psi^{\otimes q}}
 \label{eq:normal-rejection-first-bound}\\
 &=\sqrt q\,d_{\mathcal T}(1-d_{\mathcal T}^2)^{(q-1)/2}\sqrt{s}
 +\Re\ip zw
 \nonumber\\
 &\ge\sqrt q\,d_{\mathcal T}(1-d_{\mathcal T}^2)^{(q-1)/2}\sqrt{s}
 -\sqrt{1-s}\,\|\ket w\|_2
 \nonumber\\
 &\ge\sqrt{\lambda}\,d_{\mathcal T}(1-d_{\mathcal T}^2)^{(q-1)/2}
 -\sqrt{1-\lambda/q}\,\|\ket w\|_2.
 \label{eq:normal-rejection-final-bound}
\end{align}
The first inequality is Cauchy--Schwarz. The first equality uses $Q\ket g^{\otimes q}=0$ and $Q\ket z=\ket z$. For the next inequality, $\ip{\widetilde\eta}w=0$, so only the component of $\ket z$ orthogonal to $\ket{\widetilde\eta}$ contributes to $\ip zw$. Since $\ip{\widetilde\eta}z=\sqrt{s}$, that component has norm $\sqrt{1-s}$. The final inequality uses $s\ge\lambda/q$ since $s=\frac1q\left\|Q\bigl(D(g^{\otimes q}[\eta]\bigr)\right\|_2^2\ge\frac{\lambda}{q}$.

The norm on the left is nonnegative, so we may replace the final right-hand side by its positive part. Squaring and using Eq.~\eqref{eq:normal-h} gives $R(\psi)\ge h_{q,\lambda}(d_{\mathcal T}^2)$.
\end{proof}

For fixed $q$ and $0<\lambda\le q$, the contribution containing exactly one copy of $\ket\eta$ has amplitude proportional to $d_{\mathcal T}$, while the remainder $\ket w$ has norm of order $d_{\mathcal T}^2$. Their possible cancellation therefore changes the squared amplitude only at order $d_{\mathcal T}^3$ or higher. Expanding Eq.~\eqref{eq:normal-h} gives
\begin{align}
h_{q,\lambda}(d_{\mathcal T}^2)
&=\lambda d_{\mathcal T}^2
-2\sqrt{\lambda(1-\lambda/q)\binom q2}\,d_{\mathcal T}^3
+O_q(d_{\mathcal T}^4).
\label{eq:normal-h-expansion}
\end{align}
The correction terms vanish faster than $d_{\mathcal T}^2$, so
\begin{align}
\frac{h_{q,\lambda}(d_{\mathcal T}^2)}
{\lambda d_{\mathcal T}^2}
&\longrightarrow 1
\qquad\text{as }d_{\mathcal T}\to0.
\nonumber
\end{align}

Thus, sufficiently close to the target family, Lemma~\ref{lem:normal-bound} gives a rejection lower bound with leading term $\lambda d_{\mathcal T}^2$. For the coefficients $\lambda_{\mathcal T,q}$ in Table~\ref{table:normal-coefficients}, the following corollary shows that suitable input states attain this leading behavior. These examples establish that the coefficients cannot be increased uniformly over the number of modes.

\begin{corollary}[Sharpness of the rejection coefficients]
\label{cor:normal-distance-sharpness}
For each family, whenever the required modes and particle-number sector are available, there are states $\ket{\psi_t}$ satisfying
\begin{align}
 d_{\mathcal T}(\psi_t)&=t, \nonumber \\
 R_{\mathcal T,q}(\psi_t)&=\lambda_{\mathcal T,q}t^2+O_q(t^3)
 \qquad(t\to0^+).
 \label{eq:higher-copy-distance-sharpness}
\end{align}
Thus the constants in Table~\ref{table:normal-coefficients} cannot be increased uniformly over the number of modes.
\end{corollary}
\begin{proof}
The examples constructed later in Appendices~\ref{subsec:fermion-distance-curves} and \ref{subsec:boson-distance-curves} have the form
\begin{align}
 \ket{\psi_t}&=\sqrt{1-t^2}\ket g+t\ket\eta,
 \qquad \ip g\eta=0,\quad\|\ket\eta\|_2=1.
 \label{eq:higher-copy-sharpness-states}
\end{align}
They use four fermions above a Gaussian vacuum, two occupied-orbital replacements for Slaters, four photons for centered Gaussians, two photons for coherent states, and three photons for full bosonic Gaussians. Those calculations prove that $g$ remains a closest target for sufficiently small $t\ge0$, and that Eq.~\eqref{eq:higher-copy-coefficient-bound} is attained with equality. The first fact gives $d_{\mathcal T}(\psi_t)=t$ directly from Eq.~\eqref{eq:distances}. Expanding the tensor power in Eq.~\eqref{eq:higher-copy-sharpness-states} gives
\begin{align}
 (I-P_{\mathcal T,q})\ket{\psi_t}^{\otimes q}
 &=t(I-P_{\mathcal T,q})\bigl(D(g^{\otimes q})[\eta]\bigr)+O_q(t^2), \nonumber \\
 R_{\mathcal T,q}(\psi_t)
 &=t^2\left\|(I-P_{\mathcal T,q})
       \bigl(D(g^{\otimes q})[\eta]\bigr)\right\|_2^2+O_q(t^3) \nonumber \\
 &=\lambda_{\mathcal T,q}t^2+O_q(t^3).
 \label{eq:higher-copy-sharpness-expansion}
\end{align}
The constant term vanishes by perfect completeness. The first remainder is bounded in the 2-norm, and its cross term with the leading vector is $O_q(t^3)$ after squaring. This proves Eq.~\eqref{eq:higher-copy-distance-sharpness}. The fermionic examples require sufficiently many occupied and unoccupied modes. No equality claim is made when the specified deviation is unavailable.
\end{proof}

In Subsection~\ref{sec:restoration}, we develop another way to relate rejection to distance. We continuously change the input in a direction that decreases rejection and use a bound on the gradient to control the total distance travelled before reaching a perfectly accepted state.

\subsection{Gradient flow}\label{sec:restoration}

In this section we develop a method for obtaining lower bounds on the rejection probability. For a test whose perfectly accepted states are exactly the target family $\mathcal T$, the spectral condition introduced below will imply
\begin{align}
R(\psi)&\ge
\begin{cases}
\kappa\sin^2\!\bigl(q\theta_{\mathcal T}(\psi)\bigr),
&\theta_{\mathcal T}(\psi)\le\pi/(2q),\\
\kappa,
&\theta_{\mathcal T}(\psi)\ge\pi/(2q).
\end{cases}
\nonumber
\end{align}
Here $\theta_{\mathcal T}(\psi)=\arcsin d_{\mathcal T}(\psi)$ is the angular distance to the target family, and $0<\kappa\le1$ is determined by the spectral estimate. This bound guarantees a positive rejection probability for inputs at positive distance from the target family.

We first prove Theorem~\ref{thm:overlap}, which derives a lower bound on the norm of the gradient from a spectral condition on overlapping acceptance projectors. We then apply this gradient bound in Theorem~\ref{thm:restoration}. Starting from an input with sufficiently small rejection, we continuously vary the normalized state in a direction that decreases rejection. The gradient bound controls the total distance travelled before convergence to a perfectly accepted state, yielding the rejection bound above. Finally, Lemma~\ref{lem:so3} establishes the required spectral estimate for projectors associated with rotations about orthogonal axes.

The idea is to continuously vary the normalized input in a direction that decreases rejection. For an input with sufficiently small rejection, the gradient bound controls the total distance travelled before the state converges to a perfectly accepted state. When the test is exact, this limiting state belongs to the target family. The resulting bound on the distance to the target family then gives a lower bound on rejection.

To formulate the spectral condition, let $P$ be a permutation-invariant acceptance projector. Choose $m\ge2$ sets of copy labels $B_i$, each containing $q$ copies, with exactly one shared copy between any two distinct sets. Let $P_{B_i}$ apply $P$ to the copies in $B_i$ and the identity to all remaining copies, and define their average
\begin{align}
T&=\frac1m\sum_{i=1}^m P_{B_i}.
\label{eq:overlap-operator}
\end{align}
For two-copy tests, we can use the three pairs of three copies. The required spectral estimates follow from Lemma~\ref{lem:so3} for the fermionic and centered bosonic Gaussian tests, Lemma~\ref{lem:fermions-su3} for the Slater test, and Lemma~\ref{lem:boson-overlaps} for the coherent-state test. The following theorem shows how such a spectral estimate yields the gradient bound. 

\begin{theorem}[Overlapping projections]\label{thm:overlap}
If the spectrum of $T$ satisfies
\begin{align}
 \Spec(T)&\subset\{1\}\cup[0,\beta],\qquad
 0\le\beta<(m-1)/m,
 \label{eq:overlap-spectrum}\\
 \kappa&=1-\frac m{m-1}\beta,
 \label{eq:overlap-gradient-constant}
\end{align}
then every state $\psi$ obeys that the norm of $\ket{\chi_\psi}$ from Eq.~\eqref{eq:ketchi_psi} is lower-bounded by
\begin{align}
 \|\ket{\chi_\psi}\|_2^2
 &\ge R(\psi)\bigl(\kappa-R(\psi)\bigr),
 \label{eq:overlap-e}\\
\end{align}
which using Lemma~\ref{lem:first-variation} we can find the bound
\begin{align}
 \|\nabla R(\psi)\|_2^2
 &\ge4q^2R(\psi)\bigl(\kappa-R(\psi)\bigr).
 \label{eq:overlap-gradient-bound}
\end{align}
\end{theorem}
\begin{proof}
Apply these operators to the product state
\begin{align}
 \ket z&=\ket\psi^{\otimes|\cup_iB_i|}.
 \label{eq:overlap-input}
\end{align}
Each $P_{B_i}$ has acceptance probability $p(\psi)$ on this state. For $i\ne j$, contract the copies that occur in only one of $B_i,B_j$. On their shared copy, each contraction leaves the vector $p(\psi)\ket\psi+\ket{\chi_\psi}$ from Eq.~\eqref{eq:restoration-contraction}. Thus
\begin{align}
\ip{P_{B_i}z}{P_{B_j}z}
&=\ip{\xi_\psi}{\xi_\psi}
\nonumber\\
\nonumber
&=p(\psi)^2\ip\psi\psi
+2p(\psi)\Re\ip\psi{\chi_\psi}
+\|\ket{\chi_\psi}\|_2^2\\
&=p(\psi)^2+\|\ket{\chi_\psi}\|_2^2,
 \label{eq:overlap-cross-term}
\end{align}
where we have used $\ip\psi{\chi_\psi}=0$ to compute the squared norm of the contracted vector.  We can use Eq.~\eqref{eq:overlap-operator} to see that
\begin{align}
 \ip z{T^2z}
 &=\frac{p(\psi)}m+\frac{m-1}{m}
       \bigl(p(\psi)^2+\|\ket{\chi_\psi}\|_2^2\bigr).
 \label{eq:overlap-second-moment}
\end{align}
Where we sum the $m$ diagonal terms and $m(m-1)$ off-diagonal terms in $T^2$. Since $\ip z{Tz}=p(\psi)$, subtracting gives
\begin{align}
 \ip z{T(I-T)z}
 &=\frac{m-1}{m}\bigl(p(\psi)R(\psi)-\|\ket{\chi_\psi}\|_2^2\bigr).
 \label{eq:overlap-moments}
\end{align}
For every spectral value $t$ allowed by Eq.~\eqref{eq:overlap-spectrum}, $t(1-t)\le\beta(1-t)$. Hence $T(I-T)\le\beta(I-T)$. Taking its expectation in $\ket z$ and rearranging yields
\begin{align}
 \frac{m-1}{m}\bigl(p(\psi)R(\psi)-\|\ket{\chi_\psi}\|_2^2\bigr)
 &\le\beta R(\psi), \nonumber \\
 \|\ket{\chi_\psi}\|_2^2
 &\ge R(\psi)\left(p(\psi)-\frac m{m-1}\beta\right) \nonumber \\
 &=R(\psi)\bigl(\kappa-R(\psi)\bigr).
 \label{eq:overlap-correction-bound}
\end{align}
The first inequality uses $\ip z{(I-T)z}=1-p(\psi)=R(\psi)$. The final equality substitutes $p(\psi)=1-R(\psi)$ and the definition of $\kappa$. The gradient bound follows from $\nabla R(\psi)=-2q\chi_\psi$ in Lemma~\ref{lem:first-variation}.
\end{proof}

Theorem~\ref{thm:overlap} gives a lower bound on the norm of the gradient. We now use this bound to control the distance from the input to the set of perfectly accepted states,
\begin{align}
\mathcal Z&=\{\ket g:\|\ket g\|_2=1,\ R( g)=0\}.
\label{eq:zero-rejection-set}
\end{align}
Starting from an input with $R(\psi)<\kappa$, we continuously change the normalized state along $2q\ket{\chi_\psi}=-\nabla R(\psi)$, decreasing its rejection probability. The next theorem shows that the gradient bound guarantees convergence to a state in $\mathcal Z$ and bounds the total length of this path. This gives an upper bound on the distance to $\mathcal Z$, which can be rearranged into a lower bound on rejection. For an exact test, $\mathcal Z=\mathcal T$, so the bound measures distance to the target family.

\begin{theorem}[Gradient flow and distance]\label{thm:restoration}
Suppose $q\ge2$, $P$ is permutation invariant, $\mathcal Z\ne\varnothing$, and, for some $0<\kappa\le1$ assume that the gradient fulfills
\begin{align}
 \|\nabla R(\psi)\|_2^2
 &\ge4q^2R(\psi)\bigl(\kappa-R(\psi)\bigr)
 \quad\text{whenever }\|\ket\psi\|_2=1,\quad 0<R(\psi)<\kappa.
 \label{eq:restoration-slope}
\end{align}
Every input with rejection $r_0<\kappa$ can be continuously changed into an element of $\mathcal Z$ along a path of length at most $q^{-1}\arcsin\sqrt{r_0/\kappa}$. Consequently,
\begin{equation}
\begin{aligned}
 R(\psi)&\ge
 \begin{cases}
 \kappa\sin^2\!\bigl(q\theta_{\mathcal Z}(\psi)\bigr),
   &\theta_{\mathcal Z}(\psi)\le\pi/(2q), \nonumber \\
 \kappa,&\theta_{\mathcal Z}(\psi)\ge\pi/(2q).
 \end{cases}
 \label{eq:restoration-envelope}
\end{aligned}
\end{equation}
\end{theorem}
\begin{proof}
If $R(\psi(0))=0$, the input is already in $\mathcal Z$. Otherwise, assume $0<r_0=R(\psi(0))<\kappa$ and evolve the state according to
\begin{align}
 \frac{\dd}{\dd t}\ket{\psi(t)}&=2q\ket{\chi_{\psi(t)}}.
 \label{eq:restoration-flow}
\end{align}
We first check that this equation defines normalized states for all $t\ge0$. To prove that the evolution has a unique solution, we first define the correction for vectors whose norm may differ from one. We then show that an initially normalized state remains normalized,
\begin{align}
 p(u)&=\|P\ket u^{\otimes q}\|_2^2,
 \label{eq:restoration-extended-acceptance}\\
 \ket{\chi_u}
 &=(\bra u^{\otimes(q-1)}\otimes I)P\ket u^{\otimes q}
   -p(u)\ket u.
 \label{eq:restoration-extended-correction}
\end{align}
For unit vectors this agrees with Eq.~\eqref{eq:restoration-contraction}. We now bound the change in the correction when $\|\ket u\|_2,\|\ket v\|_2\le B$.

First, for any positive integer $k$, the difference of tensor powers can be expanded
\begin{align}
 \ket u^{\otimes k}-\ket v^{\otimes k}
 &=\sum_{a=1}^k
 \ket u^{\otimes(a-1)}
 \otimes(\ket u-\ket v)
 \otimes\ket v^{\otimes(k-a)}
\end{align}
by changing one factor at a time. The triangle inequality and multiplicativity of the norm under tensor products give
\begin{align}
 \|\ket u^{\otimes k}-\ket v^{\otimes k}\|_2
 &\le\sum_{a=1}^k
 \|\ket u\|_2^{a-1}
 \|\ket u-\ket v\|_2
 \|\ket v\|_2^{k-a} \nonumber \\
 &\le\sum_{a=1}^k B^{k-1}\|\ket u-\ket v\|_2 \nonumber \\
 &=kB^{k-1}\|\ket u-\ket v\|_2.
\end{align}

For the contraction term, add and subtract $(\bra v^{\otimes(q-1)}\otimes I)P\ket u^{\otimes q}$. Since applying $P$ cannot increase the norm, we obtain
\begin{align}
 &\left\|
 (\bra u^{\otimes(q-1)}\otimes I)P\ket u^{\otimes q}
 -(\bra v^{\otimes(q-1)}\otimes I)P\ket v^{\otimes q}
 \right\|_2 \nonumber \\
 &\quad\le
 \|\ket u^{\otimes(q-1)}-\ket v^{\otimes(q-1)}\|_2
 \|\ket u\|_2^q
 +\|\ket v\|_2^{q-1}
 \|\ket u^{\otimes q}-\ket v^{\otimes q}\|_2 \nonumber \\
 &\quad\le
 \bigl[(q-1)B^{q-2}B^q+B^{q-1}qB^{q-1}\bigr]
 \|\ket u-\ket v\|_2 \nonumber \\
 &\quad=(2q-1)B^{2q-2}\|\ket u-\ket v\|_2.
\end{align}
The first contribution changes the $q-1$ bra factors, and the second changes the $q$ ket factors.

Next, use the difference of squares and the reverse triangle inequality
\begin{align}
 \left|
 \|P\ket u^{\otimes q}\|_2^2
 -\|P\ket v^{\otimes q}\|_2^2
 \right|&=\bigl(\|P\ket u^{\otimes q}\|_2
         +\|P\ket v^{\otimes q}\|_2\bigr)
   \left|\|P\ket u^{\otimes q}\|_2
         -\|P\ket v^{\otimes q}\|_2\right| \nonumber \\
 &\le 2B^q
 \left\|P\bigl(\ket u^{\otimes q}-\ket v^{\otimes q}\bigr)\right\|_2 \nonumber \\
 &\le 2qB^{2q-1}\|\ket u-\ket v\|_2.
\end{align}
Together with $p(u,v) = \|P\ket{u,v}^{\otimes q}\|_2^2$ and $p(u)\le B^{2q}$, this gives
\begin{align}
 \|p(u)\ket u-p(v)\ket v\|_2
 &=\left\|
 p(u)(\ket u-\ket v)+(p(u)-p(v))\ket v
 \right\|_2 \nonumber \\
 &\le p(u)\|\ket u-\ket v\|_2
      +|p(u)-p(v)|\,\|\ket v\|_2 \nonumber \\
 &\le B^{2q}\|\ket u-\ket v\|_2
      +2qB^{2q-1}B\,\|\ket u-\ket v\|_2 \nonumber \\
 &=(2q+1)B^{2q}\|\ket u-\ket v\|_2.
\end{align}

Finally, the difference between the corrections is bounded by the sum of the two contributions
\begin{align}
 \|\ket{\chi_u}-\ket{\chi_v}\|_2
 &\le (2q-1)B^{2q-2}\|\ket u-\ket v\|_2
      +(2q+1)B^{2q}\|\ket u-\ket v\|_2 \nonumber \\
 &=\bigl[(2q-1)B^{2q-2}+(2q+1)B^{2q}\bigr]
   \|\ket u-\ket v\|_2.
 \label{eq:restoration-lipschitz}
\end{align}
Thus, on any set of vectors with bounded norm, the change in the correction is bounded by a constant times the change in the input. The Picard--Lindel\"of theorem therefore gives a unique solution on a short time interval. The definition of the correction also gives
\begin{align}
 \ip\psi{\chi_\psi}
 &=p(\psi)\bigl(1-\|\ket\psi\|_2^2\bigr),
 \label{eq:restoration-norm-overlap}\\
 \frac{\dd}{\dd t}\|\ket{\psi(t)}\|_2^2
 &=4q\,p(\psi(t))\bigl(1-\|\ket{\psi(t)}\|_2^2\bigr).
 \label{eq:restoration-norm-preservation}
\end{align}
Starting from norm one, the last equation keeps the norm equal to one. On unit vectors, the contraction in Eq.~\eqref{eq:restoration-extended-correction} has norm at most one and $0\le p(\psi)\le1$. Hence
\begin{align}
 \|\ket{\chi_{\psi(t)}}\|_2&\le2, \nonumber \\
 \left\|\frac{\dd}{\dd t}\ket{\psi(t)}\right\|_2&\le4q.
 \label{eq:restoration-speed-bound}
\end{align}
The speed bound implies that the distance between the states at times $t$ and $s$ is at most $4q|t-s|$. If the solution were defined only up to a finite time, the states would therefore converge in the 2-norm as that time is approached. Starting the differential equation at this limiting state extends the solution. Thus the evolution exists for every $t\ge0$, also in an infinite-dimensional Hilbert space.

Since the state follows the negative gradient, the chain rule and \eqref{eq:restoration-slope} give
\begin{align}
 \dot R(\psi(t))&=-\|\nabla R(\psi(t))\|_2^2
 \label{eq:restoration-decay}\\
 &\le-4q^2R(\psi(t))\bigl(\kappa-R(\psi(t))\bigr) \\
 &\le-4q^2(\kappa-r_0)R(\psi(t)), \\
 R(\psi(t))&\le r_0\exp\!\bigl[-4q^2(\kappa-r_0)t\bigr].
 \label{eq:restoration-exponential-decay}
\end{align}
The rejection decreases, so $R(\psi(t))\le r_0<\kappa$ throughout the evolution and the assumed gradient bound remains applicable. The final inequality follows by integrating the preceding differential inequality. In particular, rejection tends to zero.

The speed equals $\|\nabla R(\psi(t))\|_2$, while the rejection decreases at the square of this speed. Changing the integration variable from $t$ to $R(\psi(t))$, recalling that $\frac{\dd}{\dd t}\ket{\psi(t)} =-\nabla R(\psi(t)) $, and using Eq.~\eqref{eq:restoration-slope} therefore gives
\begin{align}
 \int_0^\infty\left\|\frac{\dd}{\dd t}\ket{\psi(t)}\right\|_2\,\dd t
 &\le\int_0^{r_0}\frac{\dd r}{2q\sqrt{r(\kappa-r)}}
 \label{eq:restoration-integral}\\
 &=\frac1q\arcsin\sqrt{r_0/\kappa}. 
\end{align}
For the equality, substitute $r=\kappa\sin^2\alpha$. If rejection reaches zero at a finite time, the gradient then vanishes and the state stays fixed. Otherwise the same estimate applies to the entire evolution. For $t_2\ge t_1$, the difference between the states is bounded by the length travelled between those times
\begin{align}
 \|\ket{\psi(t_2)}-\ket{\psi(t_1)}\|_2
 &\le\int_{t_1}^{t_2}
 \left\|\frac{\dd}{\dd t}\ket{\psi(t)}\right\|_2\,\dd t.
 \label{eq:restoration-state-convergence}
\end{align}
Since the total integral is finite, the right-hand side tends to zero as $t_1\to\infty$, uniformly for $t_2\ge t_1$. The states therefore converge in the 2-norm to a unit vector $\ket g$. Continuity of $R$ and Eq.~\eqref{eq:restoration-exponential-decay} give $R(g)=0$.

It remains to compare the path length with the angular distance in Eq.~\eqref{eq:distances}. Differentiating the normalization gives $\Re\ip{\psi(t)}{\dot\psi(t)}=0$, and hence
\begin{align}
\left|\frac{{\rm d}}{{\rm d} t}\Re\ip{\psi(0)}{\psi(t)}\right|
 &=\left|\Re\ip{\psi(0)-(\Re\ip{\psi(0)}{\psi(t)})\psi(t)}{\dot\psi(t)}\right| \nonumber \\
 &\le\|\ket{\psi(0)}-(\Re\ip{\psi(0)}{\psi(t)})\ket{\psi(t)}\|_2\,
       \|\ket{\dot\psi(t)}\|_2 \nonumber \\
 &=\sqrt{1-\bigl(\Re\ip{\psi(0)}{\psi(t)}\bigr)^2}\,\|\ket{\dot\psi(t)}\|_2.
 \label{eq:restoration-angular-speed}
\end{align}
Here $\ket{\dot\psi(t)}$ denotes the derivative of $\ket{\psi(t)}$. The inequality bounds an inner product by the product of the two norms. The last equality follows by expanding the first squared norm and using $\|\ket{\psi(0)}\|_2=\|\ket{\psi(t)}\|_2=1$. Dividing by $\sqrt{1-\bigl(\Re\ip{\psi(0)}{\psi(t)}\bigr)^2}$ bounds the absolute derivative of $\arccos \Re\ip{\psi(0)}{\psi(t)}$ by the speed. Integrating from $\Re\ip{\psi(0)}{\psi(0)}=1$ and taking endpoint limits gives
\begin{align}
 \theta_{\mathcal Z}(\psi(0))
 &\le\arccos|\ip{\psi(0)}g| \nonumber \\
 &\le\arccos\Re\ip{\psi(0)}g \nonumber \\
 &\le\int_0^\infty\left\|\frac{\dd}{\dd t}\ket{\psi(t)}\right\|_2\,\dd t \nonumber \\
 &\le\frac1q\arcsin\sqrt{r_0/\kappa}.
 \label{eq:restoration-angle-bound}
\end{align}
The first inequality uses $g\in\mathcal Z$. The second uses $|\ip{\psi(0)}g|\ge\Re\ip{\psi(0)}g$ and the fact that $\arccos$ is decreasing. The last two inequalities use Eqs.~\eqref{eq:restoration-angular-speed} and \eqref{eq:restoration-integral}, respectively. Thus $q\theta_{\mathcal Z}(\psi(0))<\pi/2$ whenever $R(\psi(0))<\kappa$, and taking the sine and squaring gives the first branch of Eq.~\eqref{eq:restoration-envelope}. A state with $\theta_{\mathcal Z}(\psi)\ge\pi/(2q)$ must therefore have $R(\psi)\ge\kappa$, giving the second branch. States with $R(\psi)\ge\kappa$ also satisfy the first branch directly.
\end{proof}

For two-copy tests, the theorem gives a bound in squared distance
\begin{align}
 R(\psi)&\ge\ell_\kappa\bigl(d_{\mathcal Z}(\psi)^2\bigr).
 \label{eq:two-copy-flow-bound}
\end{align}
For exact tests, $\mathcal Z=\mathcal T$. Writing $d_{\mathcal T}=d_{\mathcal T}(\psi)$, the function is
\begin{equation}
\begin{aligned}
 \ell_\kappa(d_{\mathcal T}^2)&=
 \begin{cases}
 4\kappa d_{\mathcal T}^2(1-d_{\mathcal T}^2),
    &0\le d_{\mathcal T}^2\le\tfrac12,  \\
 \kappa,&\tfrac12\le d_{\mathcal T}^2\le1,
 \end{cases}
 \qquad 0<\kappa\le1.
 \label{eq:ell-definition}
\end{aligned}
\end{equation}
The first branch follows from $d_{\mathcal T}=\sin\theta_{\mathcal T}$ and the double-angle identity for the sine. The full bosonic Gaussian argument in Section~\ref{sec:boson-full} combines a coherent-state distance bound with a centered-Gaussian distance bound. For that argument, define their composition
\begin{equation}
\begin{aligned}
 \ell_{\mathrm B}(d_{\mathcal T}^2)
 &=\ell_{1/8}\!\left(\ell_{1/4}(d_{\mathcal T}^2)\right) \nonumber \\
 &=\begin{cases}
 \tfrac12d_{\mathcal T}^2(1-d_{\mathcal T}^2)
     (1-d_{\mathcal T}^2+d_{\mathcal T}^4),
     &0\le d_{\mathcal T}^2\le\tfrac12, \nonumber \\
 \tfrac3{32},&\tfrac12\le d_{\mathcal T}^2\le1.
 \end{cases}
 \label{eq:ellB-definition}
\end{aligned}
\end{equation}
Both functions are continuous and nondecreasing on $[0,1]$ and satisfy
\begin{align}
 \ell_\kappa(d_{\mathcal T}^2)&\ge\kappa d_{\mathcal T}^2,
 \label{eq:global-envelope-linear-bounds}\\
 \ell_{\mathrm B}(d_{\mathcal T}^2)&\ge\frac{3d_{\mathcal T}^2}{32}.
 \label{eq:bosonic-envelope-linear-bound}
\end{align}
The first inequality follows from $4(1-d_{\mathcal T}^2)\ge1$ on the first branch and $d_{\mathcal T}^2\le1$ on the second. For the second inequality, the definition of $\ell_{1/4}$ gives
\begin{align}
 \nonumber
 \frac{d_{\mathcal T}^2}{4}
 &\le\ell_{1/4}(d_{\mathcal T}^2)\le\frac14, \\
 \nonumber
 \ell_{\mathrm B}(d_{\mathcal T}^2)
 &=\frac12\ell_{1/4}(d_{\mathcal T}^2)
       \left[1-\ell_{1/4}(d_{\mathcal T}^2)\right] \\
        \nonumber
 &\ge\frac38\ell_{1/4}(d_{\mathcal T}^2) \\
 &\ge\frac{3d_{\mathcal T}^2}{32}.
 \label{eq:bosonic-envelope-linear-bound-proof}
\end{align}

\subsubsection{Rotation-invariant subspaces}
\label{sec:copy-angles}

Theorem~\ref{thm:overlap} proves the gradient inequality in Eq.~\eqref{eq:overlap-gradient-bound} from the spectral condition in Eq.~\eqref{eq:overlap-spectrum}. For two-copy Gaussian tests, the three projectors acting on pairs of three copies can be identified with projectors onto states unchanged by rotations about three orthogonal axes. The following lemma bounds the spectrum of their average. Its bound is independent of the angular-momentum sector, which allows us to apply it also to the bosonic Fock space.

\begin{lemma}[Spectral bound for orthogonal rotation axes]
\label{lem:so3}
Consider a unitary representation of $SO(3)$ on a finite-dimensional Hilbert space, or on an orthogonal direct sum of finite-dimensional spaces. Let $u_1,u_2,u_3\in\mathbb R^3$ be orthogonal unit vectors, and let $P_j$ project onto the states unchanged by every rotation about $u_j$. Then
\begin{align}
 \Spec\!\left(\frac{P_1+P_2+P_3}{3}\right)
 &\subseteq\{1\}\cup[0,7/12].
 \label{eq:fermions-so3-bound}
\end{align}
More generally, in a single spin-$\ell$ representation of dimension $2\ell+1$, with integer $\ell\ge0$, there is a unique normalized state $\ket{z_u}$ up to phase with zero angular momentum along each unit axis $u$. These are precisely the states unchanged by rotations about $u$. Their phases can be chosen so that, for any unit axes $u,v$,
\begin{align}
 \ip{z_u}{z_v}&=P_\ell(u\cdot v),
 \label{eq:fermions-zonal-overlap}\\
 \|\ket{z_u}\|_2&=1.
 \label{eq:so3-axis-state-normalization}
\end{align}
Here $P_\ell(t)$ is the degree-$\ell$ Legendre polynomial, normalized by $P_\ell(1)=1$.
\end{lemma}

\begin{proof}
Decompose the representation into irreducible angular-momentum
sectors. In a spin-$\ell$ sector, with integer $\ell\ge0$,
the states unchanged by rotations about an axis $u$ form
a one-dimensional space: the zero-angular-momentum eigenspace
along that axis. Averaging these rotations therefore gives
\begin{align}
 \ket{z_u}\bra{z_u}
 &=\frac1{2\pi}\int_0^{2\pi}
       \e^{-i\varphi L_u}\,\dd\varphi,
 \label{eq:so3-axis-projector}
\end{align}
where $L_u$ is the angular-momentum operator along $u$.

To compute the overlaps, realize this sector as the degree-$\ell$
spherical harmonics on the unit sphere, with its usual surface
measure. The normalized state unchanged by rotations about $u$
can be chosen as
\begin{align}
 z_u(\hat{\mathbf r})
 &=\sqrt{\frac{2\ell+1}{4\pi}}\,
   P_\ell(u\cdot\hat{\mathbf r}).
 \label{eq:so3-axis-wavefunctions}
\end{align}
This is the standard spherical harmonic with zero angular momentum
along $u$. Its normalization and the spherical-harmonic addition
theorem give
$\ip{z_u}{z_v}=P_\ell(u\cdot v)$,
as stated in Eq.~\eqref{eq:fermions-zonal-overlap};
see~\cite[\S14.30]{dlmf}.

For orthogonal axes, Eq.~\eqref{eq:fermions-zonal-overlap} gives the same overlap between every pair of distinct states. Let $A$ be the matrix whose three columns are the states $\ket{z_{u_1}},\ket{z_{u_2}},\ket{z_{u_3}}$. Its matrix of column inner products is
\begin{align}
 G_\ell=A^\dagger A
 &=\begin{pmatrix}
 1&P_\ell(0)&P_\ell(0) \\
 P_\ell(0)&1&P_\ell(0) \\
 P_\ell(0)&P_\ell(0)&1
 \end{pmatrix},
 \label{eq:so3-overlap-matrix}\\
 \left.\frac{P_1+P_2+P_3}{3}\right|_{\text{spin }\ell}
 &=\frac13\sum_{j=1}^3\ket{z_{u_j}}\bra{z_{u_j}}
 =\frac13AA^\dagger.
 \label{eq:so3-average-factorization}
\end{align}
Recall that by definition, $P_j$ keeps the components unchanged by rotations about $u_j$, thus $P_j$ and $\ketbra{z_{u_j}}$ act identically in every state in the sector. The operators $AA^\dagger$ and $A^\dagger A$ have the same nonzero eigenvalues. Indeed, an eigenvector of $A^\dagger A$ with eigenvalue $\mu>0$ is mapped by $A$ to an eigenvector of $AA^\dagger$ with the same eigenvalue. The converse follows using $A^\dagger$. The vector $(1,1,1)^{\mathsf T}$ is an eigenvector of $G_\ell$ with eigenvalue $1+2P_\ell(0)$. Every vector orthogonal to it has eigenvalue $1-P_\ell(0)$. Therefore the eigenvalues of $G_\ell/3$ are
\begin{align}
 \frac{1+2P_\ell(0)}{3}&\quad\text{with multiplicity one},
 \label{eq:fermions-three-lines}\\
 \frac{1-P_\ell(0)}{3}&\quad\text{with multiplicity two}.
 \label{eq:so3-repeated-eigenvalue}
\end{align}
Every nonzero eigenvalue of the averaged projector in Eq.~\eqref{eq:so3-average-factorization} is among these values.
To evaluate these eigenvalues, we use Rodrigues' formula
\cite[Eq.~(14.7.10) with $m=0$]{dlmf}:
\begin{align}
 P_\ell(t)
 &=\frac{1}{2^\ell\ell!}
   \frac{\dd^\ell}{\dd t^\ell}(t^2-1)^\ell.
 \label{eq:fermions-rodrigues}
\end{align}
Evaluating Eq.~\eqref{eq:fermions-rodrigues} at zero gives
\begin{align}
 P_{2j+1}(0)&=0, \\
 P_{2j}(0)&=(-1)^j\frac{\binom{2j}{j}}{4^j},\qquad j\ge0,
 \label{eq:so3-legendre-zero-values}\\
 \frac{|P_{2j+2}(0)|}{|P_{2j}(0)|}
 &=\frac{(2j+2)(2j+1)}{4(j+1)^2} \\
 &=\frac{2j+1}{2j+2}<1.
 \label{eq:so3-overlap-magnitude-ratio}
\end{align}
The odd values vanish because $(t^2-1)^\ell$ contains only even powers. For $\ell=2j$, its coefficient of $t^{2j}$ is $(-1)^j\binom{2j}{j}$, giving the second line after differentiation. The remaining lines show that the magnitudes decrease with $j$. Since $P_2(0)=-1/2$ and $P_4(0)=3/8$, it follows that
\begin{align}
 -\frac12&\le P_\ell(0)\le\frac38\qquad(\ell\ge1),
 \label{eq:so3-overlap-interval}\\
 \frac{1+2P_\ell(0)}{3}&\le\frac{1+2(3/8)}3=\frac7{12}, \\
 \frac{1-P_\ell(0)}{3}&\le\frac{1+1/2}{3}=\frac12.
 \label{eq:so3-eigenvalue-upper-bounds}
\end{align}
These are upper bounds on all nonzero eigenvalues by Eqs.~\eqref{eq:fermions-three-lines} and \eqref{eq:so3-repeated-eigenvalue}. All eigenvalues are nonnegative because the operator is an average of positive projectors. The value $7/12$ is attained at $\ell=4$. At $\ell=0$, all rotations act as the identity, so the average has only the eigenvalue one. This proves Eq.~\eqref{eq:fermions-so3-bound} on each sector. For an orthogonal direct sum, the operator is block diagonal and its spectrum is the closure of the union of the block spectra. The same closed set $\{1\}\cup[0,7/12]$ contains every block spectrum, so the bound also holds on the full space.
\end{proof}

When the three projectors in Theorem~\ref{thm:overlap} have this rotation representation, Eq.~\eqref{eq:fermions-so3-bound} verifies Eq.~\eqref{eq:overlap-spectrum} with $m=3$ and $\beta=7/12$. In particular, $7/12<2/3$, as required there, and Eq.~\eqref{eq:overlap-gradient-constant} gives
\begin{align}
 \kappa&=1-\frac32\frac7{12}=\frac18.
 \label{eq:so3-gradient-constant}
\end{align}
For nonorthogonal axes, Eq.~\eqref{eq:fermions-zonal-overlap}
requires values of $P_\ell(t)$ at $t=u\cdot v$.
We will use the standard integral representation
\cite[Eq.~(18.10.5)]{dlmf}:
\begin{align}
 P_\ell(t)
 &=\frac1{2\pi}\int_0^{2\pi}
   \bigl(t+i\sqrt{1-t^2}\cos\varphi\bigr)^\ell\,\dd\varphi,
 \qquad -1\le t\le1.
 \label{eq:fermions-legendre-integral}
\end{align}

\color{black}
\subsection{Higher-copy tests}
\label{sec:higher-copy-tests}
The overall goal of this section is to establish explicit upper and lower bounds on the rejection probability in terms of the distance to the target family and the number of copies $q$, providing the estimates needed for tolerant testing. 

We begin with Definition~\ref{def:normal-projectors}, which specifies the $q$-copy acceptance projectors for each target family. Lemma~\ref{lem:normal-transition} then shows how copy rotations act on the first-order tensor terms, giving a way to calculate their acceptance probabilities. Theorem~\ref{thm:normal-hierarchy} gives the resulting rejection coefficients $\lambda_{\mathcal T,q}$ and combines the closest-target expansion with comparisons to tests on smaller groups of copies to obtain upper and lower bounds on the rejection probability, both near the target family and at larger distances.

\begin{definition}[Higher-copy acceptance projectors]
\label{def:normal-projectors}
For the bosonic tests below, let $a_{\mu,j}$ denote the annihilation operator for physical mode $\mu$ in copy $j$. Define
\begin{align}
 \mathbf u&=\frac{(1,\ldots,1)}{\sqrt q},
 \label{eq:higher-copy-collective-vector}\\
 a_{\mu,\mathbf u}&=\frac1{\sqrt q}\sum_{j=1}^q a_{\mu,j}.
 \label{eq:higher-copy-collective-mode}
\end{align}
The entries of $\mathbf u$ are the coefficients used to combine the same physical mode across all $q$ copies with equal weights. The remaining $q-1$ combinations use an orthonormal basis of coefficient vectors orthogonal to $\mathbf u$. For example, for two copies these are the sum $(a_{\mu,1}+a_{\mu,2})/\sqrt2$ and the difference $(a_{\mu,1}-a_{\mu,2})/\sqrt2$.

For $q\ge2$, define:
\begin{enumerate}[label=(\roman*)]
\item $P_{\GF,q}$ projects onto the simultaneous kernels of the two-copy Gaussian constraints
 \begin{align}
  \Lambda_{a,b}&=\sum_{\mu=1}^{2n}\gamma_{\mu,a}\gamma_{\mu,b},
  \qquad 1\le a<b\le q.
  \label{eq:higher-copy-gaussian-constraints}
 \end{align}
Here $\gamma_{\mu,a}$ is Majorana operator $\mu$ acting on copy $a$ of the ordinary Hilbert-space tensor product. Thus $\Lambda_{a,b}$ applies the two-copy operator $\sum_\mu\gamma_\mu\otimes\gamma_\mu$ to copies $a,b$ and the identity to all other copies.
\item $P_{\Sl,q}$ projects onto the intersection of the accepted spaces of the two-copy Slater test over all pairs of copies. For each pair, the accepted space consists of the states unchanged by all $SU(2)$ transformations mixing the two fermionic copies, with the fermionic parity factors included. These states are called \emph{singlets}.
\item $P_{\Gzero,q}$ is the average of the unitary action of $SO(q)$ mixing the $q$ copy labels, with the same rotation applied to every physical mode.
\item $P_{\Coh,q}$ checks that all $q-1$ combinations orthogonal to $\mathbf u$ are in vacuum, for every physical mode $\mu$. It acts as the identity on each equally weighted mode $a_{\mu,\mathbf u}$.
\end{enumerate}
For $q\ge3$, $P_{\GB,q}$ averages real copy rotations that leave $\mathbf u$ unchanged. These rotations form an $SO(q-1)$ group acting on the orthogonal complement of $\mathbf u$. In each case, define
\begin{align}
 R_{\mathcal T,q}(\psi)
 &=\|(I-P_{\mathcal T,q})\ket\psi^{\otimes q}\|_2^2.
 \label{eq:higher-copy-rejection}
\end{align}
Every group average uses Haar measure normalized to total mass one. It is the orthogonal projector onto the states unchanged by every unitary in the group. 
\end{definition}

The projectors in Definition~\ref{def:normal-projectors} are perfectly complete. For a closest target $\ket g$ and a unit deviation $\ket\eta$ as in Eq.~\eqref{eq:normal-frame}, they satisfy
\begin{align}
 \left\|(I-P_{\mathcal T,q})
       \bigl(D(g^{\otimes q})[\eta]\bigr)\right\|_2^2
 &\ge\lambda_{\mathcal T,q}.
\label{eq:higher-copy-coefficient-bound}
\end{align}
with $\lambda_{\mathcal T,q}$ summarized in Table~\ref{table:normal-coefficients}. 

The next lemma explains why moments of individual rotation-matrix entries determine acceptance of the first-order tensor terms. Write $\ket\Omega$ for the vacuum of one copy. For $k\ge1$, define the map $J_{i,k}$ by placing a $k$-particle state in copy $i$ and vacuum in all other copies
\begin{align}
 J_{i,k}\ket{\eta_k}
 &=\ket\Omega^{\otimes(i-1)}\otimes\ket{\eta_k}
   \otimes\ket\Omega^{\otimes(q-i)}.
 \label{eq:higher-copy-insertion-map}
\end{align}
The map preserves norms. Different copy positions give orthogonal states because $\ket{\eta_k}$ is orthogonal to vacuum.

\begin{lemma}[Transition amplitudes under copy rotations]
\label{lem:normal-transition}
Let $O$ be a real orthogonal matrix mixing the copy labels, and let $\Gamma(O)$ be its unitary action on bosonic or fermionic Fock space. Then, on the one-copy $k$-particle space,
\begin{align}
 J_{j,k}^{\dagger}\Gamma(O)J_{i,k}&=O_{j,i}^{\,k}I.
 \label{eq:normal-transition}
\end{align}
Consequently, if $P=\mathbb E_O[\Gamma(O)]$ averages a compact group of such rotations and $\|\ket{\eta_k}\|_2=1$, then
\begin{align}
 \frac1q\left\|P\bigl(D(\Omega^{\otimes q})[\eta_k]\bigr)\right\|_2^2
 &=\frac1q\sum_{i,j=1}^q\mathbb E_O [O_{j,i}^{\,k}].
 \label{eq:higher-copy-first-order-acceptance}
\end{align}
This is the acceptance probability of the normalized vector $D(\Omega^{\otimes q})[\eta_k]/\sqrt q$.
\end{lemma}

\begin{proof}
Let $a_{\mu,i}^\dagger$ create a particle in physical mode $\mu$ and copy $i$. Copy mixing preserves vacuum and transforms each creation operator as
\begin{align}
 \Gamma(O)a_{\mu,i}^\dagger\Gamma(O)^\dagger
 &=\sum_{j=1}^q O_{j,i}a_{\mu,j}^\dagger,
 \label{eq:higher-copy-creation-action}\\
 \Gamma(O)\ket\Omega^{\otimes q}&=\ket\Omega^{\otimes q}.
 \label{eq:higher-copy-vacuum-invariance}
\end{align}
Therefore a state with all $k$ particles in copy $i$ transforms as
\begin{align}
 &\Gamma(O)a_{\mu_1,i}^\dagger\cdots a_{\mu_k,i}^\dagger
   \ket\Omega^{\otimes q} \nonumber \\
 &\qquad=\prod_{r=1}^k
   \left(\sum_{j_r=1}^q O_{j_ri}a_{\mu_r,j_r}^\dagger\right)
   \ket\Omega^{\otimes q} \nonumber \\
 &\qquad=\sum_{j_1,\ldots,j_k=1}^q
   \left(\prod_{r=1}^k O_{j_ri}\right)
   a_{\mu_1,j_1}^\dagger\cdots a_{\mu_k,j_k}^\dagger
   \ket\Omega^{\otimes q}.
 \label{eq:higher-copy-particle-expansion}
\end{align}

The adjoint $J_{j,k}^\dagger$ checks that every copy except $j$ is in vacuum and then returns the $k$-particle state in copy $j$. Consequently, on each term of the expansion above, 
\begin{align}
 &J_{j,k}^\dagger
 a_{\mu_1,j_1}^\dagger\cdots a_{\mu_k,j_k}^\dagger
 \ket\Omega^{\otimes q}  \\
 &\qquad=
 \begin{cases}
 a_{\mu_1}^\dagger\cdots a_{\mu_k}^\dagger\ket\Omega,
 &j_1=\cdots=j_k=j,  \\
 0,&\text{otherwise}.
 \end{cases}
 \label{eq:higher-copy-insertion-adjoint}
\end{align}
Indeed, if any particle occupies another copy, that copy is orthogonal to vacuum, so the term vanishes. Here $a_\mu^\dagger$ denotes the creation operator in the one-copy space.

Each of the $k$ creation operators contributes a factor $O_{j,i}$ to the surviving term. Therefore,
\begin{align}
 &J_{j,k}^\dagger\Gamma(O)J_{i,k}
 a_{\mu_1}^\dagger\cdots a_{\mu_k}^\dagger\ket\Omega \nonumber \\
 &\qquad=
 \left(\prod_{r=1}^k O_{j,i}\right)
 a_{\mu_1}^\dagger\cdots a_{\mu_k}^\dagger\ket\Omega \nonumber \\
 &\qquad=
 O_{j,i}^k
 a_{\mu_1}^\dagger\cdots a_{\mu_k}^\dagger\ket\Omega.
 \label{eq:higher-copy-transition-on-products}
\end{align}
The physical mode labels and the order of the creation operators are unchanged. Since these vectors span the one-copy $k$-particle space, linearity gives
\begin{align}
 J_{j,k}^\dagger\Gamma(O)J_{i,k}\ket{\eta_k}
 &=O_{j,i}^k\ket{\eta_k}
 \qquad\text{for every }\ket{\eta_k}.
 \label{eq:higher-copy-transition-on-states}
\end{align}
This proves Eq.~\eqref{eq:normal-transition}.

For fermions, the creation operators acting on different copies include the Jordan--Wigner parity factors in Eq.~\eqref{eq:higher-copy-creation-action}. In the surviving term, every creation operator acts on copy $j$, while all preceding copies remain in vacuum. Their parity is therefore $+1$ throughout. The same argument applies to the initial state in copy $i$.

By Eqs.~\eqref{eq:tensor-derivative} and \eqref{eq:higher-copy-insertion-map},
\begin{align}
 D(\Omega^{\otimes q})[\eta_k]
 &=\sum_{i=1}^q J_{i,k}\ket{\eta_k}, \nonumber \\
 \|D(\Omega^{\otimes q})[\eta_k]\|_2^2
 &=\sum_{i=1}^q\|J_{i,k}\ket{\eta_k}\|_2^2=q.
 \label{eq:higher-copy-derivative-norm}
\end{align}
The norm identity uses the orthogonality of different copy positions and $\|\ket{\eta_k}\|_2=1$. Since $P^\dagger P=P$, its acceptance probability is
\begin{align}
 \frac1q\left\|P\bigl(D(\Omega^{\otimes q})[\eta_k]\bigr)\right\|_2^2
 &=\frac1q\sum_{i,j=1}^q
       \ip{J_{j,k}\eta_k}{P J_{i,k}\eta_k} \nonumber \\
 &=\frac1q\sum_{i,j=1}^q\mathbb E_O
       \ip{\eta_k}{J_{j,k}^\dagger\Gamma(O)J_{i,k}\eta_k} \nonumber \\
 &=\frac1q\sum_{i,j=1}^q\mathbb E_O
       \bigl[O_{j,i}^{\,k}\ip{\eta_k}{\eta_k}\bigr] \nonumber \\
 &=\frac1q\sum_{i,j=1}^q\mathbb E_O O_{j,i}^{\,k}.
 \label{eq:higher-copy-acceptance-expansion}
\end{align}
The first equality expands the sum of copy positions. The second uses $P=\mathbb E_O\Gamma(O)$ and moves $J_{j,k}$ to the other side of the inner product. The third applies Eq.~\eqref{eq:normal-transition}, and the last uses normalization of $\ket{\eta_k}$.
\end{proof}

We next combine the response near a target with the bounds from smaller copy blocks. The coefficients for each family require the moment calculations in the Fermionic and Bosonic appendices. We state their values here and identify those deferred calculations in the proof.

\begin{theorem}[Upper and lower bounds for higher-copy tests]
\label{thm:normal-hierarchy}
The projectors in Definition~\ref{def:normal-projectors} are perfectly complete. Let $\ket\psi$ be a normalized pure input. For every target family and copy count, the rejection probability satisfies
\begin{align}
 R_{\mathcal T,q}(\psi)
\le 1-(1-d_{\mathcal T}^2(\psi))^q.
\label{eq:normal-hierarchy-calibration}
\end{align}

For a lower bound, let $\ell_{\GF}=\ell_{\Gzero}=\ell_{1/8}$, $\ell_{\Sl}=\ell_{1/6}$, $\ell_{\Coh}=\ell_{1/4}$, and $\ell_{\mathrm B}(x)=\ell_{1/8}(\ell_{1/4}(x))$, using the functions defined in Eqs.~\eqref{eq:ell-definition} and \eqref{eq:ellB-definition}. Combining the block comparison with the bounds for the basic tests gives
\begin{align}
R_{\mathcal T,q}(\psi)
&\ge1-\left[1-\ell_{\mathcal T}(d_{\mathcal T}^2(\psi))\right]^{\lfloor q/2\rfloor},
\quad \mathcal T\in\{\GF,\Sl,\Gzero,\Coh\},\quad q\ge2,
\label{eq:normal-inherited-pairs}\\
R_{\GB,3}(\psi)
&\ge\tfrac49\ell_{\mathrm B}(d_{\GB}^2(\psi)),
\label{eq:normal-inherited-three}\\
R_{\GB,q}(\psi)
&\ge1-\left[1-\ell_{\mathrm B}(d_{\GB}^2(\psi))\right]^{\lfloor q/4\rfloor},
\quad q\ge4.
\label{eq:normal-inherited-four}
\end{align}

In particular, every family satisfies, for $q\ge4$,
\begin{align}
 R_{\mathcal T,q}(\psi)
 &\ge1-\left(1-\frac{3d_{\mathcal T}^2(\psi)}{32}\right)^{\lfloor q/4\rfloor}.
 \label{eq:normal-universal-global}
\end{align}
\end{theorem}

\begin{proof}

Perfect completeness follows from the family-specific constructions established in the fermionic and bosonic appendices. Eq.~\eqref{eq:universal-upper} therefore gives the upper bound in Eq.~\eqref{eq:normal-hierarchy-calibration}.

For the bounds at arbitrary distance, the additional family-specific inputs are the following inequalities for the basic tests
\begin{align}
 R_{\mathcal T,2}(\psi)&\ge\ell_{\mathcal T}(d_{\mathcal T}^2),
 \qquad \mathcal T\in\{\GF,\Sl,\Gzero,\Coh\},
 \label{eq:higher-copy-basic-pair-bounds}\\
 R_{\mathrm{diff},4}(\psi)&\ge\ell_{\mathrm B}(d_{\GB}^2),
 \label{eq:higher-copy-basic-four-bound}\\
 R_{\GB,3}(\psi)&\ge\tfrac49 R_{\mathrm{diff},4}(\psi).
 \label{eq:higher-copy-basic-three-comparison}
\end{align}
Their proofs are given later in Eqs.~\eqref{eq:fermions-calibration}, \eqref{eq:fermions-slater-calibration}, \eqref{eq:boson-centered-calibration}, and \eqref{eq:boson-coherent-calibration} for the pairs, and Eqs.~\eqref{eq:boson-four-calibration} and \eqref{eq:boson-three-calibration} for the full bosonic family. Here $P_{\mathrm{diff},4}$ is the projector obtained by averaging rotations in the plane spanned by $(1,-1,0,0)/\sqrt2$ and $(0,0,1,-1)/\sqrt2$ in four-copy space, and $R_{\mathrm{diff},4}$ is its rejection probability. This is the four-copy test used in Eq.~\eqref{eq:higher-copy-basic-four-bound}.

We first prove Eq.~\eqref{eq:normal-inherited-pairs}. Every state accepted by $P_{\mathcal T,q}$ is accepted by the two-copy projector on each pair. For fermions this follows from the intersection in Definition~\ref{def:normal-projectors}. For centered bosons, invariance under all copy rotations includes invariance under rotations of any pair. For coherent states, vacuum in every coordinate orthogonal to $\mathbf u$ includes vacuum in the difference of any pair of copies.

Choose $m=\lfloor q/2\rfloor$ disjoint pairs. Their projectors act on separate tensor factors, so they commute and their product projects onto their common accepted space. The inclusion just described gives
\begin{align}
 P_{\mathcal T,q}
 &\le\prod_{j=1}^m P_{\mathcal T,2}^{(2j-1,2j)}.
 \label{eq:higher-copy-pair-inclusion}
\end{align}
The superscript specifies the two copies acted on. All other copies are left unchanged. Taking the expectation in the product input and using Eq.~\eqref{eq:higher-copy-basic-pair-bounds} yields
\begin{align}
 1-R_{\mathcal T,q}(\psi)
 &=\ip{\psi^{\otimes q}}{P_{\mathcal T,q}\psi^{\otimes q}} \nonumber \\
 &\le\ip{\psi^{\otimes q}}
       {\left(\prod_{j=1}^m P_{\mathcal T,2}^{(2j-1,2j)}\right)
        \psi^{\otimes q}} \nonumber \\
 &=\prod_{j=1}^m
       \ip{\psi^{\otimes2}}{P_{\mathcal T,2}\psi^{\otimes2}} \nonumber \\
 &=\bigl[1-R_{\mathcal T,2}(\psi)\bigr]^m \nonumber \\
 &\le\left[1-\ell_{\mathcal T}(d_{\mathcal T}^2)\right]^m.
 \label{eq:higher-copy-pair-acceptance}
\end{align}
The first equality writes acceptance as an expectation of its projector. The inequality applies Eq.~\eqref{eq:higher-copy-pair-inclusion}. The next equality factors the expectation over disjoint pairs. An unused copy contributes one because $\ket\psi$ is normalized. Each factor is the acceptance probability of the same two-copy input. The last line uses Eq.~\eqref{eq:higher-copy-basic-pair-bounds}. Subtracting from one proves Eq.~\eqref{eq:normal-inherited-pairs}.

For full bosonic Gaussians, each rotation defining $P_{\mathrm{diff},4}$ leaves the equal-superposition coordinate unchanged. When embedded in any four of the $q$ copies, it therefore belongs to the group defining $P_{\GB,q}$. Thus a state accepted by $P_{\GB,q}$ is accepted by this four-copy test on every block. For $m=\lfloor q/4\rfloor$ disjoint blocks, the same inclusion and factorization give
\begin{align}
 P_{\GB,q}
 &\le\prod_{j=1}^m
       P_{\mathrm{diff},4}^{(4j-3,4j-2,4j-1,4j)}, \nonumber \\
 1-R_{\GB,q}(\psi)
 &\le\bigl[1-R_{\mathrm{diff},4}(\psi)\bigr]^m \nonumber \\
 &\le\left[1-\ell_{\mathrm B}(d_{\GB}^2)\right]^m.
 \label{eq:higher-copy-four-acceptance}
\end{align}
The last inequality uses Eq.~\eqref{eq:higher-copy-basic-four-bound}. Subtracting from one proves Eq.~\eqref{eq:normal-inherited-four}. Combining Eqs.~\eqref{eq:higher-copy-basic-three-comparison} and \eqref{eq:higher-copy-basic-four-bound} gives the three-copy bound Eq.~\eqref{eq:normal-inherited-three}.

Finally, Eqs.~\eqref{eq:global-envelope-linear-bounds} and \eqref{eq:bosonic-envelope-linear-bound} imply
\begin{align}
 \ell_{\mathcal T}(d_{\mathcal T}^2)
 &\ge\frac{3d_{\mathcal T}^2}{32},
 \qquad\mathcal T\in\{\GF,\Sl,\Gzero,\Coh\}, \nonumber \\
 \ell_{\mathrm B}(d_{\GB}^2)&\ge\frac{3d_{\GB}^2}{32}.
 \label{eq:higher-copy-common-linear-bound}
\end{align}
For the first line, the smallest of the four constants defining $\ell_{\mathcal T}$ is $1/8\ge3/32$. Substitution into Eqs.~\eqref{eq:normal-inherited-pairs} and \eqref{eq:normal-inherited-four}, together with $\lfloor q/2\rfloor\ge\lfloor q/4\rfloor$, proves Eq.~\eqref{eq:normal-universal-global}.

For every allowed copy count, the bounds in Eqs.~\eqref{eq:normal-inherited-pairs}, \eqref{eq:normal-inherited-three}, and \eqref{eq:normal-inherited-four} are strictly positive when $d_{\mathcal T}>0$. At $d_{\mathcal T}=0$, a closest target has squared overlap one with the input by Eq.~\eqref{eq:distances}, so the two states differ only by a global phase. Since each target family includes global phases, the input belongs to $\mathcal T$. Together with perfect completeness, this identifies precisely the pure inputs with zero rejection.
\end{proof}

\subsection{Tolerant testing}\label{sec:testing}

For $0\le a<b\le1$, an $(a,b)$ tolerant test accepts inputs with $d_{\mathcal T}(\psi)\le a$ and rejects inputs with $d_{\mathcal T}(\psi)\ge b$, with bounded error in each case. There is no requirement for inputs between these two distances. The case $a=0$ is non-tolerant testing. 
From Proposition~\ref{prop:testing-statistics} in Section~\ref{sec:tolerant_testing_intro} we need an upper bound $\mathcal{U}$ and a lower bound $L$ on the rejection probability. We then obtain $\mathcal{U}$ and $L$ from the distance bounds in Theorem~\ref{thm:normal-hierarchy}.

\color{black}

\color{red}

\color{black}
\subsubsection{Small-gap tolerant testing}
\label{subsec:testing-optimal-gap}

We next give explicit choices of $q$ and the total number of copies. The first step is to control $h_{q,q-3}$ on a specified interval of squared distances.

\begin{lemma}[Local lower bound for higher-copy rejection]
\label{lem:testing-local-window}
For every integer $q\ge12$, the function $h_{q,q-3}$ is nondecreasing on $[0,1/(3q)]$. For $0\le d_{\mathcal T}^2\le1/(3q)$, it satisfies
\begin{align}
 h_{q,q-3}(d_{\mathcal T}^2)
 &\ge qd_{\mathcal T}^2\left[
 1-\frac2q-\frac{q-1}{2}d_{\mathcal T}^2
 -\sqrt{\frac{3(q-1)}{2q}}\,d_{\mathcal T}\right]_+^2.
 \label{eq:testing-local-window-bound}
\end{align}
\end{lemma}

\begin{proof}
Write the expression inside the positive part in $h_{q,q-3}$, Eq.~\eqref{eq:normal-h} as
\begin{align}
 S(d_{\mathcal T}^2)
 &=\sqrt{q-3}\,d_{\mathcal T}(1-d_{\mathcal T}^2)^{(q-1)/2}
   -\sqrt{\frac3q \|\ket w\|_2^2},
 \label{eq:testing-S}
\end{align}
and thus
\begin{align}
 h_{q,q-3}(d_{\mathcal T}^2)&=[S(d_{\mathcal T}^2)]_+^2.
 \label{eq:testing-S-positive-part}
\end{align}
We show that the first term of $S$ increases faster than the term being subtracted. Throughout the derivative calculation, $0<d_{\mathcal T}^2\le1/(3q)$, and primes denote differentiation with respect to squared distance.

The derivative of the first term is
\begin{align}
 &\frac{\dd}{\dd(d_{\mathcal T}^2)}
   \left[\sqrt{q-3}\,d_{\mathcal T}
              (1-d_{\mathcal T}^2)^{(q-1)/2}\right] \nonumber \\
 &\qquad=\frac{\sqrt{q-3}}{2d_{\mathcal T}}
       (1-d_{\mathcal T}^2)^{(q-3)/2}(1-qd_{\mathcal T}^2).
 \label{eq:testing-first-term-derivative}
\end{align}
Each factor has a uniform lower bound on this interval
\begin{align}
 \frac{\sqrt{q-3}}{2d_{\mathcal T}}
 &\ge\frac{\sqrt{3q(q-3)}}2\ge\frac{3q}{4}, \nonumber \\
 (1-d_{\mathcal T}^2)^{(q-3)/2}
 &\ge1-\frac{q-3}{2}d_{\mathcal T}^2
 \ge1-\frac{q-3}{6q}\ge\frac56, \nonumber \\
 1-qd_{\mathcal T}^2&\ge\frac23.
 \label{eq:testing-derivative-factors}
\end{align}
The first line uses $q\ge12$. The second uses Bernoulli's inequality, with exponent $(q-3)/2\ge1$. Multiplying these bounds shows that Eq.~\eqref{eq:testing-first-term-derivative} is at least $5q/12$.

From its definition in Eq.~\eqref{eq:normal-B}, $\|\ket w\|_2^2$ is the probability of at least two successes in $q$ independent Bernoulli trials, each with success probability $d_{\mathcal T}^2$. Differentiating $\|\ket w\|_2^2$ with respect to $d_{\mathcal T}^2$ gives
\begin{align}
 \frac{\dd}{\dd (d_{\mathcal T}^2)}\|\ket w\|_2^2
 &=q(q-1)d_{\mathcal T}^2(1-d_{\mathcal T}^2)^{q-2}, 
 \end{align}
 and the probability of exactly two successes gives the second inequality.
 \begin{align}
     \|\ket w\|_2^2
 &\ge\binom q2d_{\mathcal T}^4(1-d_{\mathcal T}^2)^{q-2}.
 \label{eq:testing-binomial-lower}
\end{align}
Consequently, the derivative of the subtracted term satisfies
\begin{align}
 \frac{\dd}{\dd(d_{\mathcal T}^2)}
       \sqrt{\frac3q }\|\ket w\|_2
 &=\frac{\sqrt{3/q}}
         {2\|\ket w\|_2}\frac{\dd}{\dd (d_{\mathcal T}^2)}\|\ket w\|_2^2 \nonumber \\
 &\le\sqrt{\frac{3(q-1)}2}
       (1-d_{\mathcal T}^2)^{(q-2)/2} \nonumber \\
 &\le\sqrt{\frac{3q}2} \nonumber \\
 &\le\frac q{\sqrt8}<\frac{5q}{12}.
 \label{eq:testing-second-term-derivative}
\end{align}
The penultimate inequality uses $q\ge12$. Comparing Eqs.~\eqref{eq:testing-first-term-derivative} and \eqref{eq:testing-second-term-derivative} gives $S'>0$. Since $S(0)=0$ and $S$ is continuous at zero, $S$ is nonnegative and increasing throughout $[0,1/(3q)]$. Equation~\eqref{eq:testing-S-positive-part} then proves the asserted monotonicity of $h_{q,q-3}$.

To obtain Eq.~\eqref{eq:testing-local-window-bound}, use the union bound over the $\binom q2$ pairs of trials. At least two successes require at least one successful pair, and each pair succeeds with probability $d_{\mathcal T}^4$. Hence
\begin{align}
 \|\ket w\|_2^2&\le\binom q2d_{\mathcal T}^4.
 \label{eq:testing-binomial-upper}
\end{align}
For $d_{\mathcal T}>0$, substitution into Eq.~\eqref{eq:testing-S} gives
\begin{align}
 \frac{S(d_{\mathcal T}^2)}{\sqrt q\,d_{\mathcal T}}
 &=\sqrt{1-3/q}(1-d_{\mathcal T}^2)^{(q-1)/2}
   -\sqrt{\frac{3\|\ket w\|_2^2}{q^2d_{\mathcal T}^2}} \nonumber \\
 &\ge\left(1-\frac2q\right)
       \left(1-\frac{q-1}{2}d_{\mathcal T}^2\right)
       -\sqrt{\frac{3(q-1)}{2q}}\,d_{\mathcal T} \nonumber \\
 &\ge1-\frac2q-\frac{q-1}{2}d_{\mathcal T}^2
       -\sqrt{\frac{3(q-1)}{2q}}\,d_{\mathcal T}.
 \label{eq:testing-S-lower}
\end{align}
The first inequality also uses $\sqrt{1-3/q}\ge1-2/q$ for $q\ge4$ and Bernoulli's inequality. Both factors on the second line are nonnegative on the interval under consideration. The last inequality drops their nonnegative product correction. Taking the positive part, squaring, and multiplying by $qd_{\mathcal T}^2$ proves Eq.~\eqref{eq:testing-local-window-bound}. At $d_{\mathcal T}=0$ both sides vanish.
\end{proof}

\begin{theorem}[Small-gap tolerant testing]
\label{thm:testing-optimal-gap}
Consider any of the five families
$\GF,\Sl,\Gzero,\Coh,\GB$ with more than one state up to global
phase. The same statement holds for the fixed-parity family
$\GFs$, provided that the input is promised to have parity $\sigma$. Let
\begin{subequations}\label{eq:testing-gap-parameters}
\begin{align}
 0<\eta&\le1,\qquad 0<a\le\frac{\eta}{100}, \nonumber \\
 b&=(1+\eta)a, \nonumber \\
 q&=\left\lceil\frac{12}{\eta}\right\rceil.
\end{align}
\end{subequations}
The $q$-copy projector test distinguishes $d_{\mathcal T}(\psi)\le a$ from $d_{\mathcal T}(\psi)\ge b$, with each error at most $\delta/2$, using at most
\begin{align}
 q+\frac{16}{(b-a)^2}\log\frac2\delta
 \label{eq:testing-gap-count}
\end{align}
copies, for $0<\delta<1$.
\end{theorem}

\begin{proof}
We will bound rejection by $\mathcal{U}$ on all near inputs and by $L>\mathcal{U}$ on all far inputs, then apply Eq.~\eqref{eq:testing-tolerant-count}. All five coefficients in Table~\ref{table:normal-coefficients} satisfy $\lambda_{\mathcal T,q}\ge q-3$. In Eq.~\eqref{eq:normal-h}, increasing $\lambda$ increases the first term and decreases the subtracted term. The positive part and square preserve this ordering. Therefore Lemma~\ref{lem:normal-bound}, applied using
perfect completeness and
Eq.~\eqref{eq:higher-copy-coefficient-bound}, gives
\begin{align}
 R_{\mathcal T,q}(\psi)
 &\ge h_{q,\lambda_{\mathcal T,q}}(d_{\mathcal T}^2) \nonumber \\
 &\ge h_{q,q-3}(d_{\mathcal T}^2).
 \label{eq:testing-common-local-bound}
\end{align}
At $d_{\mathcal T}=0$ or $1$, the same bound holds
because both functions on the right vanish.
First consider rejection at distance $b$. The parameter choices in Eq.~\eqref{eq:testing-gap-parameters} give $q\le13/\eta$ and $b\le2\eta/100$, so
\begin{align}
 qb^2&\le\frac{52\eta}{10000}<\frac13.
 \label{eq:testing-gap-inside-window}
\end{align}
Thus $b^2$ lies in the interval covered by Lemma~\ref{lem:testing-local-window}. The three terms subtracted from one in Eq.~\eqref{eq:testing-local-window-bound} obey
\begin{align}
 \frac2q&\le\frac\eta6, \nonumber \\
 \frac{q-1}{2}b^2&\le\frac{26\eta}{10000}, \nonumber \\
 \sqrt{\frac{3(q-1)}{2q}}\,b&\le\frac{\sqrt6\eta}{100}.
 \label{eq:testing-gap-losses}
\end{align}
Their sum is less than $\eta/4$. Hence
\begin{align}
 h_{q,q-3}(b^2)&\ge qb^2(1-\eta/4)^2 \nonumber \\
 &\ge qb^2(1-\eta/2)=:L.
 \label{eq:testing-L}
\end{align}
The second line uses $(1-\eta/4)^2\ge1-\eta/2$. By monotonicity in Lemma~\ref{lem:testing-local-window}, Eqs.~\eqref{eq:testing-common-local-bound} and \eqref{eq:testing-L} now give rejection at least $L$ for every input with $b^2\le d_{\mathcal T}^2\le1/(3q)$.

For the remaining far inputs, with $d_{\mathcal T}^2\ge1/(3q)$, use the global bound in Eq.~\eqref{eq:normal-universal-global}. Writing $m=\lfloor q/4\rfloor$, we have $m\ge q/6$ for $q\ge12$. It follows that
\begin{align}
 R_{\mathcal T,q}(\psi)
 &\ge1-\left(1-\frac{3d_{\mathcal T}^2}{32}\right)^m \nonumber \\
 &\ge1-\e^{-3md_{\mathcal T}^2/32} \nonumber \\
 &\ge1-\e^{-1/192} \nonumber \\
 &\ge\frac1{384}.
 \label{eq:testing-distant-far}
\end{align}
The second line uses $1-z\le \e^{-z}$. The third uses both $m\ge q/6$ and $d_{\mathcal T}^2\ge1/(3q)$. The last uses $1-\e^{-z}\ge z/2$ for $0\le z\le1$, at $z=1/192$. This bound is larger than the same $L$ chosen in Eq.~\eqref{eq:testing-L}, because
\begin{align}
 (1-\eta/2)(1+\eta)^2
 &=1+\frac{3\eta}{2}-\frac{\eta^3}{2}\le2,
 \label{eq:testing-gap-polynomial}\\
 L&=(1-\eta/2)(1+\eta)^2qa^2 \\
 &\le2qa^2 \\
 &\le\frac{26\eta}{10000} \\
 &\le\frac{13}{5000}<\frac1{384}.
 \label{eq:testing-L-global-comparison}
\end{align}
Here the polynomial in Eq.~\eqref{eq:testing-gap-polynomial} is increasing on $[0,1]$ and equals two at $\eta=1$. The bound on $2qa^2$ uses $q\le13/\eta$ and $a\le\eta/100$. Together, the two distance intervals cover every far input.

For near inputs, Eq.~\eqref{eq:normal-hierarchy-calibration}, followed by Bernoulli's inequality, gives
\begin{align}
 R_{\mathcal T,q}(\psi)
 &\le \mathcal{U}_q(a^2) \nonumber \\
 &=1-(1-a^2)^q \nonumber \\
 &\le qa^2=:\mathcal{U}.
 \label{eq:testing-gap-near-bound}
\end{align}
Using Eq.~\eqref{eq:testing-gap-polynomial} again,
\begin{align}
 L-\mathcal{U}&=qa^2\left(\frac{3\eta}{2}-\frac{\eta^3}{2}\right) \nonumber \\
 &\ge q\eta a^2>0.
 \label{eq:testing-gap-separation}
\end{align}
We can therefore apply Proposition~\ref{prop:testing-statistics}. Since $L\le2qa^2$, the copy count in Eq.~\eqref{eq:testing-tolerant-count} satisfies
\begin{align}
 q\left\lceil\frac{8L}{(L-\mathcal{U})^2}\log\frac2\delta\right\rceil
 &\le q+\frac{8qL}{(L-\mathcal{U})^2}\log\frac2\delta \nonumber \\
 &\le q+\frac{16}{\eta^2a^2}\log\frac2\delta \nonumber \\
 &=q+\frac{16}{(b-a)^2}\log\frac2\delta.
 \label{eq:testing-gap-count-proof}
\end{align}
The first line accounts for rounding up the number of blocks. The second substitutes the bounds on $L$ and $L-\mathcal{U}$, and the last uses $b-a=\eta a$.
\end{proof}

\subsubsection{Bounds from the basic tests}
\label{subsec:testing-baseline}

The basic tests also give tolerant guarantees whenever their
rejection bounds separate the near and far inputs.
Write $\mathcal U_q(s)=1-(1-s)^q$ for the universal upper bound.

\begin{proposition}[Separation conditions for the basic tests]
\label{prop:testing-baseline-radii}
Let $0\le a<b\le1$. For the basic two-copy tests, the near and
far rejection bounds are strictly separated when
\begin{align}
 \mathcal U_2(a^2)&<\ell_\kappa(b^2),
 \label{eq:testing-base-far-radius}
\end{align}
where $\kappa=1/8,1/6,1/8,1/4$ for
$\GF,\Sl,\Gzero,\Coh$, respectively.

For full bosonic Gaussian states, the corresponding conditions
for the three-copy test and the four-copy difference test are,
respectively,
\begin{align}
 \mathcal U_3(a^2)&<\frac49\ell_{\mathrm B}(b^2),
 \label{eq:testing-full-three-radius}\\
 \mathcal U_4(a^2)&<\ell_{\mathrm B}(b^2).
 \label{eq:testing-full-four-radius}
\end{align}
Under each condition, Proposition~\ref{prop:testing-statistics}
gives a tolerant test with the copy count in
Eq.~\eqref{eq:testing-tolerant-count}, taking $\mathcal U$
and $L$ to be the left- and right-hand sides and using the
corresponding block size $q=2,3,4$.

These conditions characterize separation by the stated bounds.
Stronger rejection bounds may certify additional pairs $(a,b)$.
\end{proposition}

\begin{proof}
Perfect completeness and Eq.~\eqref{eq:universal-upper} give
rejection at most $\mathcal U_q(a^2)$ whenever
$d_{\mathcal T}\le a$.
The global lower bounds in
Eqs.~\eqref{eq:higher-copy-basic-pair-bounds},
\eqref{eq:normal-inherited-three}, and
\eqref{eq:higher-copy-basic-four-bound} give the respective
far-input bounds.
Since $\ell_\kappa$ and $\ell_{\mathrm B}$ are nondecreasing,
their values at $b^2$ bound rejection from below for every
$d_{\mathcal T}\ge b$.
The displayed conditions therefore give $\mathcal U<L$,
so Proposition~\ref{prop:testing-statistics} applies.
\end{proof}

\subsection{Copy lower bounds}\label{sec:testing-lower-bounds}

To lower-bound the number of copies, we choose one near input and one far input with a large overlap. A tolerant tester must distinguish these two states, so any lower bound for their discrimination also applies to testing the whole family.

\begin{lemma}[States with prescribed target distance]\label{lem:testing-exact-curves}
For each orthonormal pair in the table below, the following states have the stated distance given by
\begin{align}
 \ket{\psi_t}&=\sqrt{1-t^2}\,\ket\phi+t\ket\zeta,
 \label{eq:testing-normal-curve}\\
 d_{\mathcal T}(\psi_t)&=t,\qquad 0\le t\le1/5.
 \label{eq:testing-curve-distance}
\end{align}
\begin{center}
\begin{tabular}{lll}
\toprule
Target&$\ket\phi$&$\ket\zeta$\\
\midrule
Fermionic Gaussian&$\ket{0000}$&$\ket{1111}$\\
Slater determinant&$\ket{e_1\wedge e_2}$&$\ket{e_3\wedge e_4}$\\
Centered bosonic Gaussian / vacuum&$\ket0$&$\ket1$\\
Coherent state&$\ket0$&$\ket2$\\
Full bosonic Gaussian&$\ket0$&$\ket3$\\
\bottomrule
\end{tabular}
\end{center}
The $e_j$ are orthonormal one-particle orbitals. Their wedges denote normalized two-particle Slater states. The fermionic examples use four modes and the bosonic examples one. Adding modes in vacuum preserves the stated distances. Applying one Majorana operator to the fermionic Gaussian example gives an example in the odd-parity sector.
\end{lemma}

\begin{proof}
The required family-specific overlap calculations are proved later: Proposition~\ref{prop:fermion-sharpness} treats the fermionic examples, and Section~\ref{subsec:boson-distance-curves} treats the bosonic ones. They establish, on the indicated interval,
\begin{align}
 \sup_{g\in\mathcal T}|\ip g{\psi_t}|^2&=1-t^2.
 \label{eq:testing-curve-overlap}
\end{align}
In each case $\ket\phi$ attains this overlap. By Eq.~\eqref{eq:distances}, this proves the claimed distance in Eq.~\eqref{eq:testing-curve-distance}. For the centered Gaussian row, the overlap statement follows directly from parity: every centered pure Gaussian has even photon number, so
\begin{align}
 \ip g{\psi_t}&=\sqrt{1-t^2}\,\ip g0, \nonumber \\
 |\ip g{\psi_t}|^2&\le1-t^2.
\end{align}
Vacuum attains equality. The same calculation applies when vacuum is the only target. This one-photon example establishes the distance needed here. The four-photon example in Section~\ref{subsec:boson-distance-curves} attains the smallest quadratic rejection coefficient for centered Gaussians.

The cited family-specific arguments also prove that adding vacuum modes preserves the maximum overlap. A Majorana operator is unitary and maps even Gaussian states to odd Gaussian states, so applying it to both inputs and targets preserves all overlaps and exchanges parity.
\end{proof}

\begin{theorem}[Copy lower bound]\label{thm:testing-lower-bound}
For a family and mode sector admitting an example in Lemma~\ref{lem:testing-exact-curves}, let $0\le a<b\le1/5$ and $0<\delta<1/2$. Every $(a,b)$ tolerant test with each error at most $\delta$ requires a total number of copies satisfying
\begin{align}
 N&\ge\frac{\log(1/[4\delta(1-\delta)])}
 {-\log\cos^2(\arcsin b-\arcsin a)}.
 \label{eq:testing-copy-lower}
\end{align}
This includes arbitrary joint, adaptive, and randomized measurements on at most $N$ copies. In particular, for $0<\delta\le1/4$,
\begin{align}
 N&=\Omega\!\left((b-a)^{-2}\log(1/\delta)\right).
 \label{eq:testing-copy-lower-order}
\end{align}
Taking $a=0$ gives the non-tolerant lower bound.
\end{theorem}

\begin{proof}
Choose $\ket{\psi_a}$ and $\ket{\psi_b}$ from Eq.~\eqref{eq:testing-normal-curve}. By Eq.~\eqref{eq:testing-curve-distance}, these satisfy the near and far promises, respectively. Orthonormality of $\ket\phi$ and $\ket\zeta$ gives their overlap
\begin{align}
 s&=\ip{\psi_a}{\psi_b} \nonumber \\
 &=\sqrt{(1-a^2)(1-b^2)}+ab \nonumber \\
 &=\cos(\arcsin b-\arcsin a).
 \label{eq:testing-promised-overlap}
\end{align}
The last equality is the cosine subtraction formula. For $N$ copies, inner products multiply, so
\begin{align}
 \ip{\psi_a^{\otimes N}}{\psi_b^{\otimes N}}&=s^N, \nonumber \\
 \frac12\left\|
   (\ket{\psi_a}\bra{\psi_a})^{\otimes N}
   -(\ket{\psi_b}\bra{\psi_b})^{\otimes N}
 \right\|_1
 &=\sqrt{1-s^{2N}}.
 \label{eq:testing-promised-trace-distance}
\end{align}
The second line applies Eq.~\eqref{eq:prelim-trace-distance} to the two normalized tensor-power states. The Helstrom bound~\cite{helstrom} says that, with equal prior probabilities, the minimum average error is one half of one minus this trace distance. Thus every tester obeys
\begin{align}
 P_{\rm err}&\ge\frac{1-\sqrt{1-s^{2N}}}{2}.
 \label{eq:testing-helstrom-error}
\end{align}
If both conditional errors are at most $\delta$, their average is also at most $\delta$. Since $\delta<1/2$, rearranging and squaring gives
\begin{align}
 1-2\delta&\le\sqrt{1-s^{2N}}, \nonumber \\
 s^{2N}&\le1-(1-2\delta)^2 \nonumber \\
 &=4\delta(1-\delta), \nonumber \\
 N(-\log s^2)&\ge\log\frac1{4\delta(1-\delta)}.
 \label{eq:testing-overlap-copy-bound}
\end{align}
Here $0<s<1$, because $a<b\le1/5$. Dividing the last line by $-\log s^2>0$ and substituting Eq.~\eqref{eq:testing-promised-overlap} proves Eq.~\eqref{eq:testing-copy-lower}. An adaptive or randomized procedure also defines a two-outcome measurement on the joint input, so the same discrimination bound applies to those procedures.
To obtain Eq.~\eqref{eq:testing-copy-lower-order}, put
$\Delta=\arcsin b-\arcsin a$.
The mean value theorem on $[0,1/5]$ gives
\begin{align}
 b-a&\le\Delta\le\frac5{\sqrt{24}}(b-a).
 \label{eq:testing-angle-gap}
\end{align}
Since $\Delta\le\arcsin(1/5)$, we have
$\cos^2\Delta\ge24/25$. Using
$-\log(1-x)\le x/(1-x)$ and $\sin\Delta\le\Delta$,
we obtain
\begin{align}
 -\log\cos^2\Delta
 &\le\tan^2\Delta
 \le\frac{25}{24}\Delta^2
 \le\frac{625}{576}(b-a)^2.
 \label{eq:testing-log-overlap-gap}
\end{align}
For $0<\delta\le1/4$, the numerator satisfies
$\log(1/[4\delta(1-\delta)])\ge c\log(1/\delta)$
for a universal constant $c>0$: the ratio of these
expressions is positive and continuous, and tends to one
as $\delta\to0$.
Substituting these estimates into
Eq.~\eqref{eq:testing-copy-lower} proves
Eq.~\eqref{eq:testing-copy-lower-order}.
\end{proof}

Adding vacuum modes extends these examples to all fermionic mode counts $n\ge4$ and bosonic counts $n\ge1$. The Slater example lies in the two-particle sector. The lower bound requires that the promised sector contain such a pair of inputs. In a fixed fermionic parity sector with at most three modes, every pure state is Gaussian. Indeed, in the even sector a state with nonzero vacuum amplitude is a pair-creation exponential: its two-particle amplitudes determine the exponential, and higher even particle numbers are unavailable. States with zero vacuum amplitude follow by taking limits, since the Gaussian family is closed. A Majorana operator gives the same conclusion in the odd sector.

Proposition~\ref{prop:boson-two-copy-impossible} rules out
nontrivial perfectly complete two-copy tests for full bosonic
Gaussianity. Theorem~\ref{thm:testing-lower-bound} instead bounds
the total number of copies needed to achieve specified near
and far distances and error probabilities.

\subsection{Measurement accuracy}\label{sec:testing-implementation}

The copy counts in Proposition~\ref{prop:testing-statistics} depend on the separation between near and far rejection probabilities. We now bound how much an imperfect implementation can reduce that separation.

\begin{proposition}[Robustness to measurement error]\label{prop:testing-measurement-error}
Let $E$ and $\widetilde E$ be the ideal and implemented acceptance operators for one block, with $0\le E,\widetilde E\le I$. Suppose
\begin{align}
 \|\widetilde E-E\|_\infty&\le\xi,
 \label{eq:testing-operator-error}
\end{align}
where $\|\cdot\|_\infty$ is the operator norm and $\xi\ge0$. If the ideal near and far rejection intervals are $[0,\mathcal U]$ and $[L,1]$, with $0\le\mathcal U<L\le1$, then the implemented rejection probabilities satisfy
\begin{align}
 \Tr[(I-\widetilde E)\rho]&\le\mathcal U+\xi
 \qquad\text{on near inputs}, \nonumber \\
 \Tr[(I-\widetilde E)\rho]&\ge L-\xi
 \qquad\text{on far inputs}.
 \label{eq:testing-implemented-endpoints}
\end{align}
Here $\rho$ is the block input density operator, equal to $(\ket\psi\bra\psi)^{\otimes q}$ for the pure inputs considered above. If
\begin{align}
 \xi&<\frac{L-\mathcal U}{2},
 \label{eq:testing-positive-implemented-gap}
\end{align}
independent repetitions of the implemented measurement give a tolerant test with each error at most $\delta/2$, using
\begin{align}
 q\left\lceil
 \frac{8(L-\xi)}{(L-\mathcal U-2\xi)^2}\log\frac2\delta
 \right\rceil
 \label{eq:testing-implemented-count}
\end{align}
copies, for $0<\delta<1$.
\end{proposition}

\begin{proof}
For any block density operator $\rho\ge0$ with $\Tr\rho=1$, its trace norm is $\|\rho\|_1=1$. The change in rejection probability therefore obeys
\begin{align}
 &\left|\Tr[(I-\widetilde E)\rho]-\Tr[(I-E)\rho]\right| \nonumber
  \\
 &\qquad=\left|\Tr[(\widetilde E-E)\rho]\right| \nonumber\\
 &\qquad\le\|\widetilde E-E\|_\infty\|\rho\|_1 \nonumber\\
 &\qquad\le\xi.
 \label{eq:testing-rejection-perturbation}
\end{align}
The middle inequality is the trace-norm/operator-norm inequality. Thus an ideal rejection probability at most $\mathcal U$ can increase to at most $\mathcal U+\xi$, and one at least $L$ can decrease to at least $L-\xi$. This proves Eq.~\eqref{eq:testing-implemented-endpoints}. Under Eq.~\eqref{eq:testing-positive-implemented-gap},
\begin{align}
 0\le\mathcal U+\xi&<L-\xi\le1, \nonumber \\
 (L-\xi)-(\mathcal U+\xi)&=L-\mathcal U-2\xi>0.
\end{align}
Substituting these endpoints into Eq.~\eqref{eq:testing-tolerant-count} gives Eq.~\eqref{eq:testing-implemented-count}. The midpoint decision threshold remains $(L+\mathcal U)/2$.
\end{proof}

For tolerant testing, a fixed operator error satisfying Eq.~\eqref{eq:testing-positive-implemented-gap} leaves a positive separation, and Eq.~\eqref{eq:testing-implemented-count} accounts for the resulting increase in repetitions. For an ideally perfectly complete non-tolerant procedure that rejects at the first failed block, an ideal target can instead be falsely rejected with probability up to $\xi$ per block. The union bound over $M$ blocks gives
\begin{align}
 \Pr\{\text{at least one block rejects an ideal target}\}&\le M\xi.
 \label{eq:testing-accumulated-error}
\end{align}

For bosonic implementations without an energy bound, a finite photon-count cutoff or a small interferometer-parameter error alone does not guarantee Eq.~\eqref{eq:testing-operator-error} uniformly over all inputs. The probability estimate in Eq.~\eqref{eq:testing-rejection-perturbation} holds for any block density operator. The distance bounds used throughout this paper concern pure input states.

\newpage

\section{Fermions}\label{sec:fermions}
This appendix applies the framework of
Appendix~\ref{sec:preliminaries} to pure fermionic Gaussian states
and Slater determinants. We identify the exactly accepted inputs
of the two-copy tests and prove their global robustness through
gradient flow. We also study the geometry of rejection, determine
the sharp local coefficients for higher-copy tests, and describe
the measurements.

\begin{enumerate}
\item \textbf{Section~\ref{subsec:fermion-targets-tests}: Define
the target families and their tests.}
We introduce the two-copy Gaussian and Slater tests and prove
that their perfectly accepted pure inputs are exactly the
corresponding target states
(Lemmas~\ref{lem:fermions-exact}
and~\ref{lem:fermions-slater-exact}). We then state the global
distance bounds (Theorem~\ref{thm:fermions-calibration}).
These apply without a fixed parity or particle-number promise,
and also within the corresponding fixed sectors.

\item \textbf{Section~\ref{subsec:fermion-first-proof}: Describe
closest targets and the geometry of rejection.}
We show that closest targets exist and identify the components
allowed in deviations from them
(Lemma~\ref{lem:fermion-closest}). These restrictions will be used
to compute the local rejection coefficients. We also prove that
target states are the only local minima of rejection: every
other critical point has a direction of negative curvature.
This conclusion holds both within fixed sectors
(Theorem~\ref{thm:fermion-hessian}) and for superpositions of
sectors (Corollary~\ref{cor:fermion-unpromised-landscape}).

\item \textbf{Section~\ref{subsec:fermion-second-proof}: Prove
the global distance bounds by gradient flow.}
We identify the accepted spaces through rotations of the copy
labels (Lemma~\ref{lem:normal-fermion-fixed-spaces}). Three
overlapping instances of each two-copy test have a spectral gap
independent of the number of physical modes
(Lemmas~\ref{lem:so3} and~\ref{lem:fermions-su3}). The
overlapping-projection and gradient-flow theorems of
Appendix~\ref{sec:preliminaries}
(Theorems~\ref{thm:overlap} and~\ref{thm:restoration}) convert
these spectral estimates into the global distance bounds for
Gaussian states and Slater determinants stated in
Theorem~\ref{thm:fermions-calibration}. The third copy is used
in the proof; the measured tests still act on two copies.

\item \textbf{Section~\ref{subsec:fermion-higher-response}:
Determine the local rejection coefficients and prove sharpness.}
We calculate how the higher-copy tests respond to the deviations
allowed by the closest-target condition. This determines the
smallest leading rejection coefficients in
Table~\ref{table:normal-coefficients}, used in the tolerant-testing
guarantees of Theorem~\ref{thm:testing-optimal-gap}.
Explicit four-mode examples establish sharpness
(Proposition~\ref{prop:fermion-sharpness} and
Corollary~\ref{cor:normal-distance-sharpness}) and provide the
fermionic states used for the paper's copy lower bounds
(Lemma~\ref{lem:testing-exact-curves} and
Theorem~\ref{thm:testing-lower-bound}).

\item \textbf{Section~\ref{subsec:fermion-measurements}:
Describe the measurements.}
We implement the Gaussian test using Bell measurements and
simple classical post-processing. For the Slater test, we
describe a collective measurement of total copy spin that
accepts precisely the spin-zero subspace.
\end{enumerate}

\subsection{Targets and tests}
\label{subsec:fermion-targets-tests}

\subsubsection{Gaussian states}

A Gaussian unitary is a product of unitaries $\exp(\theta\gamma_\mu\gamma_\nu/2)$ with $\mu\ne\nu$ and real $\theta$. These unitaries implement all Majorana rotations in $SO(2n)$: each factor rotates one coordinate plane, and these rotations generate $SO(2n)$. Our fermionic Gaussian target $\GF$ consists of all states obtained by applying Gaussian unitaries to occupation-basis states, with arbitrary global phases included. Its members have definite parity. This is the usual pure-state Gaussian family~\cite{bravyi,hackl}. The notation $\mathcal G_{\mathrm F}^{\sigma}$ restricts this target to parity $\sigma\in\{+1,-1\}$.

\begin{definition}[Two-copy fermionic Gaussian test]
Let $P_{\GF,2}$ be the orthogonal projector onto the kernel of the two-copy operator $\Lambda$, and define the rejection probability by
\begin{align}
 \Lambda&=\sum_{\mu=1}^{2n}\gamma_\mu\otimes\gamma_\mu,
 \label{eq:fermions-bridge}\\
 R_{\GF,2}(\psi)&=\|(I-P_{\GF,2})\ket\psi^{\otimes2}\|_2^2.
 \label{eq:fermions-rejection}
\end{align}
Thus $P_{\GF,2}$ projects onto the states annihilated by $\Lambda$.
\end{definition}
To identify the accepted tensor squares, introduce the real antisymmetric covariance matrix
\begin{align}
 \Gamma_{\mu,\nu}(\psi)
 &=\frac i2\ip\psi{[\gamma_\mu,\gamma_\nu]\psi}.
 \label{eq:fermions-covariance-definition}
\end{align}
Here $[A,B]=AB-BA$, and
\begin{align}
 \|\Gamma\|_{\mathrm F}^2&=\sum_{\mu,\nu}|\Gamma_{\mu,\nu}|^2
\end{align}
is the squared Frobenius norm. The two-copy criterion is due to Bravyi~\cite{bravyi}. The proof below also establishes its implication for pure inputs without a parity promise.

\begin{lemma}[Exact Gaussianity criterion]
\label{lem:fermions-exact}
For every normalized state $\ket\psi$ in fermionic Fock space,
\begin{align}
 \|\Lambda\ket\psi^{\otimes2}\|_2^2
 &=2n-\|\Gamma(\psi)\|_{\mathrm F}^2.
 \label{eq:fermions-covariance}
\end{align}
Consequently $R_{\GF,2}(\psi)=0$ if and only if $\ket\psi\in\GF$. In a fixed parity sector the zero set is $\mathcal G_{\mathrm F}^{\sigma}$.
\end{lemma}

\begin{proof}
Since $\Lambda$ is self-adjoint, expanding $\Lambda^2$ and factorizing expectations on the two copies gives
\begin{align}
 \|\Lambda\ket\psi^{\otimes2}\|_2^2
 &=\ip{\psi^{\otimes2}}{\Lambda^2\psi^{\otimes2}} \nonumber \\
 &=\sum_{\mu,\nu}\ip\psi{\gamma_\mu\gamma_\nu\psi}^2 \nonumber \\
 &=2n+\sum_{\mu\ne\nu}(-i\Gamma_{\mu,\nu})^2 \nonumber \\
 &=2n-\|\Gamma\|_{\mathrm F}^2.
 \label{eq:fermions-covariance-expansion}
\end{align}
The second equality contains an ordinary square because each copy contributes the same expectation. The third uses $\gamma_\mu^2=I$ and $\ip\psi{\gamma_\mu\gamma_\nu\psi}=-i\Gamma_{\mu,\nu}$ for $\mu\ne\nu$.

The inner products of the vectors $\gamma_\mu\ket\psi$ form a positive semidefinite Gram matrix
\begin{align}
 \ip{\gamma_\mu\psi}{\gamma_\nu\psi}
 &=\delta_{\mu,\nu}-i\Gamma_{\mu,\nu}, \nonumber \\
 I-i\Gamma&\ge0.
\end{align}
Since $\Gamma$ is real antisymmetric, its eigenvalues occur in pairs $\pm i\nu_j$, where $\nu_j\ge0$ are the singular values of $\Gamma$. Thus $I-i\Gamma$ has eigenvalues $1\pm\nu_j$, and positivity implies
\begin{align}
 0&\le\nu_j\le1, \nonumber \\
 2n-\|\Gamma\|_{\mathrm F}^2
 &=2\sum_{j=1}^n(1-\nu_j^2).
 \label{eq:fermions-covariance-saturation}
\end{align}
The last expression vanishes exactly when every $\nu_j=1$. Every real antisymmetric matrix admits an orthogonal canonical transformation, therefore $\Gamma$ is equivalent to a direct sum of $2\times2$ blocks
\begin{align}
 \bigoplus_j^n&\begin{pmatrix}0&s_j \nonumber \\-s_j&0\end{pmatrix},
 \qquad s_j\in\{-1,1\}.
\end{align}
If the transformation has determinant $-1$, reversing one coordinate changes its determinant to $+1$ and flips the sign of one $s_j$. It can therefore be implemented by a Gaussian unitary.

Let $\widetilde\gamma_\mu$ denote the transformed Majorana operators in the new basis that puts the covariance matrix into canonical $2\times2$ blocks. The operators
\begin{align}
 Z_j&=i\widetilde\gamma_{2j-1}\widetilde\gamma_{2j}, \nonumber \\
 Z_j^2&=I, \nonumber \\
 \ip\psi{Z_j\psi}&=s_j
\end{align}
are self-adjoint and commute with one another. Their expectations have magnitude one, so
\begin{align}
 \|(Z_j-s_jI)\ket\psi\|_2^2
 &=\ip\psi{(Z_j-s_jI)^2\psi} \nonumber \\
 &=2-2s_j\ip\psi{Z_j\psi} \nonumber \\
 &=0.
\end{align}
Thus $\ket\psi$ is a simultaneous eigenstate of all the $Z_j$. Each $Z_j$ fixes the occupation of one transformed mode. Specifying all these occupations determines a unique occupation-basis state up to phase. Hence $\ket\psi$ is Gaussian and has definite parity. Conversely, occupation-basis states have $\nu_j=1$ for every $j$, and Gaussian rotations preserve these singular values. Equation~\eqref{eq:fermions-covariance} then gives zero rejection for every Gaussian target.
\end{proof}

The projector $P_{\GF,2}$ is invariant under ordinary swap, since swapping the copies leaves $\Lambda$ unchanged. It also commutes with the parity of either copy: conjugating $\Lambda$ by one parity changes its sign and hence preserves its kernel. Thus the test restricts to either parity sector.

\subsubsection{Slater determinants}

For orthonormal one-particle orbitals $\ket{u_1},\ldots,\ket{u_k}$, the wedge product denotes their normalized antisymmetric state
\begin{align}
 \ket{u_1\wedge\cdots\wedge u_k}
 &=\frac1{\sqrt{k!}}\sum_{\pi\in S_k}\operatorname{sgn}(\pi)
   \ket{u_{\pi(1)}}\otimes\cdots\otimes\ket{u_{\pi(k)}}.
 \label{eq:fermions-wedge-definition}
\end{align}
Here $S_k$ is the set of permutations of $k$ labels, and $\operatorname{sgn}(\pi)$ is $+1$ for an even permutation and $-1$ for an odd permutation. This state has one fermion in each orbital $u_1,\ldots,u_k$. For $k=0$, the empty wedge denotes the vacuum.

\begin{definition}[Slater determinants]
A Slater determinant is the normalized antisymmetric state of orthonormal occupied orbitals, with an arbitrary global phase. Let
\begin{align}
 \Sl&=\bigcup_{k=0}^n
 \left\{\e^{i\theta}\ket{u_1\wedge\cdots\wedge u_k}:
   \theta\in\mathbb R,\ \ip{u_i}{u_j}=\delta_{i,j}\right\}.
\end{align}
\end{definition}

On two ordinary tensor-product copies, define
\begin{align}
 N_1&=N\otimes I,\qquad N_2=I\otimes N, \nonumber \\
 \Pi_1&=\Pi\otimes I,\qquad \Pi_2=I\otimes\Pi, \nonumber \\
 c_{j,1}&=a_j\otimes I, \qquad  c_{j,2}=\Pi\otimes a_j.
\end{align}
The factor $\Pi$ makes fermionic operators in different copies anticommute. Define
\begin{align}
 J_z&=\tfrac12(N_1-N_2),
 \label{eq:fermions-slater-generators}\\
 J_+&=\sum_j c_{j,1}^\dagger c_{j,2}, \\
 J_-&=J_+^\dagger, \\
 J^2&=J_z^2+\tfrac12(J_+J_-+J_-J_+). 
\end{align}
The operators $J_+$ and $J_-$ transfer particles between copies. The anticommutation relations give
\begin{align}
 [J_z,J_\pm]&=\pm J_\pm, \nonumber \\
 [J_+,J_-]&=2J_z.
\end{align}
These are the angular-momentum commutation relations, describing $SU(2)$ transformations that mix the two copy labels. A state unchanged by all these transformations is called a \emph{singlet}. Equivalently, it is annihilated by $J_z,J_+,J_-$. In particular,
\begin{align}
 \ip v{J^2v}
 &=\|J_z\ket v\|_2^2
   +\tfrac12\|J_+\ket v\|_2^2
   +\tfrac12\|J_-\ket v\|_2^2
 \label{eq:fermions-singlet-norm}
\end{align}
shows that the kernel of $J^2$ is precisely their common kernel.

\begin{definition}[Two-copy Slater test]
Let $P_{\Sl,2}$ be the orthogonal projector onto $\ker J^2$, and define
\begin{align}
 R_{\Sl,2}(\psi)&=\|(I-P_{\Sl,2})\ket\psi^{\otimes2}\|_2^2.
 \label{eq:fermions-slater-test}
\end{align}
\end{definition}

\begin{lemma}[Exact Slater criterion]
\label{lem:fermions-slater-exact}
For every normalized state $\ket\psi$ in fermionic Fock space, $R_{\Sl,2}(\psi)=0$ if and only if $\ket\psi\in\Sl$. The projector $P_{\Sl,2}$ is invariant under ordinary swap.
\end{lemma}

\begin{proof}
In two copies of a Slater determinant, each occupied orbital occurs once in each copy. Mixing these two occurrences by an $SU(2)$ matrix multiplies their antisymmetric two-particle state by the determinant of that matrix, which is one. Thus the entire tensor square is unchanged, and every Slater square is accepted.

Conversely, suppose the tensor square is accepted. Decompose the input according to particle number
\begin{align}
 \ket\psi&=\sum_{k=0}^n\ket{\psi_k},\qquad
 \ket{\psi_k}\in\Lambda^k\mathbb C^n.
\end{align}
Since $J_z\ket\psi^{\otimes2}=0$,
\begin{align}
 0&=\frac12\sum_{k,l}(k-l)\ket{\psi_k}\otimes\ket{\psi_l}, \nonumber \\
 0&=\|J_z\ket\psi^{\otimes2}\|_2^2 \nonumber \\
  &=\frac14\sum_{k,l}(k-l)^2
       \|\ket{\psi_k}\|_2^2\|\ket{\psi_l}\|_2^2.
\end{align}
The terms with different particle numbers in either copy are orthogonal, which gives the last equality. Every summand is nonnegative, so two distinct particle-number components cannot both be nonzero. Thus $\ket\psi$ has a definite particle number $k$.

In this proof, $\rho$ denotes the one-particle density matrix, whose trace is the particle number
\begin{align}
 \rho_{i,j}&=\ip\psi{a_j^\dagger a_i\psi}, \nonumber \\
 \Tr\rho&=k.
 \label{eq:fermions-one-particle-density}
\end{align}
For a one-particle state (orbital) $\ket{v} = \sum_{j=1}^n v_j \ket{e_j} \in \mathbb
C^n$, define the corresponding annihilation operator
\begin{align}
 a(v)&=\sum_{j=1}^n\overline{v_j}a_j.
\end{align}
The anticommutation relations imply
\begin{align}
 \ip v{\rho v}&=\|a(v)\ket\psi\|_2^2\ge0, \nonumber \\
 \ip v{(I-\rho)v}&=\|a(v)^\dagger\ket\psi\|_2^2\ge0.
\end{align}
Hence $0\le\rho\le I$.

Because $\Pi\ket\psi=(-1)^k\ket\psi$, the transfer operator acts as
\begin{align}
 J_+\ket\psi^{\otimes2}
 &=(-1)^k\sum_j a_j^\dagger\ket\psi\otimes a_j\ket\psi.
\end{align}
The parity sign disappears on taking the squared norm. Expanding the inner products then gives
\begin{align}
 \|J_+\ket\psi^{\otimes2}\|_2^2
 &=\sum_{i,j}\ip\psi{a_i a_j^\dagger\psi}
              \ip\psi{a_i^\dagger a_j\psi} \nonumber \\
 &=\sum_{i,j}(\delta_{i,j}-\rho_{i,j})\rho_{j,i} \nonumber \\
 &=\Tr\rho-\Tr\rho^2 \nonumber \\
 &=k-\Tr\rho^2.
 \label{eq:fermions-slater-rdm}
\end{align}
The second equality uses $a_i a_j^\dagger=\delta_{i,j}I-a_j^\dagger a_i$. Acceptance implies $J_+\ket\psi^{\otimes2}=0$, so $\Tr[\rho(I-\rho)]=0$. Since every eigenvalue of $\rho$ is in $[0,1]$, each contribution to this trace is nonnegative and vanishes only at zero or one. Thus $\rho$ is a rank-$k$ projector.

Let $W$ be the $k$-dimensional space of orbitals onto which $\rho$ projects. For $v\perp W$,
\begin{align}
 \|a(v)\ket\psi\|_2^2&=\ip v{\rho v}=0.
\end{align}
Thus $a(v)\ket\psi=0$ for every $v\in W^\perp$, so $\ket{\psi}$ has no occupation amplitude in any one-particle orbital outside $W$, therefore, all occupied orbitals lie in $W$. As the state has exactly $k$ particles, it belongs to $\Lambda^kW$. This space is one-dimensional: its normalized states are the wedge product of any orthonormal basis of $W$, up to global phase. Hence $\ket\psi$ is a Slater determinant.

Finally, let $S$ swap the two ordinary tensor-product copies. Using the parity factors in $c_{j,2}$ gives
\begin{align}
 SJ_zS&=-J_z, \nonumber \\
 SJ_+S&=-\Pi_1\Pi_2J_-, \nonumber \\
 SJ_-S&=-\Pi_1\Pi_2J_+.
\end{align}
These identities show that swapping a state annihilated by $J_z,J_+,J_-$ gives another such state. Their common kernel is therefore preserved by swap, and so is its orthogonal projector $P_{\Sl,2}$.
\end{proof}

We use $\ell_\kappa$ from Eq.~\eqref{eq:ell-definition} and the upper-bound function $\mathcal U_q (d_\mathcal{T}^2) = 1-(1-d_\mathcal{T}^2)^q$ from Eq.~\eqref{eq:normal-hierarchy-calibration}. These functions give the following bounds for the two tests.

\begin{theorem}[Fermionic distance bounds]
\label{thm:fermions-calibration}
Every normalized fermionic input $\ket\psi$, without a sector promise, satisfies
\begin{align}
 \ell_{1/8}\bigl(d_{\GF}(\psi)^2\bigr)
 &\le R_{\GF,2}(\psi)\le\mathcal U_2\bigl(d_{\GF}(\psi)^2\bigr),
 \label{eq:fermions-calibration}\\
 \ell_{1/6}\bigl(d_{\Sl}(\psi)^2\bigr)
 &\le R_{\Sl,2}(\psi)\le\mathcal U_2\bigl(d_{\Sl}(\psi)^2\bigr).
 \label{eq:fermions-slater-calibration}
\end{align}
The Gaussian assertion also holds for inputs of fixed parity, with distance measured to $\mathcal G_{\mathrm F}^{\sigma}$. The Slater assertion holds for inputs of fixed particle number, with distance measured to Slater determinants of that particle number.
\end{theorem}

Perfect completeness and Eq.~\eqref{eq:universal-upper} give the upper bounds. We prove the lower bounds first by expanding around a closest target in Section~\ref{subsec:fermion-first-proof}, and then by a gradient estimate in Section~\ref{subsec:fermion-second-proof}. The first argument also determines the local minima of rejection, and the second constructs a path of bounded length to a target.

\subsection{Closest-target geometry and rejection landscape}
\label{subsec:fermion-first-proof}

We record two geometric properties of the fermionic tests.
First, choosing a closest target restricts the possible deviations
from it. These restrictions will be used to compute the local
rejection coefficients in Section~\ref{subsec:fermion-higher-response}.
Second, we show that every critical point with positive rejection
has a direction of negative curvature, so the target states are
the only local minima. The global distance bounds are proved
separately by gradient flow in
Section~\ref{subsec:fermion-second-proof}.

\subsubsection{Expansion around a closest target}

\begin{lemma}[Expansion around a closest fermionic target]
\label{lem:fermion-closest}
For Gaussian testing in a fixed parity sector, a Gaussian unitary,
a choice of global phase, and, for odd parity, a parity-changing
unitary put the input in the form
\begin{align}
 \ket\psi
 &=\sqrt{1-d_{\mathcal T}(\psi)^2}\ket\Omega
   +d_{\mathcal T}(\psi)\ket\eta,
 \label{eq:fermion-closest-normal}\\
 \ket\eta
 &=\sum_{R\ge2}\ket{\eta_R}, \nonumber\\
 \ket{\eta_R}
 &\in\Lambda^{2R}\mathbb C^n,
\end{align}
where $\ket{\eta_R}$ denotes the component of the deviation with
exactly $2R$ particles relative to the vacuum.
Let $\mathcal T$ be the target family, and $\ket\Omega$ be the vacuum.
Let $\ket\Omega$ be the
closest target to the transformed input.

For Slater testing in $\Lambda^k\mathbb C^n$, a unitary change
of one-particle orbitals gives the same decomposition, with
\begin{align}
 \ket\Omega
 &=\ket{e_1\wedge\cdots\wedge e_k},
 \label{eq:slater-closest-normal}\\
 O&=\Span\{\ket{e_1},\ldots,\ket{e_k}\},
 \label{eq:slater-occupied-orbitals}\\
 \ket{\eta_R}
 &\in\Lambda^{k-R}O\wedge\Lambda^R O^\perp.
\end{align}
The $e_j$ are orthonormal orbitals. The last space contains states
obtained by replacing $R$ occupied orbitals of the reference
determinant with $R$ orbitals orthogonal to $O$. In both cases,
\begin{align}
 \ip\Omega\eta&=0, \nonumber\\
 \|\ket\eta\|_2&=1
 \quad\bigl(d_{\mathcal T}(\psi)>0\bigr), \nonumber\\
 \ket\eta&=0
 \quad\bigl(d_{\mathcal T}(\psi)=0\bigr).
\end{align}

For inputs without a sector promise, a closest target still exists
and the same reduction to a reference state applies. For Gaussian states,
the deviation has no zero- or two-particle component, but
odd-particle components are allowed. For Slaters, if the closest
target has $k$ particles, the deviation's $k$-particle component
contains at least two orbital replacements; components with
particle number different from $k$ are also allowed.
\end{lemma}

\begin{proof}
Within the prescribed sector, every target is obtained from a
reference state by a Gaussian unitary or a unitary change of
orbitals, together with a global phase. These transformations form
compact groups, so the target families are compact and a closest
target exists. Without a sector promise, the target family is a
finite union of compact orbits, so the same conclusion holds.
Moreover, occupation-basis states are targets and span the
prescribed sector, or the full Fock space when no sector is
prescribed. At least one has nonzero overlap with the input, so
$F_{\mathcal T}(\psi)>0$.

If the closest Gaussian target has odd parity, first apply
$\gamma_1$. It exchanges parity and conjugates the Majoranas by
an orthogonal transformation. Consequently,
\begin{align}
 (\gamma_1\otimes\gamma_1)\Lambda
 (\gamma_1\otimes\gamma_1)&=\Lambda.
\end{align}
It also maps Gaussian targets to Gaussian targets, so it preserves
rejection and distance to the corresponding target family.
We may therefore take the closest Gaussian target to have even
parity.

In the Gaussian case, apply a Gaussian unitary that maps a closest
target to vacuum. In the Slater case, use a unitary change of
orbitals that maps a closest target to the reference determinant.
Apply the same transformation to the input and choose its phase
so that its overlap with $\ket\Omega$ is
$\sqrt{1-d_{\mathcal T}(\psi)^2}$.
These transformations preserve the pair projectors: Gaussian
unitaries rotate both Majorana lists in $\Lambda$ by the same
orthogonal matrix, and orbital unitaries preserve the Slater
copy-transfer generators.

If $d_{\mathcal T}(\psi)=0$, set $\ket\eta=0$.
Otherwise, let $\ket{g(s)}$ be a smooth normalized target state
with $\ket{g(0)}=\ket\Omega$ and
$\ket\zeta=\left.\frac{\dd}{\dd s}\ket{g(s)}\right|_{s=0}$,
with $\ip\Omega\zeta=0$.
The closest-target condition in
Lemma~\ref{lem:normal-stationarity} gives
\begin{align}
 0
 &=\left.\frac{\dd}{\dd s}
      |\ip{g(s)}\psi|^2\right|_{s=0} \nonumber\\
 &=2\Re\!\left[
   \overline{\ip\Omega\psi}\,
   \ip\zeta\psi\right] \nonumber\\
 &=2d_{\mathcal T}(\psi)\sqrt{1-d_{\mathcal T}(\psi)^2}
   \,\Re\ip\zeta\eta.
 \label{eq:fermion-nearest-derivative}
\end{align}
The following operators generate such first-order changes.
For the Gaussian vacuum, with $i<j$,
\begin{align}
 (a_i^\dagger a_j^\dagger-a_ja_i)\ket\Omega
 &=a_i^\dagger a_j^\dagger\ket\Omega, \nonumber\\
 i(a_i^\dagger a_j^\dagger+a_ja_i)\ket\Omega
 &=i a_i^\dagger a_j^\dagger\ket\Omega.
\end{align}
For the Slater reference, with $i\le k<b$,
\begin{align}
 (a_b^\dagger a_i-a_i^\dagger a_b)\ket\Omega
 &=a_b^\dagger a_i\ket\Omega, \nonumber\\
 i(a_b^\dagger a_i+a_i^\dagger a_b)\ket\Omega
 &=i a_b^\dagger a_i\ket\Omega.
\end{align}
Each operator on the left is anti-Hermitian, so its exponential
is an allowed unitary. Applying
Eq.~\eqref{eq:fermion-nearest-derivative} to both changes in each
pair sets both the real and imaginary parts of the corresponding
overlap to zero:
\begin{align}
 \ip{a_i^\dagger a_j^\dagger\Omega}\eta&=0
 \qquad\text{(Gaussian states)}, \nonumber\\
 \ip{a_b^\dagger a_i\Omega}\eta&=0
 \qquad\text{(Slater )}.
\end{align}
Thus the Gaussian deviation has no zero- or two-particle component.
With a fixed parity promise, the transformed input has even
parity, so only degrees $2R\ge4$ remain. Without that promise,
odd-particle components are also allowed.

For Slaters, the deviation has no reference component or single
occupied-orbital replacement. Hence its component in the reference
particle-number sector has only $R\ge2$ replacements. Components
in other particle-number sectors are allowed when the input has
no sector promise.
\end{proof}

\subsubsection{Local geometry of fermions}

We now study the critical points of rejection. We show that any
critical point outside the target family has a direction in which
the second derivative is negative. Thus no non-target state can
be a local minimum. We first establish a general Hessian identity,
apply it within each fixed parity or particle-number sector,
and then extend the conclusion to superpositions of sectors.

A critical point is a normalized state where the first derivative
of rejection vanishes for every smooth change through normalized
states. At such a point,
$\operatorname{Hess}_S R(\psi)[v,v]$ denotes the second derivative
along a normalized path whose initial derivative is $\ket v$.
The subscript $S$ indicates the unit sphere. At a critical point
this second derivative depends only on $\ket v$, so a negative
value excludes a local minimum. The following identity produces
such a value.

\begin{theorem}[Hessian identity]
\label{thm:compact-maximum}
Let $\mathcal H$ be finite-dimensional, let $q\ge2$ be an integer,
and let $H_1,\ldots,H_s$ be Hermitian operators satisfying
\begin{align}
 \sum_jH_j^2&=cI.
\end{align}
On $q$ copies, define
\begin{align}
 T_j&=\sum_{a=1}^qH_j^{(a)}, \nonumber\\
 \mathcal C&=\sum_jT_j^2.
\end{align}
Here $H_j^{(a)}$ applies $H_j$ to copy $a$ and the identity to the
others. Let $P$ project onto the eigenspace of the largest
eigenvalue $C_*$ of $\mathcal C$, assume $[P,T_j]=0$ for every $j$,
and suppose
\begin{align}
 p(\psi)&=\|P\ket\psi^{\otimes q}\|_2^2>0
 \quad\text{for every normalized }\ket\psi.
 \label{eq:compact-positive}
\end{align}
At a critical point of $R=1-p$, put
\begin{align}
 \mu_j&=\ip\psi{H_j\psi}, \nonumber\\
 \ket{v_j}&=(H_j-\mu_jI)\ket\psi.
\end{align}
Then
\begin{align}
 \sum_j\operatorname{Hess}_S R(\psi)[v_j,v_j]
 &=-4p(\psi)\ip{\psi^{\otimes q}}
       {(C_*I-\mathcal C)\psi^{\otimes q}}.
 \label{eq:compact-hessian}
\end{align}
Consequently, every critical point with $R>0$ can be changed
smoothly, while remaining normalized, so that the second
derivative of rejection is strictly negative.
\end{theorem}

\begin{proof}
For a Hermitian operator $A$ and a normalized vector $\ket u$,
we have $\|\e^{sA}\ket u\|_2^2=\ip u{\e^{2sA}u}$.
Differentiating its logarithm gives
\begin{align}
 \left.\frac{\dd}{\dd s}\log\|\e^{sA}\ket u\|_2^2\right|_{s=0}
 &=2\ip u{Au}, \nonumber\\
 \left.\frac{\dd^2}{\dd s^2}
       \log\|\e^{sA}\ket u\|_2^2\right|_{s=0}
 &=4\bigl(\ip u{A^2u}-\ip u{Au}^2\bigr).
 \label{eq:compact-log-norm-derivatives}
\end{align}
Normalize the accepted component and choose the paths
\begin{align}
 \ket{\widehat k}
 &=\frac{P\ket\psi^{\otimes q}}{\sqrt{p(\psi)}}, \nonumber\\
 \ket{\psi_j(s)}
 &=\frac{\e^{sH_j}\ket\psi}{\|\e^{sH_j}\ket\psi\|_2}, \nonumber\\
 \left.\frac{\dd}{\dd s}\ket{\psi_j(s)}\right|_{s=0}
 &=(H_j-\mu_jI)\ket\psi=\ket{v_j}.
 \label{eq:compact-normalized-paths}
\end{align}
Because $P$ commutes with $T_j$ and
$(\e^{sH_j})^{\otimes q}=\e^{sT_j}$,
\begin{align}
 p(\psi_j(s))
 &=p(\psi)\frac{\|\e^{sT_j}\ket{\widehat k}\|_2^2}
                  {\|\e^{sH_j}\ket\psi\|_2^{2q}}, \nonumber\\
 \log p(\psi_j(s))
 &=\log p(\psi)+\log\|\e^{sT_j}\ket{\widehat k}\|_2^2
             -q\log\|\e^{sH_j}\ket\psi\|_2^2.
 \label{eq:compact-log}
\end{align}
At a critical point the first derivative vanishes, so
\begin{align}
 \left.\frac{\dd}{\dd s}\log p(\psi_j(s))\right|_{s=0}
 &=2\ip{\widehat k}{T_j\widehat k}-2q\mu_j=0, \nonumber\\
 \ip{\widehat k}{T_j\widehat k}&=q\mu_j.
\end{align}
Using Eq.~\eqref{eq:compact-log-norm-derivatives}, the second
derivative is
\begin{align}
 \left.\frac{\dd^2}{\dd s^2}\log p(\psi_j(s))\right|_{s=0}
 &=4\left(\ip{\widehat k}{T_j^2\widehat k}
       -q^2\mu_j^2-q\ip\psi{H_j^2\psi}+q\mu_j^2\right).
\end{align}
Summing over $j$ gives
\begin{align}
 \sum_j\left.\frac{\dd^2}{\dd s^2}
       \log p(\psi_j(s))\right|_{s=0}
 &=4\sum_j\left(\ip{\widehat k}{T_j^2\widehat k}
    -q\ip\psi{H_j^2\psi}-q(q-1)\mu_j^2\right) \nonumber\\
 &=4\left(C_*-qc-q(q-1)\sum_j\mu_j^2\right) \nonumber\\
 &=4\ip{\psi^{\otimes q}}
       {(C_*I-\mathcal C)\psi^{\otimes q}}.
 \label{eq:compact-logtrace}
\end{align}
The second equality uses
$\mathcal C\ket{\widehat k}=C_*\ket{\widehat k}$ and
$\sum_jH_j^2=cI$. For the last equality, expand each $T_j^2$.
Its $q$ same-copy terms contribute $\ip\psi{H_j^2\psi}$ each,
and its $q(q-1)$ different-copy terms factorize into $\mu_j^2$.
Therefore
\begin{align}
 \ip{\psi^{\otimes q}}{\mathcal C\psi^{\otimes q}}
 &=\sum_j\left(q\ip\psi{H_j^2\psi}
              +q(q-1)\mu_j^2\right).
\end{align}
At a critical point, the first derivatives of $p(\psi_j(s))$
vanish. Differentiating $R=1-p$ then gives
\begin{align}
 \left.\frac{\dd^2}{\dd s^2}R(\psi_j(s))\right|_{s=0}
 &=-p(\psi)
   \left.\frac{\dd^2}{\dd s^2}\log p(\psi_j(s))\right|_{s=0}
   \nonumber\\
 &=\operatorname{Hess}_S R(\psi)[v_j,v_j].
\end{align}
Together with Eq.~\eqref{eq:compact-logtrace}, this proves
Eq.~\eqref{eq:compact-hessian}.
If $R>0$, the tensor power has a nonzero component outside the
largest-eigenvalue eigenspace. Since $C_*I-\mathcal C$ is strictly
positive on that orthogonal complement,
\begin{align}
 \ip{\psi^{\otimes q}}
    {(C_*I-\mathcal C)\psi^{\otimes q}}&>0.
\end{align}
The sum in Eq.~\eqref{eq:compact-hessian} is therefore negative,
and at least one second derivative is negative.
\end{proof}

The positivity assumption in Eq.~\eqref{eq:compact-positive}
holds if the accepted target vectors span the one-copy space.
For any normalized accepted target $\ket g$,
\begin{align}
 p(\psi)&\ge|\ip{g^{\otimes q}}{\psi^{\otimes q}}|^2
 \nonumber\\
 &=|\ip g\psi|^{2q}.
\end{align}
A spanning family contains at least one vector with nonzero
overlap with each nonzero input. Thus $p(\psi)>0$.

For the fermionic tests, we now identify operators $\mathcal C$
whose largest-eigenvalue eigenspaces are exactly the accepted
subspaces.

\begin{theorem}[Negative second derivative away from fermionic targets]
\label{thm:fermion-hessian}
\label{cor:fermion-negative-curvature}
For Gaussian testing in either fixed parity sector, or Slater
testing at fixed particle number, every critical point outside
the target admits a smooth path of normalized states with
strictly negative second derivative of rejection.
Thus the target states are exactly the local minima.
\end{theorem}

\begin{proof}
We apply Theorem~\ref{thm:compact-maximum} with $q=2$.
For the Hermitian one-copy operators specified below, we use
\begin{align}
 T_j&=H_j\otimes I+I\otimes H_j, \nonumber\\
 \mathcal C&=\sum_jT_j^2.
\end{align}
We verify that $\sum_jH_j^2$ is a scalar multiple of the identity,
that $P$ commutes with each $T_j$, and that $P$ projects onto the
largest-eigenvalue eigenspace of $\mathcal C$.
Positivity of acceptance follows from the occupation-basis
targets, since some such $\ket g$ satisfies
\begin{align}
 p(\psi)&\ge|\ip g\psi|^4>0.
\end{align}
The variations are the normalized paths in
Eq.~\eqref{eq:compact-normalized-paths}, with
$\ket{v_j}=(H_j-\ip\psi{H_j\psi}I)\ket\psi$.

\emph{Fermionic Gaussians.}
On either parity sector, take
\begin{align}
 H_{a,b}&=\frac i2\gamma_a\gamma_b
 \qquad(a<b), \nonumber\\
 T_{a,b}&=H_{a,b}\otimes I+I\otimes H_{a,b}.
\end{align}
The Majorana anticommutation relations give
\begin{align}
 H_{a,b}^2&=I/4, \nonumber\\
 \sum_{a<b}H_{a,b}^2&=\frac{n(2n-1)}4I, \nonumber\\
 \Lambda^2
 &=2nI+2\sum_{a<b}
       \gamma_a\gamma_b\otimes\gamma_a\gamma_b.
\end{align}
Expanding the squares of $T_{a,b}$ therefore yields
\begin{align}
 \mathcal C
 &=\frac{n(2n-1)}2I
   -\frac12\sum_{a<b}
       \gamma_a\gamma_b\otimes\gamma_a\gamma_b \nonumber\\
 &=n^2I-\frac14\Lambda^2.
 \label{eq:compact-spin-casimir}
\end{align}
Its largest eigenvalue is $C_*=n^2$, attained on the
occupation-basis tensor squares, and its corresponding eigenspace
is $\ker\Lambda$. Simultaneous Gaussian transformations rotate
both Majorana lists by the same orthogonal matrix, preserving
$\Lambda$ and its kernel. Consequently $P_{\GF,2}$ commutes with
every $T_{a,b}$.
Equation~\eqref{eq:compact-hessian}, with
$\mu_{a,b}=\ip\psi{H_{a,b}\psi}$, now gives
\begin{align}
 &\sum_{a<b}\operatorname{Hess}_{S}R(\psi)
 [(H_{a,b}-\mu_{a,b}I)\psi,
  (H_{a,b}-\mu_{a,b}I)\psi] \nonumber\\
 &\qquad=-p(\psi)\|\Lambda\ket\psi^{\otimes2}\|_2^2.
 \label{eq:compact-fermion-strict}
\end{align}
The right side is strictly negative whenever $R>0$,
by Lemma~\ref{lem:fermions-exact}.

\emph{Slater determinants.}
Work on $H=\Lambda^k\mathbb C^n$ and put
\begin{align}
 E_{i,l}&=a_i^\dagger a_l, \nonumber\\
 h&=\sum_i a_i\otimes a_i^\dagger.
\end{align}
On $H\otimes H$, we have $J_z=0$ and $J_-=(-1)^{k-1}h$.
The angular-momentum commutation relations give
\begin{align*}
 \|J_+\ket v\|_2^2-\|J_-\ket v\|_2^2
 &=-2\ip v{J_zv}=0.
\end{align*}
Thus $h\ket v=0$ is equivalent to
$J_z\ket v=J_+\ket v=J_-\ket v=0$.
By Eq.~\eqref{eq:fermions-singlet-norm},
$P_{\Sl,2}$ on $H\otimes H$ therefore projects onto $\ker h$.

Choose a Hermitian basis $(b_j)$ of $n\times n$ matrices,
orthonormal for the trace inner product, and define
\begin{align}
 \Tr(b_jb_l)&=\delta_{j,l}, \nonumber\\
 H_j&=\sum_{i,l}(b_j)_{i,l}E_{i,l}.
\end{align}
These operators generate the same change of one-particle orbitals
on each copy. The matrix-basis completeness identity is
\begin{align}
 \sum_j(b_j)_{a,b}(b_j)_{c,d}&=\delta_{a,d}\delta_{b,c}.
\end{align}
With $n_i=a_i^\dagger a_i$, it gives
\begin{align}
 \sum_jH_j^2
 &=\sum_{i,l}E_{i,l}E_{l,i} \nonumber\\
 &=\sum_i n_i+\sum_{i\ne l}n_i(1-n_l) \nonumber\\
 &=N+N(n-N) \nonumber\\
 &=k(n+1-k)I, \nonumber\\
 \sum_jH_j\otimes H_j
 &=\sum_{i,l}E_{i,l}\otimes E_{l,i}.
\end{align}
The second equality uses the anticommutation relations, and the
fourth uses the fixed particle number. Applying the same
relations to $h$ gives
\begin{align}
 h^\dagger h
 &=\sum_{i,l}E_{i,l}\otimes(\delta_{i,l}I-E_{l,i})
 \nonumber\\
 &=kI-\sum_{i,l}E_{i,l}\otimes E_{l,i}.
\end{align}
Consequently,
\begin{align}
 \mathcal C
 &=\sum_j(H_j\otimes I+I\otimes H_j)^2 \nonumber\\
 &=2k(n+1-k)I
   +2\sum_{i,l}E_{i,l}\otimes E_{l,i} \nonumber\\
 &=2k(n+2-k)I-2h^\dagger h.
 \label{eq:compact-passive-casimir}
\end{align}
Applying the same orbital unitary in both copies preserves $h$,
and therefore preserves its kernel. Thus $P_{\Sl,2}$ commutes
with every $T_j$. Occupation-basis determinants span $H$, and
their tensor squares lie in $\ker h$, so the largest eigenvalue
is $C_*=2k(n+2-k)$.
Equation~\eqref{eq:compact-hessian} becomes
\begin{align}
 \sum_j\operatorname{Hess}_{S}R(\psi)[v_j,v_j]
 &=-8p(\psi)\|h\ket\psi^{\otimes2}\|_2^2<0
\end{align}
at every critical point with $R(\psi)>0$.

In both cases, at least one smooth path of normalized states has
negative second derivative. Every local minimum is critical, so
none can have positive rejection.
Lemmas~\ref{lem:fermions-exact} and
\ref{lem:fermions-slater-exact} identify the zero-rejection states
as the respective targets, completing the proof.
\end{proof}

The same conclusion holds without a sector promise. We use the
fact that both tests reject tensor products whose two factors
belong to different sectors, and examine changes in the weights
of the input's sector components.

\begin{corollary}[Unpromised fermionic landscape]
\label{cor:fermion-unpromised-landscape}
For either fermionic test on the full Fock space, every critical
point with positive rejection admits a smooth path of normalized
states with strictly negative second derivative of rejection.
Its only local minima are target states.
\end{corollary}

\begin{proof}
We first check how the tests act on different sectors.
For the Gaussian test, the operators
$M_\mu=\gamma_\mu\otimes\gamma_\mu$ are commuting self-adjoint
involutions. Since $\Lambda=\sum_{\mu=1}^{2n}M_\mu$,
\begin{align*}
 \e^{i\pi\Lambda/2}
 &=\prod_{\mu=1}^{2n}(iM_\mu)=\Pi_1\Pi_2.
\end{align*}
Here we used $\Pi=(-i)^n\gamma_1\cdots\gamma_{2n}$ for one-copy
parity. Every vector in $\ker\Lambda$ is fixed by this exponential,
so it lies in the $+1$ eigenspace of $\Pi_1\Pi_2$.
Thus $P_{\GF,2}$ rejects unequal parities.
Conjugating $\Lambda$ by the parity of either copy changes its
sign and preserves its kernel, so $P_{\GF,2}$ also preserves
each equal-parity block.

For Slaters, the accepted subspace satisfies $J_z=0$, so unequal
particle numbers are rejected. Since $J^2$ commutes with total
particle number, its kernel projector preserves each equal-number
block.

Decompose the input into its nonzero sector components, using
parity sectors for Gaussians and particle-number sectors for Slaters:
\begin{align*}
 \ket\psi&=\sum_j\sqrt{w_j}\ket{\psi_j},\\
 w_j&>0,\qquad \sum_jw_j=1,\qquad \|\ket{\psi_j}\|_2=1.
\end{align*}
Put $r_j=R(\psi_j)$ and $p_j=1-r_j$.
Unequal-sector tensor components are rejected, and accepted
components from different equal-sector blocks are orthogonal.
Therefore
\begin{align}
 R(\psi)&=1-\sum_jw_j^2+\sum_jw_j^2r_j.
 \label{eq:compact-sector-mixing}
\end{align}
Equivalently, $p(\psi)=\sum_jw_j^2p_j$.
Each $p_j$ is positive because occupation-basis targets span
that sector.

If two weights $w_i,w_j$ are positive, the normalized path
\begin{align}
 \ket{\psi(s)}
 &=\sqrt{w_i+s}\ket{\psi_i}
   +\sqrt{w_j-s}\ket{\psi_j}
   +\sum_{l\ne i,j}\sqrt{w_l}\ket{\psi_l}
\end{align}
is defined for sufficiently small $|s|$.
The sector states remain fixed, so its acceptance probability
satisfies
\begin{align}
 p(\psi(s))
 &=(w_i+s)^2p_i+(w_j-s)^2p_j
   +\sum_{l\ne i,j}w_l^2p_l, \nonumber\\
 \left.\frac{\dd^2}{\dd s^2}p(\psi(s))\right|_{s=0}
 &=2(p_i+p_j)>0, \nonumber\\
 \left.\frac{\dd^2}{\dd s^2}R(\psi(s))\right|_{s=0}
 &=-2(p_i+p_j)<0.
\end{align}
At a critical point this is the Hessian evaluated on the initial
derivative of the path. If only one sector is present, the point
is also critical under variations restricted to that sector,
and Theorem~\ref{thm:fermion-hessian} supplies a negative second
derivative. Finally, Lemmas~\ref{lem:fermions-exact} and
\ref{lem:fermions-slater-exact} identify the zero-rejection states,
which are exactly the local minima.
\end{proof}

We next prove the global distance bounds using the gradient of
rejection and the flow from Theorem~\ref{thm:restoration}.

\subsection{Gradient-flow proof}
\label{subsec:fermion-second-proof}

We now give a proof of Theorem~\ref{thm:fermions-calibration} based on the gradient flow, in Theorem~\ref{thm:restoration}, of the two-copy rejection probability. The guiding idea is that a state with sufficiently small rejection can be continuously deformed into the target family while controlling the total distance travelled. Theorem~\ref{thm:restoration} makes this idea quantitative once a suitable lower bound on the norm of the rejection gradient is available. To establish this bound, we combine the expression for the gradient in Lemma~\ref{lem:first-variation} with Theorem~\ref{thm:overlap}, which relates it to the spectrum of the average of three pair projectors acting on pairs $12$, $13$, and $23$. The third copy is an auxiliary tool in the proof. The measured test still acts on two copies.

In Section~\ref{subsec:Rotations-of-the-copy-labels}, we identify the symmetry underlying these projectors. Lemma~\ref{lem:normal-fermion-fixed-spaces} shows that the Gaussian and Slater accepted spaces are precisely the spaces invariant under $SO(q)$ and $SU(q)$ rotations of the copy labels, respectively.
Section~\ref{subsec:Pairwise-rotations-within} establishes the required $SU(3)$ estimate. Lemma~\ref{lem:fermions-su3} bounds the spectrum of the average of the projectors onto states invariant under each of the three pairwise $SU(2)$ subgroups.

We then combine the two spectral estimates with Theorem~\ref{thm:overlap} to obtain the gradient inequalities for both tests. Applying Theorem~\ref{thm:restoration} converts these inequalities into the distance bounds of Theorem~\ref{thm:fermions-calibration}. We also show that the argument applies within fixed-parity sectors for Gaussian states and fixed-particle-number sectors for Slater determinants. Finally, we record an explicit upper bound on the total length of the gradient-flow trajectory from a state with sufficiently small rejection to the corresponding target family.

\subsubsection{Rotations of the copy labels}\label{subsec:Rotations-of-the-copy-labels}

\begin{lemma}[Rotation-invariant fermionic fixed spaces]
\label{lem:normal-fermion-fixed-spaces}
On the ordinary Hilbert-space tensor product, $P_{\GF,q}$ projects onto the states unchanged by all $SO(q)$ rotations of the copy labels, and $P_{\Sl,q}$ projects onto those unchanged by all $SU(q)$ rotations. The fermionic operators implementing these rotations include the parity factors defined below.
\end{lemma}

\begin{proof}
Let $\Pi_a$ be the parity operator in copy $a$, and let $\gamma_{\mu,a}$ apply Majorana operator $\mu$ to that copy and the identity to all other copies. Define
\begin{align}
 c_{\mu,a}&=\Pi_1\cdots\Pi_{a-1}\gamma_{\mu,a}.
\end{align}
An empty product is the identity. These parity factors, often called Jordan--Wigner strings, ensure that the operators anticommute between different copies as well as within each copy
\begin{align}
 \{c_{\mu,a},c_{\nu,b}\}&=2\delta_{\mu,\nu}\delta_{a,b}I.
\end{align}
For $a<b$, the rotation generators satisfy
\begin{align}
 J_{a,b}&=\frac12\sum_{\mu=1}^{2n}c_{\mu,a}c_{\mu,b} \nonumber \\
 &=-\frac12(\Pi_a\Pi_{a+1}\cdots\Pi_{b-1})\Lambda_{a,b}.
 \label{eq:normal-fermion-generators}
\end{align}
Here $\Lambda_{a,b}$ is the ordinary placement of $\Lambda$ on copies $a,b$, as defined in Eq.~\eqref{eq:higher-copy-gaussian-constraints}. The minus sign comes from moving $\gamma_{\mu,a}$ past $\Pi_a$. Since the product of parity operators is invertible,
\begin{align}
 \ker J_{a,b}&=\ker\Lambda_{a,b}.
\end{align}
The unitaries generated by the $J_{a,b}$ rotate the copy indices of each Majorana coordinate. A full rotation in one copy plane acts as
\begin{align}
 \e^{2\pi J_{a,b}}
 &=\prod_{\mu=1}^{2n}\e^{\pi c_{\mu,a}c_{\mu,b}} \nonumber \\
 &=(-1)^{2n}I=I.
\end{align}
The factors commute, and 
each $c_{\mu,a}c_{\mu,b}$ squares to $-I$. Thus these unitaries give a representation of $SO(q)$. A state is unchanged by every such rotation precisely when all the generators annihilate it. The intersection of their kernels is the accepted space of $P_{\GF,q}$ from Definition~\ref{def:normal-projectors}.

For Slaters, let $a_{i,a}$ annihilate a particle in orbital $i$ of copy $a$, acting as the identity on the other copies. Define
\begin{align}
 \widetilde a_{i,a}&=\Pi_1\cdots\Pi_{a-1}a_{i,a}, \nonumber \\
 \mathsf E_{a,b}&=\sum_i\widetilde a_{i,a}^{\dagger}\widetilde a_{i,b}, \nonumber \\
 \mathsf E_{a,a}&=N_a.
\end{align}
The operator $N_a$ counts particles in copy $a$. The combinations $\mathsf E_{a,b}-\mathsf E_{b,a}$, $i(\mathsf E_{a,b}+\mathsf E_{b,a})$, and $i(N_a-N_b)$ generate particle-number-preserving $SU(q)$ transformations of the copy labels. For $a<b$, the ordinary placement of the two-copy transfer operator is
\begin{align}
 J_+^{ab}&=\sum_i a_{i,a}^{\dagger}\Pi_a a_{i,b}.
\end{align}
Let $P_{\Sl,2}^{ab}$ be the corresponding placement of the singlet projector from Eq.~\eqref{eq:fermions-slater-test}. Direct substitution gives
\begin{align}
 \mathsf E_{a,b}&=(\Pi_{a+1}\cdots\Pi_{b-1})J_+^{ab}, \nonumber \\
 \mathsf E_{b,a}&=(\Pi_{a+1}\cdots\Pi_{b-1})(J_+^{ab})^\dagger.
\end{align}
The intervening parity product is invertible and commutes with the operators on copies $a,b$. Therefore the states unchanged by the $SU(2)$ rotations mixing these two copies form the space
\begin{align}
 \ker(N_a-N_b)\cap\ker\mathsf E_{a,b}\cap\ker\mathsf E_{b,a}
 &=\{\ket v:P_{\Sl,2}^{ab}\ket v=\ket v\}.
\end{align}
Intersecting over all pairs gives the states unchanged by $SU(q)$. This is the accepted space of $P_{\Sl,q}$ in Definition~\ref{def:normal-projectors}.
\end{proof}

For $q=3$, the Gaussian pair projectors thus project onto states unchanged by rotations in the $12$, $23$, and $13$ coordinate planes. These are rotations about three orthogonal axes. In particular, Eq.~\eqref{eq:normal-fermion-generators} gives
\begin{align}
 J_{1,2}&=-\tfrac12\Pi_1\Lambda_{1,2},
 \label{eq:fermions-kernel-placement}\\
 J_{2,3}&=-\tfrac12\Pi_2\Lambda_{2,3}, \nonumber \\
 J_{1,3}&=-\tfrac12\Pi_1\Pi_2\Lambda_{1,3}. 
\end{align}
Thus the identification also holds for the nonadjacent pair $13$. Lemma~\ref{lem:so3}, specifically Eq.~\eqref{eq:fermions-so3-bound}, bounds the spectrum of their average by $\{1\}\cup[0,7/12]$. For Slaters we need the analogous estimate for the three $SU(2)$ subgroups of $SU(3)$ that mix pairs of copy labels.

\subsubsection{Pairwise rotations within \texorpdfstring{$SU(3)$}{SU(3)}}\label{subsec:Pairwise-rotations-within}

The subgroup labeled $12$ mixes the first two copy labels and leaves the third unchanged. The labels $13$ and $23$ have the corresponding meaning. We bound the spectrum of the average of their invariant-state projectors.

\begin{lemma}[Spectral bound for pairwise rotations in $SU(3)$]
\label{lem:fermions-su3}\label{lem:su3}
In any finite-dimensional unitary representation of $SU(3)$, let $P_{1,2},P_{1,3},P_{2,3}$ project onto the states unchanged by the three pairwise $SU(2)$ subgroups. Then
\begin{align}
 \Spec\!\left(\frac{P_{1,2}+P_{1,3}+P_{2,3}}3\right)
 &\subseteq\{1\}\cup[0,5/9].
 \label{eq:fermions-su3-bound}
\end{align}
Here $\Spec$ denotes the set of eigenvalues.
\end{lemma}

\begin{proof}
A finite-dimensional unitary representation decomposes into orthogonal irreducible subspaces, each containing no smaller nonzero subspace preserved by all group transformations. We first describe these subspaces by polynomials. We then find the state fixed by each pairwise subgroup and compute the overlaps between the three states.

For nonnegative integers $p,r$, let $\mathbb C[z,\bar z]_{p,r}$ be the polynomials of degree $p$ in $z_1,z_2,z_3$ and degree $r$ in their complex conjugates. Set
\begin{align}
 D&=\sum_{i=1}^3\partial_{z_i}\partial_{\bar z_i}, \nonumber \\
 \mathcal H_{p,r}&=\{f\in\mathbb C[z,\bar z]_{p,r}:Df=0\}.
\end{align}
Polynomials satisfying $Df=0$ are called \emph{harmonic}. The group action is
\begin{align}
 (U\cdot f)(z,\bar z)&=f(U^{\mathsf T}z,\bar U^{\mathsf T}\bar z),
 \qquad U\in SU(3).
\end{align}
The inner product is the integral of $\overline f g$ over the unit sphere in $\mathbb C^3$, with uniform measure. It is preserved by the group action.

For completeness, these spaces provide all the irreducible representations. The raising operators, which move a coordinate label towards the first coordinate, are
\begin{align}
 E_{1,2}&=z_1\partial_{z_2}-\bar z_2\partial_{\bar z_1}, \nonumber \\
 E_{2,3}&=z_2\partial_{z_3}-\bar z_3\partial_{\bar z_2}, \nonumber \\
 E_{1,3}&=[E_{1,2},E_{2,3}], \nonumber \\
 L&=\sum_i z_i\bar z_i.
\end{align}
A polynomial annihilated by $E_{1,2}$ and $E_{2,3}$ is called a \emph{highest weight vector}. All three operators preserve $z_1,\bar z_3,L$. Where $z_1\bar z_3\ne0$, their exponentials can successively set $z_2$, $\bar z_2$, and $z_3$ to zero, using $E_{1,2}$, $E_{2,3}$, and $E_{1,3}$, respectively. They transform the coordinates to
\begin{align}
 (z,\bar z)&\longmapsto
 \left(z_1,0,0;\frac{L}{z_1},0,\bar z_3\right).
\end{align}
For this polynomial calculation, the barred and unbarred coordinates are treated as independent variables. A polynomial annihilated by the raising operators is constant under these transformations and therefore can be expressed in terms of $z_1,\bar z_3,L$, initially allowing negative powers of $z_1$ and $\bar z_3$. Since the original expression is a polynomial, these negative powers must cancel: on $z_1=0$, $\bar z_3$ and $L$ remain independent, and on $\bar z_3=0$, $z_1$ and $L$ remain independent. Thus the common polynomial kernel is $\mathbb C[z_1,\bar z_3,L]$.

At degrees $(p,r)$, this kernel has basis
\begin{align}
 &z_1^{p-s}\bar z_3^{r-s}L^s,
 \qquad 0\le s\le\min(p,r), \nonumber \\
 D(z_1^{p-s}\bar z_3^{r-s}L^s)
 &=s(p+r-s+2)z_1^{p-s}\bar z_3^{r-s}L^{s-1}.
\end{align}
The images with $s>0$ are linearly independent and have nonzero coefficients. Hence only the $s=0$ term is harmonic: the highest weight vectors in $\mathcal H_{p,r}$ are multiples of $z_1^p\bar z_3^r$. Every irreducible component has a highest weight vector, so this one-dimensional space implies that $\mathcal H_{p,r}$ is irreducible. Conversely, the angular-momentum ladders for the two adjacent $SU(2)$ subgroups imply that any irreducible representation has nonnegative integer highest weight differences $p,r$. These differences specify the representation: successive lowering operators span it, and their inner products are determined by commuting the adjoint raising operators back to the highest weight vector. Thus every irreducible representation is one of the spaces $\mathcal H_{p,r}$.

We next find the state unchanged by the subgroup acting on coordinates $1,2$. This subgroup can map any vector in $\mathbb C^2$ to any other vector with the same norm. An invariant polynomial therefore depends on those coordinates only through
\begin{align}
 y&=|z_1|^2+|z_2|^2.
\end{align}
At degrees $(p,r)$ it therefore has the form
\begin{align}
 h_3(z)&=\sum_{s=0}^{\min(p,r)}
 A_s z_3^{p-s}\bar z_3^{r-s}y^s.
 \label{eq:fermions-su3-invariant}
\end{align}
The subscript indicates the coordinate left unchanged. Differentiation gives
\begin{align}
 D(y^s)&=s(s+1)y^{s-1}, \nonumber \\
 D\bigl(z_3^{p-s}\bar z_3^{r-s}y^s\bigr)
 &=(p-s)(r-s)z_3^{p-s-1}\bar z_3^{r-s-1}y^s
   +s(s+1)z_3^{p-s}\bar z_3^{r-s}y^{s-1}.
\end{align}
Consequently $Dh_3=0$ requires
\begin{align}
 (s+1)(s+2)A_{s+1}+(p-s)(r-s)A_s&=0,
 \qquad 0\le s<\min(p,r).
 \label{eq:fermions-su3-recurrence}
\end{align}
This determines every coefficient from $A_0$. The invariant subspace is therefore one-dimensional, spanned by $h_3$. Relabeling the coordinates gives $h_1$ and $h_2$ for the other two subgroups.

If $p\ne r$, apply the diagonal phase rotations
\begin{align}
 U&=\diag(\e^{i\theta_1},\e^{i\theta_2},\e^{i\theta_3}),
 \qquad \theta_1+\theta_2+\theta_3=0, \nonumber \\
 U\cdot h_a&=\e^{i(p-r)\theta_a}h_a.
\end{align}
For any two distinct coordinate labels, the phases can be chosen to give distinct eigenvalues. Eigenvectors of a unitary with distinct eigenvalues are orthogonal. Thus $h_1,h_2,h_3$ are mutually orthogonal, and the average of their projectors has nonzero eigenvalue $1/3$.

If $p=r=j$, choose $h_3(e_3)=A_0=1$, where $e_1,e_2,e_3$ are the coordinate unit vectors in $\mathbb C^3$. Equation~\eqref{eq:fermions-su3-recurrence} gives
\begin{align}
 A_j&=(-1)^j\prod_{s=0}^{j-1}
       \frac{(j-s)^2}{(s+1)(s+2)} \nonumber \\
 &=\frac{(-1)^j}{j+1}.
\end{align}
At $e_1$ or $e_2$, only the $s=j$ term of $h_3$ survives. Choose $h_1,h_2$ by relabeling coordinates, so that
\begin{align}
 h_a(e_a)&=1, \nonumber \\
 h_a(e_b)&=\frac{(-1)^j}{j+1}\qquad(a\ne b), \nonumber \\
 \|h_1\|_2&=\|h_2\|_2=\|h_3\|_2.
\end{align}
Here the norm is that of the polynomial Hilbert space defined above. The equal norms follow because coordinate relabeling is unitary.

To obtain the overlaps, apply $P_{1,2}$ to $h_2$. Its range is spanned by $h_3$, so
\begin{align}
 P_{1,2}h_2&=h_3\frac{\ip{h_3}{h_2}}{\|h_3\|_2^2}.
\end{align}
The projector averages the subgroup that fixes $e_3$. Evaluating this average at $e_3$ therefore leaves the value of $h_2$ unchanged. Since $h_3(e_3)=1$ and the norms are equal,
\begin{align}
 \frac{\ip{h_3}{h_2}}{\|h_3\|_2\|h_2\|_2}
 &=(P_{1,2}h_2)(e_3) \nonumber \\
 &=h_2(e_3) \nonumber \\
 &=\frac{(-1)^j}{j+1}.
\end{align}
The same calculation applies to the other pairs. Thus the matrix of inner products of the three normalized states is
\begin{align}
 G_j&=\begin{pmatrix}
 1&\dfrac{(-1)^j}{j+1}&\dfrac{(-1)^j}{j+1} \nonumber \\[1.2ex]
 \dfrac{(-1)^j}{j+1}&1&\dfrac{(-1)^j}{j+1} \nonumber \\[1.2ex]
 \dfrac{(-1)^j}{j+1}&\dfrac{(-1)^j}{j+1}&1
 \end{pmatrix}.
\end{align}
As in Eq.~\eqref{eq:so3-average-factorization}, the nonzero eigenvalues of the average of the three rank-one projectors are those of $G_j/3$. The vector $(1,1,1)^{\mathsf T}$ and its two-dimensional orthogonal complement give the respective eigenvalues
\begin{align}
 &\frac{1+2(-1)^j/(j+1)}3, \nonumber \\
 &\frac{1-(-1)^j/(j+1)}3\qquad\text{(twice)},
\end{align}
with zero eigenvalues discarded. For $j=0$, the representation is one-dimensional and every projector is the identity. For $j\ge1$, the largest eigenvalue is
\begin{align}
 &\begin{cases}
 \dfrac{1+1/(j+1)}3\le\dfrac12,&j\text{ odd}, \nonumber \\[1ex]
 \dfrac{1+2/(j+1)}3\le\dfrac59,&j\text{ even}.
 \end{cases}
\end{align}
In the even case $j\ge2$, with equality at $j=2$. This proves the bound on every irreducible subspace and hence on their orthogonal sum.
\end{proof}

The eigenvalue $5/9$ occurs in the fermionic three-copy representation when $n\ge4$. To see this, occupy orbitals $1,2,3,4$ in copy one, orbitals $1,2$ in copy two, and none in copy three. Every raising operator annihilates this occupation state: moving a particle from copy two to copy one is blocked by an already occupied orbital, and all moves out of copy three act on an empty copy. The adjacent differences of its copy occupations are
\begin{align}
 (4-2,2-0)&=(2,2).
\end{align}
It therefore generates $\mathcal H_{2,2}$, the representation attaining $5/9$ in the preceding calculation.

\begin{proof}[Second proof of Theorem~\ref{thm:fermions-calibration}]
Lemma~\ref{lem:normal-fermion-fixed-spaces} identifies the ordinary pair placements of each test with the subgroup projectors in Lemmas~\ref{lem:so3} and~\ref{lem:fermions-su3}. The pairs $12$, $13$, and $23$ each contain two copies and intersect pairwise in exactly one copy. The original two-copy projectors are swap invariant by Lemmas~\ref{lem:fermions-exact} and \ref{lem:fermions-slater-exact} and the paragraph following the former. Thus the hypotheses of Theorem~\ref{thm:overlap} hold with $m=3$ and $q=2$. Both spectral bounds are strictly below $(m-1)/m=2/3$. Equation~\eqref{eq:overlap-gradient-constant} gives
\begin{align}
 \kappa_{\GF}&=1-\frac32\frac7{12}=\frac18, \nonumber \\
 \kappa_{\Sl}&=1-\frac32\frac59=\frac16.
\end{align}
Using the vector $\ket{\chi_\psi}$ for the respective projector, Eq.~\eqref{eq:overlap-e} therefore gives
\begin{align}
 \|\ket{\chi_\psi}\|_2^2
 &\ge R_{\GF,2}(\psi)\bigl(\tfrac18-R_{\GF,2}(\psi)\bigr)
 \qquad(P=P_{\GF,2}),
 \label{eq:fermion-slope-both}\\
 \|\ket{\chi_\psi}\|_2^2
 &\ge R_{\Sl,2}(\psi)\bigl(\tfrac16-R_{\Sl,2}(\psi)\bigr)
 \qquad(P=P_{\Sl,2}). 
\end{align}
Multiplying by $16$, since $\nabla R(\psi)=-4\ket{\chi_\psi}$, yields
\begin{align}
 \|\nabla R_{\GF,2}(\psi)\|_2^2
 &\ge16R_{\GF,2}(\psi)\bigl(\tfrac18-R_{\GF,2}(\psi)\bigr), \nonumber \\
 \|\nabla R_{\Sl,2}(\psi)\|_2^2
 &\ge16R_{\Sl,2}(\psi)\bigl(\tfrac16-R_{\Sl,2}(\psi)\bigr).
\end{align}
These are the assumptions of Theorem~\ref{thm:restoration} with $q=2$. Lemmas~\ref{lem:fermions-exact} and \ref{lem:fermions-slater-exact} identify the nonempty zero-rejection sets as the respective target families. The two-copy consequence in Eq.~\eqref{eq:two-copy-flow-bound} therefore gives
\begin{align}
 R_{\GF,2}(\psi)&\ge\ell_{1/8}\bigl(d_{\GF}(\psi)^2\bigr), \nonumber \\
 R_{\Sl,2}(\psi)&\ge\ell_{1/6}\bigl(d_{\Sl}(\psi)^2\bigr).
\end{align}
Perfect completeness and Eq.~\eqref{eq:universal-upper} give the upper bounds in Theorem~\ref{thm:fermions-calibration}.

For fixed parity, each Gaussian pair projector commutes with the parity of each copy. For fixed particle number, each Slater pair projector vanishes on unequal-number blocks and preserves equal-number blocks, so it commutes with the number operator in each copy. Consequently the three-copy average preserves the prescribed sector in each copy, and restricting to that space preserves the spectral bounds. The vector $\ket{\chi_\psi}$ also stays in the prescribed one-copy sector, because it is obtained by contracting a projected state in the same sector in each copy. We can therefore apply Theorem~\ref{thm:restoration} within that sector, using its exact target family. This proves the restricted assertions as well.
\end{proof}

For either test, let $\kappa$ be its value above. If $R(\psi(0))<\kappa$, Eq.~\eqref{eq:restoration-integral} with $q=2$ bounds the total length of the evolution by
\begin{align}
 \int_0^\infty
 \left\|\frac{\dd}{\dd t}\ket{\psi(t)}\right\|_2\,\dd t
 &\le\frac12\arcsin\sqrt{\frac{R(\psi(0))}{\kappa}}.
\end{align}

 This establishes optimality of the constants in the gradient inequality. The best rejection bound at every distance may be stronger.

\subsection{Rejection coefficients and sharpness}\label{subsec:fermion-higher-response}

Let $q\ge2$, let $\ket g$ be a closest target, and let $\ket\eta$ be the normalized orthogonal deviation in Eq.~\eqref{eq:normal-frame}. We compute the fermionic coefficients in Table~\ref{table:normal-coefficients} by bounding
\begin{align}
 \left\|(I-P_{\mathcal T,q})
       \bigl(D(g^{\otimes q})[\eta]\bigr)\right\|_2^2.
\end{align}
This is the coefficient of the quadratic term in the rejection expansion in Eq.~\eqref{eq:normal-W}. We then give four-mode examples that attain the first branches of the two-copy distance bounds.

\subsubsection{Higher-copy coefficients}

\smallskip\noindent
\emph{Fermionic Gaussians.}
Apply a Gaussian unitary that sends the closest target to the vacuum $\ket\Omega$, first exchanging parity if necessary. These transformations preserve every pair kernel and hence $P_{\GF,q}$, as in the proof of Lemma~\ref{lem:fermion-closest}. By Lemma~\ref{lem:normal-stationarity}, the deviation has no zero-particle or two-particle component. If the input has no parity promise, odd-particle components are also allowed.

By Lemma~\ref{lem:normal-fermion-fixed-spaces}, $P_{\GF,q}$ averages $SO(q)$ rotations of the copy labels. For a uniformly distributed rotation $R\in SO(q)$, each column is uniform on the unit sphere in $\mathbb R^q$, and all entries have the same distribution. For a $k$-particle component $\ket{\eta_k}$, Eq.~\eqref{eq:higher-copy-first-order-acceptance} and linearity give
\begin{align}
 \frac1q\left\|P_{\GF,q}
        \bigl(D(\Omega^{\otimes q})[\eta_k]\bigr)\right\|_2^2
 &=\frac1q\sum_{i,j=1}^q\mathbb E\bigl[R_{j,i}^{\,k}\bigr]
                 \|\ket{\eta_k}\|_2^2 \nonumber\\
 &=q\,\mathbb E\bigl[R_{1,1}^{\,k}\bigr]\|\ket{\eta_k}\|_2^2.
 \label{eq:normal-real-acceptance}
\end{align}
For unit $\ket{\eta_k}$, this is the acceptance probability of $D(\Omega^{\otimes q})[\eta_k]/\sqrt q$.

Odd moments vanish by symmetry. To calculate the fourth moment, rotate the first two coordinates of the first column by $\pi/4$
\begin{align}
 \mathbb E\bigl[R_{1,1}^4\bigr]
 &=\mathbb E\left[\left(\frac{R_{1,1}+R_{2,1}}{\sqrt2}\right)^4\right] \nonumber\\
 &=\frac12\mathbb E\bigl[R_{1,1}^4\bigr]
   +\frac32\mathbb E\bigl[R_{1,1}^2R_{2,1}^2\bigr], \nonumber\\
 \mathbb E\bigl[R_{1,1}^4\bigr]
 &=3\mathbb E\bigl[R_{1,1}^2R_{2,1}^2\bigr].
\end{align}
The terms with an odd power of either coordinate vanish, and the two fourth moments are equal. The column has norm one, so
\begin{align}
 1&=\mathbb E\left[\left(\sum_{i=1}^qR_{i,1}^2\right)^2\right] \nonumber\\
 &=q\mathbb E\bigl[R_{1,1}^4\bigr]
   +q(q-1)\mathbb E\bigl[R_{1,1}^2R_{2,1}^2\bigr] \nonumber\\
 &=\frac{q(q+2)}3\mathbb E\bigl[R_{1,1}^4\bigr].
\end{align}
Consequently,
\begin{align}
 \mathbb E\bigl[R_{1,1}^4\bigr]&=\frac3{q(q+2)},
 \label{eq:normal-real-moments}\\
 \mathbb E\bigl[R_{1,1}^{2r}\bigr]
 &\le\mathbb E\bigl[R_{1,1}^4\bigr]
 \qquad(r\ge2).
\end{align}
The inequality follows pointwise from $|R_{1,1}|\le1$. Thus odd-particle components have zero acceptance, and allowed even-particle components have acceptance at most $3/(q+2)$.

Decompose $\ket\eta=\sum_k\ket{\eta_k}$ by particle number. Copy rotations preserve total particle number, so the projected components remain orthogonal. Using Eq.~\eqref{eq:tensor-derivative} for the norm of the derivative gives
\begin{align}
 &\left\|(I-P_{\GF,q})
       \bigl(D(\Omega^{\otimes q})[\eta]\bigr)\right\|_2^2 \nonumber\\
 &\quad=q\sum_k\left(1-q\mathbb E\bigl[R_{1,1}^{\,k}\bigr]\right)
                        \|\ket{\eta_k}\|_2^2 \nonumber\\
 &\quad\ge q\left(1-\frac3{q+2}\right)
                  \sum_k\|\ket{\eta_k}\|_2^2 \nonumber\\
 &\quad=\frac{q(q-1)}{q+2}.
\end{align}
The last equality uses $\|\ket\eta\|_2=1$. Equality holds for every normalized four-particle deviation, available when $n\ge4$. The same sphere moments are used for centered bosonic Gaussians in Section~\ref{subsec:boson-higher-response}.

\smallskip\noindent
\emph{Slater determinants.}
Choose orbitals so that the closest target is
\begin{align}
 \ket g&=\ket{e_1\wedge\cdots\wedge e_k}.
\end{align}
In $\ket g^{\otimes q}$, each occupied orbital appears in every copy. An $SU(q)$ transformation of those copy labels multiplies their antisymmetric state by its determinant, which is one. Thus the reference tensor power is unchanged.

First consider a deviation component $\ket{\eta_m}$ with $m$ particles. For a diagonal copy rotation with phases satisfying $\sum_{a=1}^q\theta_a=0$,
\begin{align}
 U_\theta\bigl(\ket g^{\otimes(i-1)}\otimes\ket{\eta_m}
                  \otimes\ket g^{\otimes(q-i)}\bigr)
 &=\e^{i(m-k)\theta_i}
   \ket g^{\otimes(i-1)}\otimes\ket{\eta_m}
                  \otimes\ket g^{\otimes(q-i)}.
\end{align}
The phase is the contribution of $m$ particles in copy $i$ and $k$ particles in every other copy. For $m\ne k$, it can be chosen different from one. Since $P_{\Sl,q}$ retains only states unchanged by every copy rotation, it annihilates each such tensor product. Summing over the insertion position gives
\begin{align}
 P_{\Sl,q}\bigl(D(g^{\otimes q})[\eta_m]\bigr)&=0
 \qquad(m\ne k).
\end{align}

Within the $k$-particle sector, decompose the deviation by the number $r$ of occupied orbitals replaced with unoccupied ones. Lemma~\ref{lem:normal-stationarity} excludes $r=0,1$, so $r\ge2$. Particle creation transforms by $U$, and removing a particle from a filled orbital transforms by $\overline U$. For fixed particle and hole orbital labels, moving all $r$ replacements from copy $i$ to copy $j$ therefore has amplitude
\begin{align}
 (U_{j,i})^r(\overline{U_{j,i}})^r&=|U_{j,i}|^{2r}.
\end{align}
Each replacement contains two fermionic operators, so moving the replacement operators past other copies introduces no extra parity sign. Different sets of particle and hole labels remain orthogonal, because copy rotations preserve the total occupancy of each physical orbital. Averaging over $U\in SU(q)$ therefore gives, for any normalized $\ket\eta$ consisting of $r$ replacements,
\begin{align}
 \frac1q\left\|P_{\Sl,q}\bigl(D(g^{\otimes q})[\eta]\bigr)\right\|_2^2
 &=\frac1q\sum_{i,j=1}^q\mathbb E|U_{j,i}|^{2r} \nonumber\\
 &=q\,\mathbb E|U_{1,1}|^{2r}.
 \label{eq:normal-complex-acceptance}
\end{align}

To evaluate this moment, the first column of a uniformly distributed $U\in SU(q)$ is uniform on the complex unit sphere. Its squared moduli are distributed as independent mean-one exponential variables $X_1,\ldots,X_q$ divided by their sum. Thus
\begin{align}
 \mathbb E|U_{1,1}|^{2r}
 &=\mathbb E\left[\left(\frac{X_1}{\sum_{i=1}^qX_i}\right)^r\right].
\end{align}
The sum of the last $q-1$ variables has the gamma density with shape $q-1$. Changing variables by
\begin{align}
 X_1&=st,\qquad \sum_{i=2}^qX_i=s(1-t),
 \qquad s>0,\quad 0<t<1,
\end{align}
with Jacobian $s$, gives
\begin{align}
 \mathbb E|U_{1,1}|^{2r}
 &=\frac1{(q-2)!}\int_0^\infty\int_0^1
      t^r \e^{-s}s^{q-1}(1-t)^{q-2}\,\dd t\,\dd s \nonumber\\
 &=(q-1)\int_0^1t^r(1-t)^{q-2}\,\dd t \nonumber\\
 &=\frac{r!(q-1)!}{(q+r-1)!}.
 \label{eq:normal-complex-moments}
\end{align}
The second equality integrates over $s$, giving $(q-1)!$, and the last is the beta integral. The ratio of successive moments is
\begin{align}
 \frac{\mathbb E|U_{1,1}|^{2r+2}}{\mathbb E|U_{1,1}|^{2r}}
 &=\frac{r+1}{q+r}\le1.
\end{align}
Hence the largest allowed acceptance occurs at $r=2$
\begin{align}
 q\,\mathbb E|U_{1,1}|^{2r}
 &\le q\,\mathbb E|U_{1,1}|^4 \nonumber\\
 &=\frac2{q+1}.
\end{align}
This bound applies to every allowed deviation by orthogonality of the particle-number and orbital-occupancy components. Since the full derivative has squared norm $q$, Eq.~\eqref{eq:tensor-derivative} gives
\begin{align}
 &\left\|(I-P_{\Sl,q})\bigl(D(g^{\otimes q})[\eta]\bigr)\right\|_2^2 \nonumber\\
 &\quad=q-\left\|P_{\Sl,q}\bigl(D(g^{\otimes q})[\eta]\bigr)\right\|_2^2 \nonumber\\
 &\quad\ge q\left(1-\frac2{q+1}\right) \nonumber\\
 &\quad=\frac{q(q-1)}{q+1}.
\end{align}
Equality holds for every normalized deviation with two occupied-orbital replacements, whenever $k\ge2$ and $n-k\ge2$.

\subsubsection{Four-mode examples}\label{subsec:fermion-distance-curves}

The following superpositions attain the first branches of the two-copy bounds in Theorem~\ref{thm:fermions-calibration}. Adding modes in vacuum preserves their distance and rejection. They therefore also supply the fermionic examples used in Lemma~\ref{lem:testing-exact-curves} for the copy lower bounds.

\begin{proposition}[Sharp four-mode fermionic examples]
\label{prop:fermion-sharpness}
On four modes, choose
\begin{align}
 (\ket\phi,\ket\zeta)&=(\ket{0000},\ket{1111})
 \qquad\text{for }\GF, \nonumber\\
 (\ket\phi,\ket\zeta)&=(\ket{e_1\wedge e_2},\ket{e_3\wedge e_4})
 \qquad\text{for }\Sl.
\end{align}
For $0\le t\le1/\sqrt2$, define
\begin{align}
 \ket{\psi_t}&=\sqrt{1-t^2}\ket\phi+t\ket\zeta.
 \label{eq:fermion-four-mode-state}
\end{align}
Then
\begin{align}
 d_{\mathcal T}(\psi_t)&=t, \nonumber\\
 R_{\GF,2}(\psi_t)&=\tfrac12t^2(1-t^2)
 \qquad\text{for the Gaussian example}, \nonumber\\
 R_{\Sl,2}(\psi_t)&=\tfrac23t^2(1-t^2)
 \qquad\text{for the Slater example}.
 \label{eq:fermion-four-mode-rejections}
\end{align}
The same statements hold after adding modes in vacuum, and after exchanging fermionic parity in the Gaussian example. Thus both first branches are attained throughout $0\le d_{\mathcal T}^2\le1/2$, and the plateau values $1/8$ and $1/6$ are attained at $d_{\mathcal T}^2=1/2$.
\end{proposition}

\begin{proof}
We first show that $\ket\phi$ is a closest target along each curve. For a normalized four-mode even Gaussian, write
\begin{align}
 \ket g&=g_0\ket\phi+\sum_{i<j}g_{i,j}\ket{i,j}+g_4\ket\zeta.
\end{align}
Here $\ket{i,j}$ has precisely modes $i,j$ occupied, in increasing order. The same convention applies to longer lists of occupied modes. By Lemma~\ref{lem:fermions-exact}, $\Lambda\ket g^{\otimes2}=0$. Its coefficient of $\ket1\otimes\ket{2,3,4}$ gives
\begin{align}
 g_0g_4&=g_{1,2}g_{3,4}-g_{1,3}g_{2,4}+g_{1,4}g_{2,3}.
\end{align}
Taking absolute values and applying $2|uv|\le|u|^2+|v|^2$ to each product yields
\begin{align}
 2|g_0g_4|&\le\sum_{i<j}|g_{i,j}|^2, \nonumber\\
 \nonumber
 (|g_0|+|g_4|)^2
 &\le|g_0|^2+\sum_{i<j}|g_{i,j}|^2+|g_4|^2\\
  \nonumber
 &=1.
  \nonumber
\end{align}
Since $t\le\sqrt{1-t^2}$ in the stated interval,
\begin{align}
 |\ip g{\psi_t}|
 &\le\sqrt{1-t^2}|g_0|+t|g_4| \nonumber\\
  \nonumber
 &\le\sqrt{1-t^2}(|g_0|+|g_4|)\\
  \nonumber
 &\le\sqrt{1-t^2}.
  \nonumber
\end{align}
Odd Gaussian targets have zero overlap with this even state, and the vacuum attains the upper bound. Therefore
\begin{align}
 F_{\GF}(\psi_t)&=1-t^2, \nonumber\\
  \nonumber
 d_{\GF}(\psi_t)&=t.
\end{align}
For Slaters, only two-particle targets can have nonzero overlap with the example. Write such a target as $\ket{u\wedge v}$, with $\ket u,\ket v$ orthonormal one-particle orbitals. The antisymmetric coefficient matrix of the example is
\begin{align}
 M&=\begin{pmatrix}0&\sqrt{1-t^2} \nonumber\\-\sqrt{1-t^2}&0\end{pmatrix}
       \mathbin{\oplus}\begin{pmatrix}0&t \nonumber\\-t&0\end{pmatrix}.
\end{align}
Expanding the wedge product from Eq.~\eqref{eq:fermions-wedge-definition} gives
\begin{align}
 |\ip{u\wedge v}{\psi_t}|
 &=|\ip u{M\bar v}| \nonumber\\
 \nonumber
 &\le\|\ket u\|_2\,\|M\|_\infty\,\|\ket{\bar v}\|_2\\
  \nonumber
 &=\|M\|_\infty\\
 &=\sqrt{1-t^2}.
  \nonumber
\end{align}
Here $\bar v$ is the coordinatewise complex conjugate of $v$, and $\|M\|_\infty$ is the operator norm. Its singular values are $\sqrt{1-t^2}$ and $t$, each twice. The target $\ket{e_1\wedge e_2}$ attains the bound, so
\begin{align}
 F_{\Sl}(\psi_t)&=1-t^2, \nonumber\\
 d_{\Sl}(\psi_t)&=t.
\end{align}

We next calculate rejection. Both $\ket\phi$ and $\ket\zeta$ are targets, so their tensor squares are accepted. Expanding the input therefore gives
\begin{align}
 Q\ket{\psi_t}^{\otimes2}
 &=t\sqrt{1-t^2}\,Q(\ket\phi\otimes\ket\zeta
                         +\ket\zeta\otimes\ket\phi), \nonumber\\
 R(\psi_t)
 &=t^2(1-t^2)
   \|Q(\ket\phi\otimes\ket\zeta+\ket\zeta\otimes\ket\phi)\|_2^2,
 \label{eq:fermion-four-mode-cross-term}\\
 \|\ket\phi\otimes\ket\zeta+\ket\zeta\otimes\ket\phi\|_2^2&=2.
\end{align}
The vector
$\ket\phi\otimes\ket\zeta+\ket\zeta\otimes\ket\phi
=D(\phi^{\otimes2})[\zeta]$
is the first-order tensor term considered in the higher-copy
calculations. We apply those formulas with $q=2$.

For the Gaussian example, $\ket\phi$ is vacuum and $\ket\zeta$
has four particles. Equations~\eqref{eq:normal-real-acceptance}
and~\eqref{eq:normal-real-moments} give
\begin{align*}
 \|P_{\GF,2}(\ket\phi\otimes\ket\zeta
                  +\ket\zeta\otimes\ket\phi)\|_2^2
 &=4\cdot\frac{3}{2(2+2)}
 =\frac32.
\end{align*}
Since the unprojected vector has squared norm $2$ and
$Q_{\GF,2}=I-P_{\GF,2}$,
\begin{align*}
 \|Q_{\GF,2}(\ket\phi\otimes\ket\zeta
                  +\ket\zeta\otimes\ket\phi)\|_2^2
 &=2-\frac32=\frac12.
\end{align*}

For the Slater example, $\ket\zeta$ replaces the two occupied
orbitals of $\ket\phi$ with two unoccupied orbitals.
Equations~\eqref{eq:normal-complex-acceptance}
and~\eqref{eq:normal-complex-moments}, with $q=2$ and $r=2$,
give
\begin{align*}
 \|P_{\Sl,2}(\ket\phi\otimes\ket\zeta
                  +\ket\zeta\otimes\ket\phi)\|_2^2
 &=4\cdot\frac{2!}{3!}
 =\frac43.
\end{align*}
Consequently,
\begin{align*}
 \|Q_{\Sl,2}(\ket\phi\otimes\ket\zeta
                  +\ket\zeta\otimes\ket\phi)\|_2^2
 &=2-\frac43=\frac23.
\end{align*}
Substituting these values into
Eq.~\eqref{eq:fermion-four-mode-cross-term} proves the two
rejection formulas in Eq.~\eqref{eq:fermion-four-mode-rejections}.

Finally, consider adding modes in vacuum. Order the $m$ active modes first, and let $V$ be the isometry that appends vacuum in the remaining $n-m$ modes. For the present examples $m=4$. Its adjoint $V^\dagger$ takes the overlap with vacuum in those added modes. From the expression
\begin{align}
 \Lambda_n&=2\sum_{j=1}^n
          (a_j^\dagger\otimes a_j+a_j\otimes a_j^\dagger)
\end{align}
we obtain
\begin{align}
 (V^\dagger\otimes V^\dagger)\Lambda_n
 &=\Lambda_m(V^\dagger\otimes V^\dagger).
 \label{eq:fermion-vacuum-contraction}
\end{align}
Every term for an added mode vanishes after contraction, because one factor creates a particle and hence has zero overlap with vacuum. For a Gaussian target $\ket g$ on all $n$ modes, this gives
\begin{align}
 \Lambda_m(V^\dagger\ket g)^{\otimes2}
 &=(V^\dagger\otimes V^\dagger)\Lambda_n\ket g^{\otimes2} \nonumber\\
 &=0.
\end{align}
By Lemma~\ref{lem:fermions-exact}, $V^\dagger\ket g$ is zero or a multiple of a normalized Gaussian target, with norm at most one. For a Slater target, the same contraction projects each occupied orbital onto the active modes and takes their wedge product. The result is zero or a multiple of a normalized Slater determinant, again with norm at most one. The four-mode overlap bound therefore implies
\begin{align}
 |\ip g{V\psi_t}|
 &=|\ip{V^\dagger g}{\psi_t}| \nonumber\\
 &\le\sqrt{1-t^2}\,\|V^\dagger\ket g\|_2 \nonumber\\
 &\le\sqrt{1-t^2}.
\end{align}
The target $V\ket\phi$ attains this bound, proving that adding vacuum modes preserves the distance.

Rejection is also unchanged. Copy rotations act trivially on the added vacuum modes, so on the active-state subspace their average projector is exactly the original projector. Finally, exchanging Gaussian parity preserves both fidelity and rejection, as proved in Lemma~\ref{lem:fermion-closest}.
\end{proof}

These examples also limit a global bound proportional to squared distance. If $R_{\GF,2}(\psi)\ge c\,d_{\GF}(\psi)^2$ for every input, the example at $t=1/\sqrt2$ gives
\begin{align}
 \frac18&\ge c\,\frac12, \nonumber\\
 c&\le\frac14.
\end{align}
Theorem~\ref{thm:fermions-calibration} and Eq.~\eqref{eq:ell-definition} already give the valid coefficient $c=1/8$. Thus the optimal universal coefficient remains in $[1/8,1/4]$.

\subsection{Measurements}\label{subsec:fermion-measurements}

For the fermionic Gaussianity test, use a Jordan--Wigner encoding,
\begin{align}
 \gamma_{2j-1}&=Z_1\cdots Z_{j-1}X_j, \nonumber\\
 \gamma_{2j}&=Z_1\cdots Z_{j-1}Y_j.
\end{align}
On two copies, the Majorana products commute, since for $\mu\ne\nu$,
\begin{align}
 (\gamma_\mu\otimes\gamma_\mu)(\gamma_\nu\otimes\gamma_\nu)
 &=(-\gamma_\nu\gamma_\mu)\otimes(-\gamma_\nu\gamma_\mu) \nonumber\\
 &=(\gamma_\nu\otimes\gamma_\nu)(\gamma_\mu\otimes\gamma_\mu).
\end{align}
A Bell measurement on the corresponding pair of qubits gives simultaneous eigenvalues $x_j,z_j$ of $X\otimes X,Z\otimes Z$ and the eigenvalue $y_j=-x_jz_j$ of $Y\otimes Y$. The two Majorana-product eigenvalues at orbital $j$ are then
\begin{align}
 &x_j\prod_{k<j}z_k, \nonumber\\
 &y_j\prod_{k<j}z_k.
\end{align}
Their sum over $j$ is the measured eigenvalue of the Gaussianity operator $\Lambda$. Accepting exactly when this sum is zero measures its kernel projector. This uses $n$ Bell measurements and $O(n)$ arithmetic steps.

For Slaters, the measurement is the spectral measurement of the total copy spin $J^2$, accepting spin zero. At each physical orbital, empty and doubly occupied copy configurations are spin-zero vectors. The two singly occupied configurations form a spin-$1/2$ pair. The accepting space is the spin-zero part of their tensor product, with the fermionic ordering convention used in the definition of the copy action. This describes a collective measurement across orbitals.

\newpage
\section{Bosons}\label{sec:bosons}

This appendix applies the framework of Appendix~\ref{sec:preliminaries} to centered pure Gaussian states, coherent states, and pure Gaussian states with arbitrary displacement. We work with finitely many physical modes and impose no energy or photon-number assumption on the input. The tests first recognize Gaussian states through symmetries of several identical copies. Their accepted spaces are linear, whereas the inputs form a nonlinear set of tensor powers. We prove that a tensor power close to an accepted space comes from a state close to the corresponding Gaussian family.

\begin{enumerate}

\item \textbf{Section~\ref{subsec:boson-coordinates}: Express states and copy symmetries in Bargmann coordinates.}
We first introduce the Bargmann representation and the target families. This representation expresses the symmetry conditions as identities between entire functions (that is, functions complex differentiable everywhere). Every normalized Fock state defines such a function, even when its mean photon number is infinite. We can therefore differentiate these identities pointwise without assuming finite energy. Decomposing the Fock space into photon-number sectors allows uniform estimates on these sectors to extend to arbitrary inputs (Eq.~\eqref{eq:boson-sector-inequality}).

\item \textbf{Section~\ref{subsec:boson-base-tests}: Turn copy symmetries into global distance bounds for centered and coherent states.}
Bargmann coordinates identify the exactly accepted inputs of the centered-Gaussian and coherent-state tests. For each two-copy test, three overlapping instances have a spectral gap uniform in photon number and physical-mode count. The overlapping-projection and gradient-flow theorems (Theorems~\ref{thm:overlap} and~\ref{thm:restoration}) convert these gaps into global bounds on distance from the corresponding target family (Theorem~\ref{thm:boson-base}).

\item \textbf{Section~\ref{subsec:boson-frames}: Establish the needed unitary identities.}
We show how Gaussian unitaries relate the target states to vacuum and coherent states, and how these unitaries interact with mixing the copies (Lemma~\ref{lem:boson-frames}). These identities allow us to choose a convenient Gaussian frame when proving global distance bounds for arbitrarily displaced Gaussian states and when establishing sharpness.

\item \textbf{Section~\ref{sec:boson-full}: Include arbitrary displacement.}
We mix two identical inputs and examine the difference output, where their common displacement cancels. This lets us use centered-state testing without knowing the displacement. Extending the centered bound to this possibly mixed output, and combining it with the coherent bound and the Gaussian-unitary identities of Section~\ref{subsec:boson-frames}, gives a four-copy bound for arbitrary pure Gaussian states.We then show that the symmetry of three copies detects departures from Gaussianity strongly enough to inherit the four-copy guarantee (Lemma~\ref{lem:boson-tetrahedron}). Finally, we prove that three copies are necessary for a nontrivial test with perfect completeness (Proposition~\ref{prop:boson-two-copy-impossible}).

\item \textbf{Section~\ref{subsec:boson-rejection-sharpness}: Determine the local rejection coefficients and prove sharpness.}
The global bounds establish robustness; here we determine how sensitively the tests detect small departures from Gaussianity. We prove that closest targets exist (Lemma~\ref{lem:normal-attainment}) and identify the photon-number components allowed in deviations from them. This analysis establishes sharp local rejection coefficients for the tolerant-testing guarantees and provides explicit states used for the paper's copy lower bounds (Lemma~\ref{lem:testing-exact-curves} and Theorem~\ref{thm:testing-lower-bound}).

\item \textbf{Section~\ref{subsec:boson-measurements}: Implement the basic tests optically.}
Finally, we implement the basic tests with passive interferometers and photon detection.

\end{enumerate}

\subsection{Bargmann representation}
\label{subsec:boson-coordinates}

We begin by introducing the Bargmann representation, which describes bosonic state vectors as analytic functions. In this representation, we can write the condition of perfect acceptance by a test as an identity between Bargmann functions. For every Fock-space input, these functions can be differentiated at every point, even when the expected photon number is infinite. For a detailed introduction and
the standard properties recalled below, see
Hall~\cite{hall_holomorphic_2000} and
Chabaud and Mehraban~\cite[Secs.~4.1--4.2]{chabaud2022holomorphic}.
The Bargmann function is called the \emph{stellar function} in the
latter reference.

\begin{definition}[Bargmann function]
Let $\mathcal F_n=\bigoplus_{k\ge0}\Sym^k\mathbb C^n$. 
The \emph{Bargmann function} of an $n$-mode Fock-space vector
$\ket f=\sum_{\alpha\in\mathbb N^n}
f_\alpha\ket{\alpha_1,\ldots,\alpha_n}$ is
\begin{align*}
 f(z)
 &=\sum_{\alpha\in\mathbb N^n}
 f_\alpha\frac{z^\alpha}{\sqrt{\alpha!}},
 \qquad z\in\mathbb C^n,
\end{align*}
where $z^\alpha=\prod_{j=1}^n z_j^{\alpha_j}$ and
$\alpha!=\prod_{j=1}^n\alpha_j!$.
This function is \emph{entire}, meaning holomorphic on all of
$\mathbb C^n$. The map $\ket f\mapsto f$ is called the
\emph{Bargmann representation}.
\end{definition}

We recall four basic properties.
\begin{enumerate}
\item \textbf{Inner products.}
The map is unitary onto the entire functions of finite Gaussian norm:
\begin{align}
 \ip fg
 &=\pi^{-n}\int_{\mathbb C^n}
 \overline{f(z)}\,g(z)\,\e^{-\|z\|_2^2}\,\dd^{2n}z.
 \label{eq:boson-bargmann}
\end{align}
The vacuum $\ket\Omega$ corresponds to the constant function $1$.

\item \textbf{Pointwise evaluation and differentiation.}
For every $z\in\mathbb C^n$,
\begin{align}
 |f(z)|
 &\le \|\ket f\|_2\,\e^{\|z\|_2^2/2}.
 \label{eq:boson-evaluation}
\end{align}
Two Fock-space vectors are equal if and only if their Bargmann
functions agree at every point. Identities between these functions
can be differentiated at each point, even when the input has
infinite mean photon number. The derivative functions need not
have finite Gaussian norm.

\item \textbf{Tensor products.}
Tensor products correspond to products of functions in separate
variables:
\begin{align*}
 \ket f\otimes\ket g
 &\longleftrightarrow f(z)g(w).
\end{align*}
In particular, $q$ identical copies of an $n$-mode state have the following
Bargmann function with one group of $n$ variables for each copy.
\begin{align*}
 \ket f^{\otimes q}
 &\longleftrightarrow
 \prod_{r=1}^{q}f\!\left(z^{(r)}\right),
 \qquad z^{(r)}\in\mathbb C^n,
\end{align*}
\item 
\textbf{Photon number.} Vectors with exactly $k$ photons correspond to homogeneous polynomials of total degree $k$: Bargmann function of this vector is a polynomial in which the exponents in every term add up to $k$, so every monomial has degree $k$.

\end{enumerate}

We will use these coordinates to express perfect acceptance by
the tests as identities between functions. We first record the
Bargmann forms of the target families.

\begin{lemma}[Gaussian coordinates]\label{lem:boson-gaussian-coordinates}
Up to a global phase, the normalized Bargmann functions of the centered, coherent, and arbitrary pure Gaussian families are
\begin{align}
 g_Z(z)&=C_Z\exp\!\left(\tfrac12z^TZz\right),
  \quad&Z=Z^T,\ &\|Z\|_\infty<1,
 \label{eq:boson-centered-form}\\
 c_\alpha(z)&=\exp\!\left(-\tfrac12\|\alpha\|_2^2+\alpha^Tz\right),
 \quad& & \qquad\qquad\quad\alpha\in\mathbb C^n, 
 \label{eq:boson-coherent-form}\\
 g_{Z,b}(z)&=C_{Z,b}\exp\!\left(\tfrac12z^TZz+b^Tz\right),
 \quad &Z=Z^T,\ &\|Z\|_\infty<1,\ b\in\mathbb C^n.
 \label{eq:boson-full-form}
\end{align}
Here $\|Z\|_\infty$ is the operator norm. The constants are nonzero, with $|C_Z|=\det(I-Z^\dagger Z)^{1/4}$. 
\end{lemma}
\begin{proof}
Takagi factorization expresses $Z$ in terms of its singular values $s_j$
\begin{align}
 Z&=V\diag(s_j)V^T,\qquad s_j\ge0,\qquad V^\dagger V=I.
\end{align}
The change $z=\bar V\zeta$ preserves Eq.~\eqref{eq:boson-bargmann}. Write $b'=V^\dagger b$, $b'_j=u_j+iv_j$, and $\zeta_j=x_j+iy_j$, with $u_j,v_j,x_j,y_j$ real. The exponent in the squared norm separates as
\begin{align}
 \Re(\zeta^T\diag(s_j)\zeta+2b'^T\zeta)-\|\zeta\|_2^2
 &=\sum_j\bigl[-(1-s_j)x_j^2-(1+s_j)y_j^2+2u_jx_j-2v_jy_j\bigr].
\end{align}
For $s_j<1$, completing the squares gives
\begin{align}
 |C_{Z,b}|&=\prod_j(1-s_j^2)^{1/4}
 \exp\!\left[-\frac12\sum_j
   \left(\frac{u_j^2}{1-s_j}+\frac{v_j^2}{1+s_j}\right)\right]\\
 &\le\det(I-Z^\dagger Z)^{1/4}\exp(-\|b\|_2^2/4).
 \label{eq:normal-gaussian-normalizer}
\end{align}
The inequality uses $1\pm s_j\le2$ and $\|b'\|_2=\|b\|_2$. Integrability requires every $s_j<1$: otherwise one real direction has zero or negative quadratic decay, and a linear term cannot make both tails integrable. Setting $b=0$ or $Z=0$ gives the other normalizations. 
\end{proof}

We will also use the fact that projecting some modes onto the vacuum preserves each of these Gaussian families after normalization.

\begin{lemma}[Vacuum projection] \label{lem:boson-vacuum-projection} Let $\ket g$ be a pure Gaussian state on disjoint sets of modes $A\sqcup B$. Taking the vacuum overlap on $B$ gives \begin{align*} \ket h:=(I_A\otimes\bra\Omega_B)\ket g \quad\longleftrightarrow\quad h(z_A)=g(z_A,0). \end{align*} Then $\ket h\ne0$, and after normalization is pure Gaussian. If $\ket g$ is centered or coherent, so is the normalized $\ket h$. \end{lemma} \begin{proof} Vacuum overlap selects the zero-occupation coefficients in modes $B$, which amounts to setting $z_B=0$. Writing $g(z)=C\exp(\tfrac12 z^T Zz+b^Tz)$, we obtain \begin{align*} h(z_A) &=C\exp\!\left(\tfrac12 z_A^T Z_{AA}z_A+b_A^Tz_A\right), \end{align*} where $Z_{AA}$ and $b_A$ retain the indices in $A$. Since $C\ne0$, this function is nonzero. Moreover, $Z_{AA}$ is symmetric and $\|Z_{AA}\|_\infty\le\|Z\|_\infty<1$, so $h$ has pure Gaussian form. Finally, $b=0$ implies $b_A=0$, and $Z=0$ implies $Z_{AA}=0$, proving the centered and coherent cases. \end{proof}

For $O\in O(q)$, $\Gamma(I_n\otimes O)$ denotes the passive unitary that applies the same orthogonal mixing of the copy labels to every physical mode. On $k$-photon states it acts as $(I_n\otimes O)^{\otimes k}$ restricted to the symmetric subspace. Orthogonal substitution preserves the 
Gaussian measure in Eq.~\eqref{eq:boson-bargmann} and total 
photon number. Thus
\begin{align}
 \mathcal F_n^{\otimes q}
  &=\bigoplus_{k\ge0}\Sym^k(\mathbb C^n\otimes\mathbb C^q)
\end{align}
decomposes the space 
into finite-dimensional subspaces of fixed photon number, each preserved by these unitaries. For a compact subgroup $G\le O(q)$, average using normalized Haar measure $\dd O$
\begin{align}
 P_G\ket\xi&=\int_G\Gamma(I_n\otimes O)\ket\xi\,\dd O.
 \label{eq:boson-group-average}
\end{align}
The integral is taken in the vector 2-norm. Haar invariance gives
\begin{align}
 P_G^\dagger&=P_G=P_G^2, \nonumber\\
 P_G\ket\xi=\ket\xi
 &\quad\Longleftrightarrow\quad
 \Gamma(I_n\otimes O)\ket\xi=\ket\xi\quad\text{for every }O\in G.
\end{align}
Hence $P_G$ projects onto the states unchanged by every transformation in $G$.

Finally, let bounded self-adjoint operators $A,B$ preserve total photon number, and let $A_k,B_k$ denote their restrictions to the $k$-photon subspace. If $A_k\le B_k$ for every $k$, then, writing $\ket\xi=\sum_{k\ge0}\ket{\xi_k}$ for the photon-number decomposition,
\begin{align}
 \ip\xi{(B-A)\xi}
   &=\sum_{k\ge0}\ip{\xi_k}{(B_k-A_k)\xi_k} \nonumber\\
   &\ge0.
 \label{eq:boson-sector-inequality}
\end{align}
Boundedness makes the sum convergent. The inequalities in each photon-number sector therefore hold on the full Fock space.

\subsection{Centered and coherent tests}
\label{subsec:boson-base-tests}

We now define the two-copy tests. The centered test checks
invariance under rotations of the copies. Tensor squares of centered Gaussian states are unchanged by rotations mixing the two copies. Averaging these rotations gives the test.

Let $U_\theta$ mix the two copies. In the Bargmann representation, for $z,w\in\mathbb C^n$,
\begin{align}
 U_\theta(f\otimes f)(z,w)
    &=f(z\cos\theta+w\sin\theta)f(-z\sin\theta+w\cos\theta).
 \label{eq:boson-copy-rotation}
\end{align}
As in Eq.~\eqref{eq:boson-group-average}, the average projects onto states unchanged by every $U_\theta$. Define the acceptance projector and rejection probability, respectively.
\begin{align}
 P_{\Gzero,2}&=\frac1{2\pi}\int_0^{2\pi}U_\theta\,\dd\theta,
 \label{eq:boson-centered-projector}\\
 R_{\Gzero,2}(\psi)&=1-\|P_{\Gzero,2}\ket\psi^{\otimes2}\|_2^2.
 \label{eq:boson-centered-rejection}
\end{align}
Swapping the inputs conjugates $U_\theta$ to $U_{-\theta}$, so $P_{\Gzero,2}$ is invariant under the ordinary tensor swap. We first show that perfect acceptance characterizes the centered family. This criterion appears in
Girardi et al.~\cite[Theorem~16]{girardi} under a
finite-second-moment assumption. The Bargmann argument below
applies to every normalized Fock-space vector: all
differentiations are pointwise, as justified by the properties of Bargmann function
recalled previously.

\begin{lemma}[Exact centered-Gaussian criterion]
\label{lem:boson-centered-zero}
For a normalized $\ket\psi\in\mathcal F_n$, $R_{\Gzero,2}(\psi)=0$ if and only if $\ket\psi\in\Gzero$.
\end{lemma}

\begin{proof}
For the entire function $f$ representing $\ket\psi$, group averaging makes zero rejection equivalent to
\begin{align}
 f(z\cos\theta+w\sin\theta)f(-z\sin\theta+w\cos\theta)&=f(z)f(w)
 \quad\text{for every }z,w,\theta.
 \label{eq:boson-centered-functional}
\end{align}
Setting $w=z$ and $\theta=\pi/4$ gives
\begin{align}
 f(\sqrt2z)f(0)&=f(z)^2.
\end{align}
Since $f$ is not identically zero, $f(0)\ne0$. Differentiate Eq.~\eqref{eq:boson-centered-functional} with respect to $\theta$ at zero, with $\nabla f=(\partial_1f,\ldots,\partial_nf)$
\begin{align}
 (w\cdot\nabla f(z))f(w)
       &=f(z)(z\cdot\nabla f(w)).
\end{align}
Setting $w=0$ gives $\nabla f(0)=0$. Differentiation in $w_j$ at zero then yields
\begin{align}
 f(0)\partial_jf(z)&=f(z)\sum_i z_i\partial_i\partial_jf(0).
\end{align}
Define the symmetric matrix $Z$ by
\begin{align}
 Z_{i,j}&=\frac{\partial_i\partial_jf(0)}{f(0)},\qquad Z=Z^T.
\end{align}
Consequently $\nabla f(z)=Zzf(z)$, so
\begin{align}
 \nabla\!\left(\e^{-z^TZz/2}f(z)\right)
   &=\e^{-z^TZz/2}\bigl(\nabla f(z)-Zzf(z)\bigr) \nonumber\\
   &=0, \nonumber\\
 f(z)&=f(0)\e^{z^TZz/2}.
\end{align}
Lemma~\ref{lem:boson-gaussian-coordinates} imposes $\|Z\|_\infty<1$. Conversely, the quadratic cross terms cancel in Eq.~\eqref{eq:boson-centered-functional} for every such $Z$.
\end{proof}

For coherent states, a balanced beam splitter sends the difference of two identical inputs to vacuum. In the Bargmann representation it acts as
\begin{align}
 B(f\otimes f)(z,w)
       &=f\!\left(\frac{z+w}{\sqrt2}\right)
        f\!\left(\frac{z-w}{\sqrt2}\right).
 \label{eq:boson-beam-splitter}
\end{align}
The test accepts when the difference output is vacuum
\begin{align}
 P_{\Coh,2}&=B^\dagger(I\otimes\ket\Omega\bra\Omega)B,
 \label{eq:boson-coherent-projector}\\
 R_{\Coh,2}(\psi)&=1-\|P_{\Coh,2}\ket\psi^{\otimes2}\|_2^2.
 \label{eq:boson-coherent-rejection}
\end{align}
The projector is swap invariant: swapping inputs reverses the sign of the difference output, which leaves vacuum fixed.

\begin{lemma}[Exact coherent-state criterion]
\label{lem:boson-coherent-zero}
For a normalized $\ket\psi\in\mathcal F_n$, $R_{\Coh,2}(\psi)=0$ if and only if $\ket\psi\in\Coh$.
\end{lemma}

\begin{proof}
Perfect acceptance says that the right side of Eq.~\eqref{eq:boson-beam-splitter} is independent of $w$. Differentiating in $w_j$ and using the independent variables $u=(z+w)/\sqrt2$, $v=(z-w)/\sqrt2$ gives
\begin{align}
 (\partial_j f)(u)f(v)&=f(u)(\partial_j f)(v).
\end{align}
Choose $v$ with $f(v)\ne0$ and set $\alpha_j=\partial_jf(v)/f(v)$. Then
\begin{align}
 \partial_jf(z)&=\alpha_j f(z), \nonumber\\
 \nabla\!\left(\e^{-\alpha^Tz}f(z)\right)&=0, \nonumber\\
 f(z)&=f(0)\e^{\alpha^Tz}.
\end{align}
Normalization gives Eq.~\eqref{eq:boson-coherent-form} up to a global phase. Conversely,
\begin{align}
 B\ket{c_\alpha}^{\otimes2}
 &=\ket{c_{\sqrt2\alpha}}\otimes\ket\Omega.
\end{align}
\end{proof}

To obtain distance bounds, we compare the three pairwise tests on three auxiliary copies. The following bound will verify the spectral condition in Eq.~\eqref{eq:overlap-spectrum} and hence the gradient bound in Eq.~\eqref{eq:overlap-gradient-bound}.

\begin{lemma}[Overlap bounds for centered and coherent tests]
\label{lem:boson-overlaps}
For either $P=P_{\Gzero,2}$ or $P=P_{\Coh,2}$, place $P$ on each pair of three copies and let $\Pi$ project onto the intersection of their accepted subspaces. Then
\begin{align}
 \frac{P_{1,2}+P_{1,3}+P_{2,3}}3
 &\leq\Pi+\beta(I-\Pi),
 \label{eq:boson-pair-average-bound}\\
 \beta&=\begin{cases}7/12,&P=P_{\Gzero,2}, \nonumber\\1/2,&P=P_{\Coh,2}.
 \end{cases}
\end{align}
\end{lemma}

\begin{proof}
For $P_{\Gzero,2}$, the three pair rotations generate $SO(3)$ on the copy labels. Their axes are orthogonal, so Eq.~\eqref{eq:fermions-so3-bound} in Lemma~\ref{lem:so3} gives $\beta=7/12$ in each fixed-photon-number subspace.

For $P_{\Coh,2}$, we compute the overlap directly. The equally weighted combination $\mathbf u=(1,1,1)/\sqrt3$ from Eq.~\eqref{eq:higher-copy-collective-vector} is unrestricted by all three tests. Among the two orthogonal combinations, the test on pair $(i,j)$ requires its difference output to be vacuum. Photons can then occupy only the combination with coefficient vector $v_{i,j}$, where
\begin{align}
 v_{1,2}&=\frac{(1,1,-2)}{\sqrt6}, \nonumber\\
 v_{1,3}&=\frac{(1,-2,1)}{\sqrt6}, \nonumber\\
 v_{2,3}&=\frac{(-2,1,1)}{\sqrt6}.
\end{align}
For pair labels $a,b\in\{(1,2),(1,3),(2,3)\}$, $\ip{v_a}{v_b}=-1/2$ when $a\ne b$. At total photon number $k$ in these two combinations, let
\begin{align}
 J_a:\Sym^k\mathbb C^n
       &\longrightarrow\Sym^k(\mathbb C^n\otimes\mathbf u^\perp)
\end{align}
assign the copy coefficients $v_a$ to every photon, leaving its physical-mode state unchanged. Each $J_a$ is an isometry, and each photon contributes one factor to the overlap
\begin{align}
 J_a^\dagger J_a&=I, \nonumber\\
 J_a^\dagger J_b&=\ip{v_a}{v_b}^k I \nonumber\\
              &=(-1/2)^kI\qquad(a\ne b).
\end{align}
Thus $J_aJ_a^\dagger$ is the pair acceptance projector on this subspace. For the block operator $L=(J_{1,2},J_{1,3},J_{2,3})/\sqrt3$,
\begin{align}
 LL^\dagger&=\frac13\sum_aJ_aJ_a^\dagger, \nonumber\\
 L^\dagger L&=\frac13
 \begin{pmatrix}
 1&(-1/2)^k&(-1/2)^k \nonumber\\
 (-1/2)^k&1&(-1/2)^k \nonumber\\
 (-1/2)^k&(-1/2)^k&1
 \end{pmatrix}\otimes I.
\end{align}
The nonzero eigenvalues of $LL^\dagger$ and $L^\dagger L$ coincide. The displayed $3\times3$ matrix, including the factor $1/3$, has eigenvalues
\begin{align}
 \frac{1+2(-1/2)^k}{3},\qquad
 \frac{1-(-1/2)^k}{3}\quad\text{(multiplicity two)}.
\end{align}
The identity factor repeats them for each physical-mode state. For $k=0$ the two combinations are in vacuum, giving the common accepted subspace. For $k=1,2$ the largest eigenvalue is $1/2$. For $k\ge3$, $|(-1/2)^k|\le1/8$ gives a strictly smaller value. Thus $\beta=1/2$. Equation~\eqref{eq:boson-sector-inequality} extends these uniform photon-number bounds to the full Fock space.
\end{proof}

Combining Eq.~\eqref{eq:boson-pair-average-bound} with Theorems~\ref{thm:overlap} and \ref{thm:restoration} gives the distance bounds. This theorem turns exact acceptance into a quantitative distance guarantee.

\begin{theorem}[Centered and coherent distance bounds]
\label{thm:boson-base}
Every normalized state $\ket\psi\in\mathcal F_n$ satisfies
\begin{align}
 \ell_{1/8}\bigl(d_{\Gzero}(\psi)^2\bigr)
       &\leq R_{\Gzero,2}(\psi)
       \leq\mathcal U_2\bigl(d_{\Gzero}(\psi)^2\bigr),
       \label{eq:boson-centered-calibration}\\
 \ell_{1/4}\bigl(d_{\Coh}(\psi)^2\bigr)
       &\leq R_{\Coh,2}(\psi)
       \leq\mathcal U_2\bigl(d_{\Coh}(\psi)^2\bigr).
       \label{eq:boson-coherent-calibration}
\end{align}
Here $\ell_\kappa$ is defined in Eq.~\eqref{eq:ell-definition}, and $\mathcal U_2$ is the two-copy upper bound in Eq.~\eqref{eq:universal-upper}. The one-copy vacuum test has the exact rejection probability
\begin{align}
 R_{\Omega}(\psi)&=d_{\{\Omega\}}(\psi)^2.
\end{align}
The constants do not depend on $n$.  There is no photon-number, energy, parity, or first-moment promise on the input.
\end{theorem}

\begin{proof}
Both projectors commute with the tensor swap. Lemmas~\ref{lem:boson-centered-zero} and \ref{lem:boson-coherent-zero} identify their zero-rejection states as the corresponding target families, each containing vacuum. Lemma~\ref{lem:boson-overlaps} verifies Eq.~\eqref{eq:overlap-spectrum} for the three pairs of copies. Theorem~\ref{thm:overlap}, with $q=2$ and $m=3$, therefore gives
\begin{align}
 \kappa&=1-\tfrac32\beta
 =\begin{cases}1/8,&P=P_{\Gzero,2}, \nonumber\\1/4,&P=P_{\Coh,2},
 \end{cases} \nonumber\\
 \|\nabla R(\psi)\|_2^2
 &\ge16R(\psi)\bigl(\kappa-R(\psi)\bigr).
 \label{eq:boson-base-gradient-bound}
\end{align}
This is the assumption of Theorem~\ref{thm:restoration}. Its two-copy conclusion, Eq.~\eqref{eq:two-copy-flow-bound}, gives both lower bounds. Perfect completeness and Eq.~\eqref{eq:universal-upper} give the upper bounds. Finally, Eq.~\eqref{eq:distances} gives
\begin{align}
 R_{\Omega}(\psi)&=1-|\ip\Omega\psi|^2 \nonumber\\
 &=d_{\{\Omega\}}(\psi)^2.
\end{align}
\end{proof}

\subsection{Gaussian unitaries}
\label{subsec:boson-frames}

A Gaussian unitary without displacement can send any centered Gaussian state to vacuum. Applying the same such unitary to every copy commutes with real orthogonal mixing of the copies.

\begin{lemma}[Gaussian unitaries without displacement]
\label{lem:boson-frames}
The Gaussian unitaries with no displacement, called homogeneous
Gaussian unitaries, form a group with these properties.
\begin{enumerate}[label=(\roman*)]
\item Every centered Gaussian has the form $U\ket\Omega$.

\item Applying the same $U$ to every copy commutes with every
real orthogonal mixing of the copies.

\item For $\alpha\in\mathbb C^n$, the displacement operators
\begin{align}
 (D(\alpha)f)(z)
 &=\e^{-\|\alpha\|_2^2/2+\alpha^Tz}f(z-\bar\alpha)
 \label{eq:boson-displacement}
\end{align}
are unitary, and $UD(\alpha)U^\dagger$ is another displacement,
up to a scalar phase.

\item The full Gaussian family satisfies
\begin{align}
 \GB&=\{D(\alpha)U\ket\Omega\} \nonumber\\
    &=\{U^\dagger\ket{c_\alpha}\},
\end{align}
where $\ket{c_\alpha}=D(\alpha)\ket\Omega$, $U$ ranges over this
group, and $\alpha\in\mathbb C^n$.
\end{enumerate}
These unitaries preserve the centered family. Any two centered
Gaussian states have nonzero inner product.
\end{lemma}

\begin{proof}
A passive transformation with unitary mode matrix $V$ acts as
\begin{align}
 (U_Vf)(z)&=f(V^Tz).
 \label{eq:boson-passive-action}
\end{align}
We construct squeezing using the single-mode Bargmann transform
\begin{align}
 (\mathcal Bh)(z)
 &=\pi^{-1/4}\int_{\mathbb R}
   \exp\!\left(-\frac{x^2}{2}+\sqrt2xz-\frac{z^2}{2}\right)
   h(x)\,\dd x.
 \label{eq:boson-bargmann-transform}
\end{align}
With the Gaussian norm in Eq.~\eqref{eq:boson-bargmann}, this
transform is unitary from $L^2(\mathbb R)$ onto the single-mode
Bargmann space~\cite[Theorem~6.2 and Eq.~(6.1), with
$\hbar=1$]{hall_holomorphic_2000}. It sends the normalized vacuum
wavefunction
\begin{align*}
 h_0(x)&=\pi^{-1/4}\e^{-x^2/2}
\end{align*}
to the constant function $1$. Tensor products give the
corresponding unitary identification for any finite number
of modes.

The unitary rescaling
\begin{align*}
 (S_rh)(x)&=\e^{-r/2}h(\e^{-r}x),
 \qquad r\in\mathbb R,
\end{align*}
implements single-mode squeezing. Gaussian integration gives
its action on vacuum:
\begin{align}
 (\mathcal B S_rh_0)(z)
 &=(\cosh r)^{-1/2}
   \exp\!\left(\tfrac12\tanh(r)z^2\right).
 \label{eq:boson-squeezed-vacuum}
\end{align}
Every centered pure Gaussian can be prepared by squeezing the
vacuum independently in each mode and then applying a passive
transformation, up to a global phase. More generally, passive
transformations, single-mode squeezers, and scalar phases generate
all homogeneous Gaussian unitaries; see
\cite[Sec.~2.3, Eqs.~(32)--(34)]{chabaud2022holomorphic}.
Squeezing phases can be included in the passive transformations.
This proves (i).

We next check commutation with copy mixing. For a passive
transformation, the single-photon matrices satisfy
\begin{align*}
 (V\otimes I_q)(I_n\otimes O)
 &=(I_n\otimes O)(V\otimes I_q).
\end{align*}
Real orthogonal copy mixing also acts by the same change of
variables in position coordinates. Indeed, a change of variables
in the product Bargmann integral gives, for $M=I_n\otimes O$,
\begin{align*}
 \mathcal B^{\otimes nq}(h\circ M^T)
 &=(\mathcal B^{\otimes nq}h)\circ M^T.
\end{align*}
Identical squeezers therefore commute with copy mixing as well:
their position rescaling matrix
$\diag(\e^{-r_1},\ldots,\e^{-r_n})\otimes I_q$
commutes with $I_n\otimes O$. Products of these operations satisfy
\begin{align}
 [U^{\otimes q},\Gamma(I_n\otimes O)]&=0,
 \label{eq:boson-unitary-copy-commutation}
\end{align}
which proves (ii). Averaging the two-copy rotations gives
\begin{align*}
 [U^{\otimes2},P_{\Gzero,2}]&=0.
\end{align*}
Lemma~\ref{lem:boson-centered-zero}, applied to $U$ and
$U^\dagger$, therefore gives $U\Gzero=\Gzero$.

For displacement, Eq.~\eqref{eq:boson-displacement} gives
\begin{align*}
 |(D(\alpha)f)(z)|^2\e^{-\|z\|_2^2}
 &=|f(z-\bar\alpha)|^2\e^{-\|z-\bar\alpha\|_2^2},\\
 D(\alpha)^{-1}&=D(-\alpha).
\end{align*}
Integration in Eq.~\eqref{eq:boson-bargmann} shows that
$D(\alpha)$ preserves the norm. The inverse identity then proves
unitarity.

To check how squeezing transforms displacements, define
translations and multiplication by a phase on $L^2(\mathbb R)$:
\begin{align*}
 (T_a h)(x)&=h(x-a),\\
 (M_b h)(x)&=\e^{ibx}h(x),
 \qquad a,b\in\mathbb R.
\end{align*}
Equations~\eqref{eq:boson-bargmann-transform}
and~\eqref{eq:boson-displacement} give
\begin{align}
 \mathcal B^\dagger D(\alpha)\mathcal B
 &=\e^{-i(\Re\alpha)(\Im\alpha)}
   M_{\sqrt2\Im\alpha}T_{\sqrt2\Re\alpha},
 \label{eq:boson-position-displacement}\\
 S_rM_bT_aS_r^\dagger
 &=M_{\e^{-r}b}T_{\e^r a}.
 \label{eq:boson-squeezing-displacement}
\end{align}
Multimode displacement is the product of these single-mode
operators. Direct substitution in
Eq.~\eqref{eq:boson-passive-action} also gives
\begin{align}
 U_VD(\alpha)U_V^\dagger&=D(V\alpha).
 \label{eq:boson-passive-displacement}
\end{align}
These identities prove (iii).

Finally, write $g_Z=g_{Z,0}$ in
Eq.~\eqref{eq:boson-full-form}. Substitution into
Eq.~\eqref{eq:boson-displacement} gives
\begin{align*}
 (D(\alpha)g_Z)(z)
 &=(D(\alpha)g_Z)(0)
   \exp\!\left(
   \tfrac12z^TZz+(\alpha-Z\bar\alpha)^Tz
   \right).
\end{align*}
For every $b$, the map $\alpha\mapsto b+Z\bar\alpha$ has
Lipschitz constant $\|Z\|_\infty<1$ and hence a unique fixed
point. Thus $b=\alpha-Z\bar\alpha$ has a unique solution.
Every full pure Gaussian is consequently a displacement of a
centered pure Gaussian. Combining this with (i) gives the first
description in (iv), and moving the displacement through $U$
using (iii) gives the second.

For centered states $\ket g,\ket h$, choose $U$ such that
$U\ket g=\ket\Omega$. Since $U\ket h$ is centered,
Eq.~\eqref{eq:boson-full-form} with $b=0$ shows that its
constant Bargmann coefficient is nonzero. Hence
\begin{align*}
 \ip gh&=\ip\Omega{Uh}\ne0.
\end{align*}
\end{proof}

\subsection{Displacement invariance}
\label{sec:boson-full}
This section extends the centered and coherent results to all pure
Gaussian states, with unknown displacement. It first establishes a
four-copy test, transfers its guarantee to three copies, and then
proves that two copies cannot suffice for a nontrivial test with
perfect completeness.

Mixing two identical inputs cancels their common displacement in
the difference output, by the displacement transformation rule in
Eq.~\eqref{eq:boson-passive-displacement}. However, this output may
be mixed, so the pure-state centered bound in
Eq.~\eqref{eq:boson-centered-calibration} cannot be applied directly.

We resolve this by keeping both outputs and analyzing their joint
pure state, with the sum output left unrestricted. The
centered-test overlap estimate remains valid on this enlarged
space. The overlapping-projection and gradient-flow theorems
(Theorems~\ref{thm:overlap} and~\ref{thm:restoration}) then show
that small rejection forces the joint state to be close to a
product whose difference factor is pure centered Gaussian
(Lemma~\ref{lem:boson-lift-zero},
Eq.~\eqref{eq:boson-lift-calibration}).

Finally, the Gaussian-unitary identities in
Lemma~\ref{lem:boson-quotients} connect this conclusion to the
coherent-state bound in Eq.~\eqref{eq:boson-coherent-calibration}.
Together, these give the four-copy bound on the original input's
distance from the full Gaussian family
(Eq.~\eqref{eq:boson-four-calibration}).

\subsubsection{The sum and difference outputs}
\label{subsec:boson-lift}

We first mix two identical inputs using the balanced beam splitter
$B$. For each physical mode $\mu$, it forms the normalized sum and
difference of the two input modes,
\begin{align*}
 \frac{a_{\mu,1}+a_{\mu,2}}{\sqrt2},
 \qquad
 \frac{a_{\mu,1}-a_{\mu,2}}{\sqrt2}.
\end{align*}
These give two groups of $n$ output modes, which we call the
\emph{sum output} and the \emph{difference output}, labeled $+$
and $-$.

For a normalized input $\ket\psi$, write
\begin{align}
 \ket{\Xi_\psi}
 &=B\ket\psi^{\otimes2}
   \in\mathcal F_{n,+}\otimes\mathcal F_{n,-},
 \label{eq:boson-lift}\\
 \rho_-(\psi)
 &=\Tr_+\ket{\Xi_\psi}\bra{\Xi_\psi}.
 \label{eq:boson-difference-state}
\end{align}
The joint output $\ket{\Xi_\psi}$ is pure, while the state
$\rho_-(\psi)$ of the difference output alone may be mixed.

We repeat this construction on a second pair of inputs and apply
the centered test to the two difference outputs, leaving the sum
outputs unrestricted. This uses four copies of the original input.
After grouping the sum factors first, the acceptance projector on
the output space is
\begin{align*}
 \widetilde P
 &=I_+^{\otimes2}\otimes P_{\Gzero,2}.
\end{align*}
Its rejection probability 
is
\begin{align}
 R_{\mathrm{diff},4}(\psi)
 &=1-\|\widetilde P\ket{\Xi_\psi}^{\otimes2}\|_2^2
 \nonumber\\
 &=1-\Tr\!\left(
 P_{\Gzero,2}\rho_-(\psi)^{\otimes2}
 \right).
 \label{eq:boson-four-rejection}
\end{align}
Conjugating this projector back to the input space through the
beam splitters and regrouping defines the four-copy acceptance
projector $P_{\mathrm{diff},4}$.

The test accepts perfectly whenever the joint output is a product
of an arbitrary sum state and a centered pure Gaussian difference
state. These states form the family
\begin{align}
 \mathcal Z
 &=\left\{
 \ket\chi\otimes\ket g:
 \|\ket\chi\|_2=1,\quad \ket g\in\Gzero
 \right\}.
 \label{eq:boson-product-target}
\end{align}
The next lemma shows that these are the only perfectly accepted
joint pure states. Thus, although the test acts only on the
difference outputs, perfect acceptance also forces them to be
unentangled with the sum outputs.

\begin{lemma}[Distance bound for the joint output state]
\label{lem:boson-lift-zero}
For a normalized joint state $\ket\Xi$, $\widetilde P\ket\Xi^{\otimes2}=\ket\Xi^{\otimes2}$ if and only if $\ket\Xi\in\mathcal Z$. Moreover,
\begin{align}
 1-\|\widetilde P\ket\Xi^{\otimes2}\|_2^2
       &\ge\ell_{1/8}\bigl(d_{\mathcal Z}(\Xi)^2\bigr).
 \label{eq:boson-lift-calibration}
\end{align}
\end{lemma}

\begin{proof}
Write a Schmidt decomposition
\begin{align}
 \ket\Xi&=\sum_j\sqrt{p_j}\,\ket{\chi_j}\otimes\ket{g_j},
 \qquad p_j>0,
\end{align}
where the two families are orthonormal. Orthogonality of the sum-output states gives the convergent nonnegative sum
\begin{align}
 \|(I-\widetilde P)\ket\Xi^{\otimes2}\|_2^2
   &=\sum_{i,j}p_ip_j
       \|(I-P_{\Gzero,2})(\ket{g_i}\otimes\ket{g_j})\|_2^2.
\end{align}
If it vanishes, every diagonal term vanishes. Each $\ket{g_j}$ is therefore centered Gaussian by Lemma~\ref{lem:boson-centered-zero}. Lemma~\ref{lem:boson-frames} shows that two such states cannot be orthogonal, so there is only one Schmidt term. The converse follows from perfect completeness of the centered test.

Swap invariance of $P_{\Gzero,2}$ makes $\widetilde P$ swap invariant on the joint copies. On three difference outputs, let
\begin{align}
 T_-&=\frac{P_{1,2}+P_{1,3}+P_{2,3}}3,
\end{align}
where $P_{i,j}$ applies $P_{\Gzero,2}$ to outputs $i,j$, and let $\Pi_-$ project onto their common accepted subspace. After regrouping three joint copies, Eq.~\eqref{eq:boson-pair-average-bound} gives
\begin{align}
 \widetilde T&=I_+^{\otimes3}\otimes T_-, \nonumber\\
 \widetilde\Pi&=I_+^{\otimes3}\otimes\Pi_-, \nonumber\\
 \widetilde T&\le\widetilde\Pi+\frac7{12}(I-\widetilde\Pi).
\end{align}
Theorem~\ref{thm:overlap} applies to these three pair projectors with
\begin{align}
 \kappa&=1-\frac32\frac7{12}=\frac18.
\end{align}
The nonempty zero-rejection set is $\mathcal Z$, as proved above. The two-copy conclusion of Theorem~\ref{thm:restoration}, Eq.~\eqref{eq:two-copy-flow-bound}, proves Eq.~\eqref{eq:boson-lift-calibration}.
\end{proof}

The evolution in Theorem~\ref{thm:restoration} acts on joint pure states and need not preserve the form $B\ket\psi^{\otimes2}$. The following identities relate distance in the joint output space to distance from the original Gaussian family.

\begin{lemma}[Distances under Gaussian unitaries]
\label{lem:boson-quotients}
Let $U$ range over all Gaussian unitaries without displacement. Then
\begin{align}
 d_{\mathcal Z}(\Xi_\psi)^2
     &=\inf_U R_{\Coh,2}(U\psi),
 \label{eq:boson-quotient-identities}\\
 d_{\GB}(\psi)^2
     &=\inf_U d_{\Coh}(U\psi)^2.
 \label{eq:boson-gaussian-coherent-distance}
\end{align}
\end{lemma}

\begin{proof}
For a fixed normalized $\ket g\in\Gzero$, maximizing over the sum-output state gives
\begin{align}
 \sup_{\|\ket\chi\|_2=1}|\ip{\chi\otimes g}{\Xi_\psi}|^2
   &=\|(I\otimes\bra g)\ket{\Xi_\psi}\|_2^2 \nonumber\\
   &=\ip g{\rho_-(\psi)g}.
\end{align}
Write $\ket g=U^\dagger\ket\Omega$ using Lemma~\ref{lem:boson-frames}. Equation~\eqref{eq:boson-unitary-copy-commutation} implies that $U\otimes U$ commutes with $B$, so
\begin{align}
 \rho_-(U\psi)&=U\rho_-(\psi)U^\dagger, \nonumber\\
 \ip g{\rho_-(\psi)g}
   &=\ip\Omega{\rho_-(U\psi)\Omega} \nonumber\\
   &=1-R_{\Coh,2}(U\psi).
\end{align}
The last equality uses Eq.~\eqref{eq:boson-coherent-projector}. Taking the supremum over centered $\ket g$ proves Eq.~\eqref{eq:boson-quotient-identities}. The full-family description in Lemma~\ref{lem:boson-frames}(iv) gives
\begin{align}
 d_{\GB}(\psi)^2
   &=1-\sup_{U,\alpha}|\ip{U^\dagger c_\alpha}\psi|^2 \nonumber\\
   &=\inf_U\left(1-\sup_\alpha|\ip{c_\alpha}{U\psi}|^2\right) \nonumber\\
   &=\inf_U d_{\Coh}(U\psi)^2.
\end{align}
Neither identity requires an attained optimum.
\end{proof}

The coherent and centered lower bounds compose to the function from Eq.~\eqref{eq:ellB-definition}
\begin{align}
 \ell_{\mathrm B}\bigl(d_{\GB}(\psi)^2\bigr)
   &=\ell_{1/8}\!\left(\ell_{1/4}\bigl(d_{\GB}(\psi)^2\bigr)\right).
 \label{eq:boson-full-envelope}
\end{align}

\subsubsection{Three-copy comparison}
\label{subsec:boson-tetrahedron}

We now show that the three-copy test inherits the distance guarantee
of the four-copy difference-output test. Following the higher-copy
construction in Section~\ref{sec:higher-copy-tests}
(Definition~\ref{def:normal-projectors}), we average copy rotations
that leave the equally weighted combination unchanged.

For three copies, the coefficient vector
\begin{align*}
 \mathbf u&=\frac{(1,1,1)}{\sqrt3}
\end{align*}
describes the normalized sum of corresponding modes across the
copies (Eq.~\eqref{eq:higher-copy-collective-vector}). Equal
displacements affect only this combination. We therefore leave
it fixed and rotate the two orthogonal combinations.

Let $r_\theta\in SO(3)$ fix $\mathbf u$ and rotate its orthogonal
complement by angle $\theta$. The acceptance projector and
rejection probability are
\begin{align}
 P_{\GB,3}
 &=\frac1{2\pi}\int_0^{2\pi}
   \Gamma(I_n\otimes r_\theta)\,\dd\theta,
 \label{eq:boson-three-projector}\\
 R_{\GB,3}(\psi)
 &=1-\|P_{\GB,3}\ket\psi^{\otimes3}\|_2^2.
\end{align}
This projector accepts $\ket g^{\otimes3}$ for every pure Gaussian
state $\ket g$: vacuum is invariant, identical Gaussian unitaries
without displacement commute with copy mixing by
Eq.~\eqref{eq:boson-unitary-copy-commutation}, and equal
displacements affect only the fixed combination $\mathbf u$.

Every copy permutation also fixes $\mathbf u$ and conjugates the
rotations $r_\theta$ to the same set. Thus the acceptance projector
commutes with every copy permutation:
\begin{align*}
 [P_{\GB,3},U_\pi]&=0.
\end{align*}

To compare the three-copy and four-copy tests as operators, we
place both on the four-copy space. Here
\begin{align*}
 \mathbf u&=\frac{(1,1,1,1)}2,
\end{align*}
and the three-dimensional space $\mathbf u^\perp$ contains the
combinations orthogonal to the common sum. Let $P_{\rm all}$ be
the projector obtained by averaging all rotations of this space,
which form an $SO(3)$ group.

For each $i=1,2,3,4$, let $P_{\widehat i}$ apply the three-copy
test to the triple omitting copy $i$, leaving that copy
unrestricted. Its rotations fix the omitted copy direction $e_i$
as well as $\mathbf u$. Their axis within $\mathbf u^\perp$ is
therefore the normalized projection of $e_i$ onto this space:
\begin{align*}
 v_i&=\frac2{\sqrt3}\left(e_i-\frac{\mathbf u}{2}\right).
\end{align*}
These are unit vectors satisfying
\begin{align*}
 \ip{v_i}{v_j}&=-\frac13
 \qquad(i\ne j),
\end{align*}
so they point toward the vertices of a regular tetrahedron. We compare the average of the four triple-test projectors,
\begin{align*}
 T_{\rm tet}&=\frac14\sum_{i=1}^4P_{\widehat i},
\end{align*}
with the four-copy difference-output projector. On four identical
inputs, every term in this average has the acceptance probability
of the three-copy test. An operator bound for $T_{\rm tet}$ will
therefore give the desired comparison of rejection probabilities.

\begin{lemma}[Tetrahedral comparison]
\label{lem:boson-tetrahedron}
On four-copy Fock space,
\begin{align}
 I-T_{\rm tet}&\ge\frac49(I-P_{\rm all}) \nonumber\\
             &\ge\frac49(I-P_{\mathrm{diff},4}).
 \label{eq:boson-tetrahedral-operator}
\end{align}
Consequently,
\begin{align}
 R_{\GB,3}(\psi)&\ge\frac49R_{\mathrm{diff},4}(\psi).
 \label{eq:boson-three-four-comparison}
\end{align}
\end{lemma}

\begin{proof}
In a spin-$\ell$ representation, Eq.~\eqref{eq:so3-axis-projector} expresses each axis projector as a rank-one projector. The normalized states have pairwise inner products
\begin{align}
 p_\ell&=P_\ell(-1/3)
\end{align}
by Eq.~\eqref{eq:fermions-zonal-overlap}. Their Gram matrix is
\begin{align}
 \begin{pmatrix}
 1&p_\ell&p_\ell&p_\ell \nonumber\\
 p_\ell&1&p_\ell&p_\ell \nonumber\\
 p_\ell&p_\ell&1&p_\ell \nonumber\\
 p_\ell&p_\ell&p_\ell&1
 \end{pmatrix}.
\end{align}
The nonzero eigenvalues of the average of the four projectors coincide with those of one quarter of this matrix. The eigenvalues of the scaled matrix are
\begin{align}
 \frac{1+3p_\ell}{4},\qquad
 \frac{1-p_\ell}{4}\quad\text{(multiplicity three)}.
 \label{eq:boson-tetrahedral-spectrum}
\end{align}
Positivity of the Gram matrix gives $p_\ell\ge-1/3$. For every $\ell\ge1$, we also have $p_\ell\le11/27$, as follows. The Legendre recurrence
\begin{align}
 (\ell+1)P_{\ell+1}(t)
     &=(2\ell+1)tP_\ell(t)-\ell P_{\ell-1}(t), \nonumber\\
 P_0(t)&=1, \nonumber\\
 P_1(t)&=t
\end{align}
gives
\begin{align}
 (p_1,p_2,p_3,p_4,p_5)
   &=\left(-\frac13,-\frac13,\frac{11}{27},\frac1{81},-\frac13\right), \nonumber\\
 (p_6,p_7,p_8,p_9,p_{10})
   &=\left(\frac{47}{243},\frac{121}{729},-\frac{199}{729},
            \frac{479}{19683},\frac{13597}{59049}\right).
\end{align}
These values are at most $11/27$. For $\ell\ge11$, the integral formula in Eq.~\eqref{eq:fermions-legendre-integral} gives
\begin{align}
 |P_\ell(-1/3)|
 &\le\frac1\pi\int_0^\pi
          \bigl[1-\tfrac89\sin^2t\bigr]^{\ell/2}\,\dd t \nonumber\\
 &\le\frac2\pi\int_0^{\pi/2}
          \exp\bigl(-\tfrac{4\ell}{9}\sin^2t\bigr)\,\dd t \nonumber\\
 &\le\frac2\pi\int_0^\infty
          \exp\bigl(-\tfrac{16\ell}{9\pi^2}t^2\bigr)\,\dd t \nonumber\\
 &=\frac{3\sqrt\pi}{4\sqrt\ell} \nonumber\\
 &<\frac{11}{27}.
\end{align}
In this chain of steps,
we have used 
the inequalities $1-y\le\e^{-y}$ and $\sin t\ge2t/\pi$ on $[0,\pi/2]$. Together with $p_\ell\ge-1/3$, this gives
\begin{align}
 \max\left\{\frac{1+3p_\ell}{4},\frac{1-p_\ell}{4}\right\}
 &\le\max\left\{\frac59,\frac13\right\}=\frac59
 \qquad(\ell\ge1).
\end{align}
The spin-zero subspaces are exactly the states accepted by $P_{\rm all}$. The bound is uniform in spin and photon number, so Eq.~\eqref{eq:boson-sector-inequality} proves the first inequality in Eq.~\eqref{eq:boson-tetrahedral-operator} on the full Fock space.
The four-copy difference-output test averages rotations in the plane spanned by
\begin{align}
 d_1&=(1,-1,0,0)/\sqrt2, \nonumber\\
 d_2&=(0,0,1,-1)/\sqrt2.
 \label{eq:boson-difference-plane}
\end{align}
Both vectors are orthogonal to $\mathbf u$. The rotation group defining $P_{\rm all}$ therefore includes these rotations, so every state it accepts is also accepted by $P_{\mathrm{diff},4}$
\begin{align}
 P_{\rm all}&\le P_{\mathrm{diff},4}.
\end{align}
This proves the second inequality. Finally, each $P_{\widehat i}$ acts as the identity on its omitted copy, so
\begin{align}
 \ip{\psi^{\otimes4}}{(I-T_{\rm tet})\psi^{\otimes4}}
   &=\frac14\sum_{i=1}^4
      \ip{\psi^{\otimes4}}{(I-P_{\widehat i})\psi^{\otimes4}} \nonumber\\
   &=R_{\GB,3}(\psi).
\end{align}
Taking the expectation of Eq.~\eqref{eq:boson-tetrahedral-operator} proves Eq.~\eqref{eq:boson-three-four-comparison}.
\end{proof}

Equations~\eqref{eq:boson-lift-calibration} and \eqref{eq:boson-three-four-comparison} now give the full-Gaussian bounds.

\begin{theorem}[Full Gaussian distance bounds]\label{thm:boson-full}
\label{prop:boson-four-calibration}\label{lem:boson-full-zero}
For every normalized input $\ket\psi$, the four-copy difference-output test and the three-copy test satisfy
\begin{align}
 \ell_{\mathrm B}\bigl(d_{\GB}(\psi)^2\bigr)
   &\le R_{\mathrm{diff},4}(\psi)
      \le\mathcal U_4\bigl(d_{\GB}(\psi)^2\bigr),
 \label{eq:boson-four-calibration}\\
 \tfrac49\ell_{\mathrm B}\bigl(d_{\GB}(\psi)^2\bigr)
   &\le\tfrac49R_{\mathrm{diff},4}(\psi)
      \le R_{\GB,3}(\psi)\le\mathcal U_3\bigl(d_{\GB}(\psi)^2\bigr).
 \label{eq:boson-three-calibration}
\end{align}
In particular,
\begin{align}
 R_{\mathrm{diff},4}(\psi)&\ge\frac3{32}d_{\GB}(\psi)^2, \nonumber\\
 R_{\GB,3}(\psi)&\ge\frac1{24}d_{\GB}(\psi)^2.
\end{align}
Both tests have zero rejection exactly on $\GB$.
\end{theorem}

\begin{proof}
The coherent bound in Eq.~\eqref{eq:boson-coherent-calibration} and Lemma~\ref{lem:boson-quotients} give
\begin{align}
 \nonumber
 d_{\mathcal Z}(\Xi_\psi)^2
 &=\inf_U R_{\Coh,2}(U\psi)\\
 \nonumber
 &\ge\inf_U\ell_{1/4}\bigl(d_{\Coh}(U\psi)^2\bigr)\\
  \nonumber
 &=\ell_{1/4}\bigl(d_{\GB}(\psi)^2\bigr).
 \label{eq:boson-output-distance-bound}
\end{align}
The last equality uses Eq.~\eqref{eq:boson-gaussian-coherent-distance} and continuity and monotonicity of $\ell_{1/4}$, so it holds even if the infimum over $U$ is not attained. Equations~\eqref{eq:boson-four-rejection} and \eqref{eq:boson-lift-calibration} then give
\begin{align}
 R_{\mathrm{diff},4}(\psi)
 &\ge\ell_{1/8}\bigl(d_{\mathcal Z}(\Xi_\psi)^2\bigr) \nonumber\\
 &\ge\ell_{1/8}\!\left(\ell_{1/4}\bigl(d_{\GB}(\psi)^2\bigr)\right) \nonumber\\
 &=\ell_{\mathrm B}\bigl(d_{\GB}(\psi)^2\bigr).
\end{align}
Equation~\eqref{eq:boson-three-four-comparison} gives the three-copy lower bound.

For perfect completeness, write $\ket\gamma=D(\alpha)U\ket\Omega$ using Lemma~\ref{lem:boson-frames}. The beam splitter sends equal displacements to the sum output, and commutes with $U\otimes U$ by Eq.~\eqref{eq:boson-unitary-copy-commutation}. Thus, up to a global phase,
\begin{align}
 B\ket\gamma^{\otimes2}
   &=D(\sqrt2\alpha)U\ket\Omega\otimes U\ket\Omega.
\end{align}
The centered difference factor is accepted by the four-copy test. Perfect completeness of $P_{\GB,3}$ was established after Eq.~\eqref{eq:boson-three-projector}. Equation~\eqref{eq:universal-upper} therefore gives the upper bounds $\mathcal U_4$ and $\mathcal U_3$.

If either rejection vanishes, Eq.~\eqref{eq:boson-three-four-comparison} gives $R_{\mathrm{diff},4}(\psi)=0$. Lemma~\ref{lem:boson-lift-zero} then gives $\ket{\Xi_\psi}=\ket\chi\otimes\ket g$ with $\ket g\in\Gzero$. Choose a Gaussian unitary without displacement such that $U\ket g=\ket\Omega$. Commutation with $B$ gives
\begin{align}
 B(U\ket\psi)^{\otimes2}
   &=(U\otimes U)\ket{\Xi_\psi} \nonumber\\
   &=U\ket\chi\otimes\ket\Omega.
\end{align}
Lemma~\ref{lem:boson-coherent-zero} implies $U\ket\psi\in\Coh$. Lemma~\ref{lem:boson-frames}(iv) then gives $\ket\psi\in\GB$. The converse is perfect completeness.

Finally, Eq.~\eqref{eq:bosonic-envelope-linear-bound} gives
\begin{align}
 \ell_{\mathrm B}\bigl(d_{\GB}(\psi)^2\bigr)
   &\ge\frac3{32}d_{\GB}(\psi)^2.
\end{align}
Combining it with the two rejection lower bounds proves the stated linear inequalities.
\end{proof}

The coefficient $4/9$ is sharp for pure inputs. This is proved later by the expansions in Eq.~\eqref{eq:normal-three-four-sharpness} of Proposition~\ref{prop:normal-circle-four}, whose ratio tends to $4/9$.

\subsubsection{Two-copy impossibility}
\label{subsec:boson-minimality}

For two copies, preserving the equally weighted sum leaves only the identity or a sign change on the difference output. The sign change corresponds to swapping the inputs and leaves every tensor square unchanged. The following obstruction holds for every two-copy acceptance operator, including non-optical measurements.

\begin{proposition}[Impossibility of a two-copy full-Gaussian test]
\label{prop:boson-two-copy-impossible}
Let $0\le E\le I$ be an acceptance operator on $\mathcal F_n^{\otimes2}$. If
\begin{align}
 \ip{g^{\otimes2}}{Eg^{\otimes2}}&=1
        \qquad\text{for every }\ket g\in\GB,
\end{align}
then $E$ acts as the identity on $\Sym^2(\mathcal F_n)$. It therefore accepts every pure tensor square.
\end{proposition}

\begin{proof}
We prove that Gaussian tensor squares have dense linear span in $\Sym^2(\mathcal F_n)$. After the balanced beam splitter, their Bargmann functions, up to nonzero scalar factors, are
\begin{align}
 F_{Z,b}(z,w)&=\exp\!\left(\tfrac12z^TZz+\tfrac12w^TZw
                          +\sqrt2b^Tz\right),
 \label{eq:boson-dense-family}
\end{align}
where $Z=Z^T$, $\|Z\|_\infty<1$, and $b\in\mathbb C^n$. Near $Z=b=0$, their parameter Taylor series converges in the 2-norm. Indeed, restrict $\|Z\|_\infty\le\rho<1$ and $\|b\|_2\le M$. The squared integrand and every fixed parameter derivative are bounded by a polynomial times
\begin{align}
 \exp\!\left[-(1-\rho)(\|z\|_2^2+\|w\|_2^2)
                  +2\sqrt2M\|z\|_2\right],
\end{align}
which is integrable. Dominated convergence therefore justifies the parameter derivatives in the Fock-space 2-norm.

Each derivative at zero is a 2-norm limit of linear combinations of the functions in Eq.~\eqref{eq:boson-dense-family}, and hence belongs to their closed linear span. Differentiation gives
\begin{align}
 \left.\partial_{b_i}F_{Z,b}\right|_{Z=b=0}&=\sqrt2z_i, \nonumber\\
 \left.\partial_{Z_{i,j}}F_{Z,b}\right|_{Z=b=0}
  &=\begin{cases}z_i z_j+w_iw_j,&i<j, \nonumber\\
                (z_i^2+w_i^2)/2,&i=j.
    \end{cases}
\end{align}
Mixed derivatives give every product of these polynomials, up to a nonzero scalar. Their linear combinations therefore belong to the closed span. Since
\begin{align}
 w_iw_j&=(z_i z_j+w_iw_j)-z_i z_j,
\end{align}
these combinations include every polynomial containing only even total powers of $w$, with arbitrary polynomial dependence on $z$. They are dense in the subspace with even photon number in the difference output. Under $B$, swapping the inputs sends $w$ to $-w$, so this subspace is $B\Sym^2(\mathcal F_n)$. Nonzero normalization factors do not change the span, proving the density claim for normalized Gaussian states.

Positivity and perfect completeness give
\begin{align}
 \|(I-E)^{1/2}\ket g^{\otimes2}\|_2^2
   &=1-\ip{g^{\otimes2}}{Eg^{\otimes2}} \nonumber\\
   &=0.
\end{align}
The bounded operator $(I-E)^{1/2}$ vanishes on their linear span and hence on its closure, $\Sym^2(\mathcal F_n)$. Thus $E$ is the identity on that subspace.
\end{proof}

Together with Theorem~\ref{thm:boson-full}, this proves that three copies are necessary and sufficient for a nontrivial perfectly complete test of all pure bosonic Gaussians. Any procedure using at most two copies, including ancillas, randomness, or adaptation, has an overall two-copy acceptance operator. An unused copy can be adjoined if needed. The conclusion does not extend to tests with completeness error.

\subsection{Rejection coefficients and sharpness}
\label{subsec:boson-rejection-sharpness}

The previous sections establish global distance bounds. We now
determine how sensitively the tests detect small departures from
the target families, including for higher-copy tests. We separate
deviations into photon-number components and identify those with
the smallest rejection at leading order
(Eq.~\eqref{eq:boson-normal-degrees} and
Section~\ref{subsec:boson-higher-response}). Explicit states with
these leading coefficients prove local sharpness and provide
examples for the paper's copy lower bounds. These calculations
also give the coefficients used in the tolerant-testing guarantees.

The analysis starts from a closest target. We therefore first show
that such a target exists for every normalized bosonic input state,
without assuming a bound on its expected photon number.

\begin{lemma}[Existence of closest bosonic targets]\label{lem:normal-attainment}
For a fixed finite number of modes, every normalized input has a closest coherent, centered pure Gaussian, and arbitrary pure Gaussian state, without an expected photon-number assumption.
\end{lemma}
\begin{proof}
For an occupation state $\ket\alpha$, the coefficient $\ip\alpha{g_{Z,b}}$ in Eq.~\eqref{eq:boson-full-form} is $C_{Z,b}$ times a polynomial in $Z,b$, of degree at most $|\alpha|$ in $b$. Since $\|Z\|_\infty<1$, Eq.~\eqref{eq:normal-gaussian-normalizer} gives a constant $C_\alpha$, independent of $Z,b$, such that
\begin{align}
 |\ip\alpha{g_{Z,b}}|
 &\le C_\alpha\det(I-Z^\dagger Z)^{1/4}
       (1+\|b\|_2)^{|\alpha|}\e^{-\|b\|_2^2/4}.
 \label{eq:boson-coefficient-escape}
\end{align}
This bound tends to zero as $\|b\|_2\to\infty$, uniformly in $Z$. It also tends to zero as $\|Z\|_\infty\to1$, uniformly in $b$: the determinant vanishes and the polynomial times the Gaussian is bounded. Thus every fixed occupation coefficient tends to zero when the parameters leave every compact subset of their domain.

The overlap with any fixed input then tends to zero as well. To see this, let $\Pi_{\le N}$ project onto total photon number at most $N$ and write $\ket\psi=\sum_\alpha\psi_\alpha\ket\alpha$. Cauchy--Schwarz gives
\begin{align}
 |\ip\psi{g_{Z,b}}|
 &\le\|(I-\Pi_{\le N})\ket\psi\|_2
    +\sum_{|\alpha|\le N}|\psi_\alpha|\,|\ip\alpha{g_{Z,b}}|.
\end{align}
Choose $N$ to make the first term small. The finite sum then tends to zero by Eq.~\eqref{eq:boson-coefficient-escape}.

Coherent states have dense linear span: a vector orthogonal to all of them has identically zero Bargmann function. Hence the supremum of the coherent or arbitrary-Gaussian overlap with a normalized input is positive. A sequence approaching this supremum therefore has a subsequence whose parameters converge inside the domain. On a compact parameter subset, choose $\|Z\|_\infty\le\rho<1$ and $\|b\|_2\le M$, and take $C_{Z,b}>0$. The squared Bargmann functions, including the integration weight in Eq.~\eqref{eq:boson-bargmann}, are bounded by a constant times
\begin{align}
 \exp\!\left[-(1-\rho)\|z\|_2^2+2M\|z\|_2\right].
\end{align}
This function is integrable, so dominated convergence gives continuity in the 2-norm and an attained maximum. Restricting to $Z=0$ proves the coherent case and to $b=0$ the centered case whenever its supremum is positive. If the centered supremum is zero, every centered Gaussian is a maximizer.
\end{proof}

At vacuum, varying the parameters in Eqs.~\eqref{eq:boson-centered-form}--\eqref{eq:boson-full-form} produces two-photon components to first order for centered Gaussians, one-photon components for coherent states, and both for arbitrary Gaussians. The parameter amplitudes are complex, so each component and its multiple by $i$ are allowed. For $0<d_{\mathcal T}(\psi)<1$, send a closest target to vacuum. Lemma~\ref{lem:normal-stationarity} restricts the deviation $\ket\eta$ in Eq.~\eqref{eq:normal-frame} to
\begin{align}
 \Gzero:&\quad k\ne0,2, \nonumber\\
 \Coh:&\quad k\ge2, \nonumber\\
 \GB:&\quad k\ge3.
\end{align}
Here $k$ is the total photon number.

These transformations preserve the corresponding tests. For centered Gaussians, applying the same Gaussian unitary without displacement to every copy commutes with copy rotations by Eq.~\eqref{eq:boson-unitary-copy-commutation}. For coherent states, equal displacements affect only the unrestricted equally weighted modes in Eq.~\eqref{eq:higher-copy-collective-mode}. Full Gaussians use both properties. Since each projector also accepts $\ket\Omega^{\otimes q}$, this proves perfect completeness of the bosonic projectors in Definition~\ref{def:normal-projectors}.

For a normalized $k$-photon state $\ket{\eta_k}$, let $a_k$ be the acceptance probability of $D(\Omega^{\otimes q})[\eta_k]/\sqrt q$. The calculations below show that $a_k$ depends only on $k$. All the projectors preserve total photon number. Thus, for a normalized deviation $\ket\eta=\sum_k\ket{\eta_k}$, orthogonality of different photon-number components and Eq.~\eqref{eq:tensor-derivative} give
\begin{align}
 \left\|Q\bigl(D(\Omega^{\otimes q})[\eta]\bigr)\right\|_2^2
 &=q\sum_k(1-a_k)\|\ket{\eta_k}\|_2^2,
 \label{eq:boson-normal-degrees}\\
 \sum_k\|\ket{\eta_k}\|_2^2&=1.
\end{align}
Consequently, the smallest coefficient of $t^2$ in Eq.~\eqref{eq:normal-W} is the minimum of $q(1-a_k)$ over the permitted photon numbers.

\subsubsection{Higher-copy tests}
\label{subsec:boson-higher-response}\label{subsec:boson-distance-curves}

We now prove the bosonic coefficients stated in Table~\ref{table:normal-coefficients} and give states attaining them. All examples below use $0\le t\le1$ and one physical mode. For additional modes, taking the overlap with vacuum in those modes gives, by Lemma~\ref{lem:boson-gaussian-coordinates}, a nonzero scalar multiple of a normalized target of the same family. The scalar has magnitude at most one. Thus the one-mode upper bounds on overlap also hold with vacuum in any additional modes.

\smallskip\noindent
\emph{Centered Gaussians.}
For Haar-distributed $SO(q)$ rotations with $q\ge2$, Eq.~\eqref{eq:higher-copy-first-order-acceptance} gives $a_k=q\mathbb E T^k$, where $T$ is one component of a uniformly distributed unit vector in $\mathbb R^q$. Odd moments vanish by symmetry, and Eq.~\eqref{eq:normal-real-moments} gives
\begin{align}
 a_k&=0\qquad(k\text{ odd}), \nonumber\\
 a_4&=q\frac3{q(q+2)}=\frac3{q+2}, \nonumber\\
 a_k&\le a_4\qquad(k\ge4\text{ even}).
\end{align}
Equation~\eqref{eq:boson-normal-degrees} therefore gives
\begin{align}
 \lambda_{\Gzero,q}&=q(1-a_4)=\frac{q(q-1)}{q+2},
\end{align}
attained by a four-photon component.

To verify sharpness as a function of distance, consider
\begin{align}
 \ket{\psi_t}&=\sqrt{1-t^2}\ket0+t\ket4.
\end{align}
By Eq.~\eqref{eq:boson-centered-form}, a one-mode centered Gaussian has vacuum and four-photon coefficients $(1-|Z|^2)^{1/4}$ and $(1-|Z|^2)^{1/4}\sqrt6Z^2/4$, up to a common phase. Concavity gives
\begin{align}
 |\ip g{\psi_t}|
 &\le(1-|Z|^2)^{1/4}
       \left(\sqrt{1-t^2}+\frac{\sqrt6}{4}t|Z|^2\right) \nonumber\\
 &\le\left(1-\frac{|Z|^2}{4}\right)
       \left(\sqrt{1-t^2}+\frac{\sqrt6}{4}t|Z|^2\right) \nonumber\\
 &=\sqrt{1-t^2}
   +\frac{\sqrt6t-\sqrt{1-t^2}}4|Z|^2
   -\frac{\sqrt6t}{16}|Z|^4 \nonumber\\
 &\le\sqrt{1-t^2}
 \qquad\left(t^2\le\frac17\right).
\end{align}
Vacuum attains this overlap, so $d_{\Gzero}(\psi_t)=t$ on this interval. For the centered and vacuum examples in Lemma~\ref{lem:testing-exact-curves}, one may instead use $\sqrt{1-t^2}\ket0+t\ket1$ for any $0\le t\le1$. Centered Gaussians have even photon number, so
\begin{align}
 \sup_{\ket g\in\Gzero}
 \left|\sqrt{1-t^2}\,\ip g0+t\,\ip g1\right|
 &=\sqrt{1-t^2}, \nonumber\\
 \left|\sqrt{1-t^2}\,\ip\Omega0+t\,\ip\Omega1\right|
 &=\sqrt{1-t^2}.
\end{align}

\smallskip\noindent
\emph{Coherent states.}
For $q\ge2$, projecting a $k$-photon state in one copy onto the equally weighted modes has amplitude $q^{-k/2}$. Each of the $q$ terms in $D(\Omega^{\otimes q})[\eta_k]$ projects onto the same state with this amplitude. Hence
\begin{align}
 a_k&=\frac1q\left(q\,q^{-k/2}\right)^2=q^{1-k} \nonumber\\
    &\le\frac1q\qquad(k\ge2), \nonumber\\
 \lambda_{\Coh,q}&=q\left(1-\frac1q\right)=q-1.
\end{align}
The minimum occurs for two photons. Consider
\begin{align}
 \ket{\psi_t}&=\sqrt{1-t^2}\ket0+t\ket2.
\end{align}
For a coherent state $\ket{c_\alpha}$, put $y=|\alpha|^2$. Equation~\eqref{eq:boson-coherent-form} gives
\begin{align}
 |\ip{c_\alpha}{\psi_t}|
 &\le\e^{-y/2}\left(\sqrt{1-t^2}+\frac{ty}{\sqrt2}\right), \nonumber\\
 \frac{\dd}{\dd y}
 \left[\e^{-y/2}\left(\sqrt{1-t^2}+\frac{ty}{\sqrt2}\right)\right]
 &=\e^{-y/2}\left(\frac t{\sqrt2}-\frac{\sqrt{1-t^2}}2
                         -\frac{ty}{2\sqrt2}\right) \nonumber\\
 &\le0\qquad\left(t^2\le\frac13,\quad y\ge0\right).
\end{align}
The upper bound is therefore maximized at $y=0$, where vacuum attains it. Thus $d_{\Coh}(\psi_t)=t$ for $t^2\le1/3$.

\smallskip\noindent
\emph{Arbitrary Gaussians.}
Take $q\ge3$. A Haar-distributed copy rotation fixing $\mathbf u$ applies a uniformly distributed $SO(q-1)$ rotation to the combinations orthogonal to $\mathbf u$. Each standard basis vector has component $1/\sqrt q$ along $\mathbf u$, and its remaining component has 2-norm $\sqrt{(q-1)/q}$. Thus every matrix entry has the distribution
\begin{align}
 X&=\frac1q+\frac{q-1}{q}\,T,
 \label{eq:normal-full-entry}
\end{align}
where $T$ is one component of a uniformly distributed unit vector in $\mathbb R^{q-1}$. Equation~\eqref{eq:higher-copy-first-order-acceptance} gives $a_k=q\mathbb E X^k$. Symmetry and the moments in Eq.~\eqref{eq:normal-real-moments}, with $q-1$ in place of $q$, give
\begin{align}
 \mathbb E X^3
 &=\frac1{q^3}+\frac{3(q-1)}{q^3}
   =\frac{3q-2}{q^3},
 \label{eq:normal-full-third}\\
 \mathbb E X^4
 &=\frac1{q^4}+\frac{6(q-1)}{q^4}
             +\frac{3(q-1)^3}{q^4(q+1)} \\
 &=\frac{3q^3-3q^2+10q-8}{q^4(q+1)}, \nonumber\\
 \mathbb E X^3-\mathbb E X^4
 &=\frac{4(q-1)(q-2)}{q^4(q+1)}>0.
\end{align}
The second moment used here is $\mathbb E T^2=1/(q-1)$, since the $q-1$ components have equal second moments and squared sum one. Because $|X|\le1$,
\begin{align}
 \mathbb E X^k&\le\mathbb E|X|^k \nonumber\\
             &\le\mathbb E X^4<\mathbb E X^3\qquad(k\ge4).
\end{align}
Thus the smallest rejection coefficient occurs for three photons
\begin{align}
 \lambda_{\GB,q}
 &=q(1-q\mathbb E X^3)=\frac{(q-1)(q-2)}q.
 \label{eq:boson-full-rejection-coefficient}
\end{align}

For sharpness as a function of distance, take
\begin{align}
 \ket{\psi_t}&=\sqrt{1-t^2}\ket0+t\ket3.
\end{align}
Write a normalized one-mode Gaussian as $g(z)=C_{Z,b}\exp(Zz^2/2+bz)$, as in Eq.~\eqref{eq:boson-full-form}, and put $y=|Z|^2+|b|^2$. Equation~\eqref{eq:normal-gaussian-normalizer} gives
\begin{align}
 |C_{Z,b}|&\le(1-|Z|^2)^{1/4}\e^{-|b|^2/4} \nonumber\\
          &\le\e^{-y/4}.
 \label{eq:testing-gaussian-C}
\end{align}
The three-photon coefficient is $C_{Z,b}(b^3+3Zb)/\sqrt6$, so
\begin{align}
 |\ip3g|
 &\le\frac{|C_{Z,b}|}{\sqrt6}\left(|b|^3+3|Z|\,|b|\right) \nonumber\\
 &\le\frac{|C_{Z,b}|}{\sqrt6}\left(y^{3/2}+\frac32y\right).
\end{align}
For $y>0$,
\begin{align}
 \frac{y^{3/2}+3y/2}{\e^{y/4}-1}
 &\le\frac{32(\sqrt y+3/2)}{y+8} \nonumber\\
 &\le10.
\end{align}
The first inequality uses $\e^{y/4}-1\ge y/4+y^2/32$. The second follows from
\begin{align}
 10(y+8)-32(\sqrt y+3/2)
 &=10\left(\sqrt y-\frac85\right)^2+\frac{32}5>0.
\end{align}
Together with $|C_{Z,b}|\e^{y/4}\le1$, these estimates imply
\begin{align}
 |\ip3g|&\le\frac{10}{\sqrt6}(1-|C_{Z,b}|), \nonumber\\
 |\ip g{\psi_t}|
 &\le\sqrt{1-t^2}|C_{Z,b}|+\frac{10t}{\sqrt6}(1-|C_{Z,b}|) \nonumber\\
 &\le\sqrt{1-t^2}\qquad\left(t^2\le\frac3{53}\right).
\end{align}
The case $y=0$ is vacuum and satisfies the same bound. Vacuum attains this overlap, so $d_{\GB}(\psi_t)=t$ on this interval, which contains $0\le t\le1/5$.

The four-photon, two-photon, and three-photon examples attain the respective minimum coefficients. Their exact distances and the tensor expansion in Eq.~\eqref{eq:normal-W} prove the sharp leading coefficients in Corollary~\ref{cor:normal-distance-sharpness}.

\subsubsection{Rotations of the four-copy difference outputs}

The full four-copy projector averages all rotations fixing $\mathbf u$. The four-copy test in Theorem~\ref{thm:boson-full} averages only rotations in the plane in Eq.~\eqref{eq:boson-difference-plane}. We now show that this test has the same smallest quadratic rejection coefficient, giving an improved distance bound near the Gaussian family.

\begin{proposition}[Four-copy difference-output bound]
\label{prop:normal-circle-four}
Every normalized input satisfies
\begin{align}
 R_{\mathrm{diff},4}(\psi)
 &\ge\max\left\{\ell_{\mathrm B}\bigl(d_{\GB}(\psi)^2\bigr),
                  h_{4,3/2}\bigl(d_{\GB}(\psi)^2\bigr)\right\}.
 \label{eq:normal-circle-four}
\end{align}
The smallest coefficient of $t^2$ in Eq.~\eqref{eq:normal-W}, over all normalized deviations from a closest full Gaussian, is $3/2$. For $\ket{\psi_t}=\sqrt{1-t^2}\ket0+t\ket3$,
\begin{subequations}\label{eq:normal-three-four-sharpness}
\begin{align}
 R_{\mathrm{diff},4}(\psi_t)&=\frac32t^2+o(t^2), \nonumber\\
 R_{\GB,3}(\psi_t)&=\frac23t^2+o(t^2).
 \label{eq:boson-three-copy-sharpness}
\end{align}
\end{subequations}
\end{proposition}

\begin{proof}
Send a closest full Gaussian to vacuum using the test-preserving transformations described above. The deviation then contains only photon numbers $k\ge3$. Each row of a rotation in the plane in Eq.~\eqref{eq:boson-difference-plane} contains, up to permutation,
\begin{align}
 A&=\frac{1+\cos\theta}{2},\qquad
 B=\frac{1-\cos\theta}{2}, \nonumber\\
 C&=\frac{\sin\theta}{2},\qquad -C.
\end{align}
Equation~\eqref{eq:higher-copy-first-order-acceptance} gives
\begin{align}
 a_k&=\mathbb E_\theta\bigl[A^k+B^k+C^k+(-C)^k\bigr],
 \qquad\theta\text{ uniform on }[0,2\pi].
\end{align}
For odd $k\ge3$, the last two terms cancel. Since $A,B\in[0,1]$,
\begin{align}
 a_k&\le a_3 \nonumber\\
    &=\mathbb E_\theta\frac{1+3\cos^2\theta}{4}=\frac58.
\end{align}
For even $k\ge4$, each nonnegative term decreases with $k$, giving
\begin{align}
 a_k&\le a_4 \nonumber\\
 &=\mathbb E_\theta\left[
   \frac{1+6\cos^2\theta+\cos^4\theta}{8}
                    +\frac{\sin^4\theta}{8}\right] \nonumber\\
 &=\frac18\left(1+3+\frac38+\frac38\right) \nonumber\\
 &=\frac{19}{32}<\frac58.
\end{align}
Equation~\eqref{eq:boson-normal-degrees} therefore gives a minimum coefficient $4(1-5/8)=3/2$, attained by a three-photon component. Lemma~\ref{lem:normal-bound} gives the lower bound $h_{4,3/2}(d_{\GB}(\psi)^2)$, and combining it with Eq.~\eqref{eq:boson-four-calibration} proves Eq.~\eqref{eq:normal-circle-four}.

For an accepted target $\ket g$ and a fixed normalized deviation $\ket\eta$ with $\ip g\eta=0$, Eq.~\eqref{eq:tensor-derivative} gives
\begin{align}
 Q\bigl(\sqrt{1-t^2}\ket g+t\ket\eta\bigr)^{\otimes q}
 &=tQ\left(\sum_{a=1}^q\ket g^{\otimes(a-1)}\otimes\ket\eta
                         \otimes\ket g^{\otimes(q-a)}\right)+O_q(t^2), \nonumber\\
 R\bigl(\sqrt{1-t^2}\ket g+t\ket\eta\bigr)
 &=t^2\left\|Q\bigl(D(g^{\otimes q})[\eta]\bigr)\right\|_2^2+o(t^2).
\end{align}
The first remainder is in the 2-norm. The three-photon coefficient just computed gives the first expansion in Eq.~\eqref{eq:normal-three-four-sharpness}. Equation~\eqref{eq:boson-full-rejection-coefficient} at $q=3$ gives Eq.~\eqref{eq:boson-three-copy-sharpness}.
\end{proof}

The composition bound in Eq.~\eqref{eq:ellB-definition} has leading coefficient $1/2$, while the same four-copy experiment has sharp local coefficient $3/2$. Equation~\eqref{eq:normal-circle-four} combines that global bound with the stronger bound near the target family.

\subsection{Optical measurements}\label{subsec:boson-measurements}

We now turn to discussing optical measurements. For centered bosons, the two-copy rotation generator is, up to sign,
\begin{equation}
 G=i\sum_{j=1}^n(a_jb_j^\dagger-a_j^\dagger b_j),
\end{equation}
where $a_j,b_j$ annihilate mode $j$ in the two copies. Its spectral projection at zero is the bounded group-average projector $P_{\Gzero,2}$. A phase shift and a balanced beam splitter diagonalize the two-dimensional single-particle rotation matrix. If $V$ is the resulting passive unitary and $N_\pm$ are the two total output photon-number operators, then, with a choice of output ordering,
\begin{align}
 VGV^\dagger&=N_+-N_-,
 \label{eq:boson-optical-generator}\\
 VP_{\Gzero,2}V^\dagger&=\mathbf1_{\{0\}}(N_+-N_-).
 \label{eq:boson-optical-projector}
\end{align}
Here $\mathbf1_{\{0\}}(N_+-N_-)$ denotes the orthogonal projector onto the zero-eigenvalue subspace of $N_+-N_-$. Photon counting implements this projector by accepting exactly when the total photon counts in the two output groups are equal.

The full bosonic three-copy measurement has the same form after a three-port change of copy coordinates. The equally weighted combination of the three copies, with coefficient vector $(1,1,1)/\sqrt3$, is left unrestricted. On the two orthogonal combinations the rotation generator has eigenvalues $+1,-1$ in a suitable complex basis. For the corresponding passive unitary $V_3$,
\begin{equation}
 V_3P_{\GB,3}V_3^\dagger
    =I_{\rm collective}\otimes\mathbf1_{\{0\}}(N_{\rm rel,+}-N_{\rm rel,-}).
\end{equation}
Acceptance is equality of the two total relative-output photon counts.

The four-copy difference-port test first applies balanced beam splitters to pairs $(1,2)$ and $(3,4)$. It then applies the centered rotation test to the two difference outputs and discards the sum outputs. This realizes the projector used in the composition proof. The full four-copy $SO(3)$ invariant projector is a different, stronger measurement.

Coherent-state testing uses a balanced beam splitter and accepts precisely when the entire difference output is vacuum. Vacuum testing directly measures its one-copy vacuum projector. These two tests need only distinguish vacuum from a nonzero photon count. They do not need to resolve the magnitude of that count.

For general $q$, the hierarchy specifies joint measurements by bounded orthogonal projectors onto invariant spaces, or onto relative vacuum for coherent states. Their existence gives an information-theoretic copy bound. No general efficient circuit construction for every higher-copy projector is required or proved here.

\end{document}